\documentclass[3p,times]{elsarticle}

\usepackage{amsmath}
\usepackage{hyperref}

\usepackage{bm,amsthm,amsmath}
\usepackage{graphicx,subfigure}
\usepackage{makecell}
\usepackage{multirow}
\usepackage{enumitem}
\usepackage{booktabs}
\usepackage{epstopdf}
\usepackage{color}
\usepackage{threeparttable}
\usepackage{appendix}

\newtheorem{lemma}{Lemma}[section]
\newtheorem{theorem}[lemma]{Theorem}

\newtheorem{remark}[lemma]{Remark}
\newcommand{\pref}[1]{{\rm(\ref{#1})}}

\begin{document}
\baselineskip11pt

\begin{frontmatter}

\title{Fairing-PIA: Progressive iterative approximation for fairing curve and surface generation}

\author[label1]{Yini Jiang}
\author[label1,label2]{Hongwei Lin\corref{cor1}}
\ead{hwlin@zju.edu.cn}
\cortext[cor1]{Corresponding author at: School of Mathematical Science, Zhejiang University, Hangzhou, 310027, China.}
\address[label1]{School of Mathematical Science, Zhejiang University, Hangzhou, 310027, China}
\address[label2]{State Key Laboratory of CAD\&CG, Zhejiang University, Hangzhou, 310058, China}

\begin{abstract}
The fairing curves and surfaces are used extensively in geometric design, 
    modeling, and industrial manufacturing. 
However, the majority of conventional fairing approaches, 
    which lack sufficient parameters to improve fairness, 
    are based on energy minimization problems. 
In this study, we develop a novel progressive-iterative approximation method for fairing curve and surface generation~(fairing-PIA). 
Fairing-PIA is an iteration method that can generate a series of curves~(surfaces) 
    by adjusting the control points of B-spline curves~(surfaces). 
In fairing-PIA, each control point is endowed with an individual weight. 
Thus, the fairing-PIA has many parameters to optimize the shapes of curves and surfaces. 
Not only a fairing curve~(surface) can be generated 
    globally through fairing-PIA, 
    but also the curve~(surface) can be improved locally. 
Moreover, we prove the convergence of the developed fairing-PIA 
    and show that the conventional energy minimization fairing model is a special case of fairing-PIA. 
Finally, numerical examples indicate that 
    the proposed method is effective and efficient.
\end{abstract}


\end{frontmatter}


\section{Introduction}

The generation of fairing curves and surfaces from scattered data points 
    is a fundamental problem in many fields.
For example, in the area of geometric design, 
    generating smooth airfoils~\cite[]{Li2005}, 
    ship hulls forms~\cite[]{SARIOZ20062105}, 
    and car hoods~\cite[]{WESTGAARD2001619,Hashemian2018} is often necessary 
    because these geometric properties can further affect their performance characteristics. 
Recently, many scholars have been interested in the applications of fairing curves 
    in robot path planning~\cite[]{Gili2011,Song2021}; 
    thus, designing effective and flexible fairing algorithms has become an urgent concern.

Traditionally, two kinds of fairing methods exist: 
    global~\cite[]{Hagen1991,WANG2010} and local methods~\cite[]{Kjellander1983,FARIN198791,Wang2015}. 
In both fairing methods, 
    the fairing problem is usually modeled as an energy optimization problem 
    or constrained energy optimization problem. 
However, the energy optimization model lacks sufficient parameters 
    to adjust and improve the fairness of curves and surfaces finely. 
The capability of fine adjustment is a key ingredient in generating fairing curves and surfaces.

In this study, 
    we develop a novel progressive-iterative approximation method 
    for fairing curve and surface generation~(fairing-PIA) 
    to overcome the deficiency of the traditional fairing methods. 
A series of curves~(surfaces) can be generated through fairing-PIA 
    by updating the control points iteratively.
In each iteration, 
    we first construct the \emph{fitting vector} for a control point 
    and the \emph{fairing vector} for a control point.     
Then, the weighted sum of the two vectors is taken as the \emph{difference vector} for each control point. 
    Finally, the new control points are generated by adding the difference vectors to the old ones.
In this way, each control point is endowed with an individual weight. 
In total, considerable weights exist. 
These weights can be utilized to adjust 
    and improve the fairness of curves and surfaces finely. 
Moreover, we prove the convergence of fairing-PIA. 
We also show that the traditional energy optimization fairing model is a special case of fairing-PIA. 
In conclusion, the following are the main contributions of this study:
 \begin{itemize}
  \item We develop the fairing-PIA for fairing curve and surface generation. 
        In this method, many parameters can be utilized to improve the fairness finely.
  \item Fairing-PIA is both a global and local method; 
        thus, it can generate a global fairing result and adjust the curves~(surfaces) locally.
  \item Fairing-PIA is a flexible method 
        that allows users to change the weights for control points, 
        the knot vectors, or data parametrization in each iteration.
 \end{itemize}
In this paper, we present the fairing-PIA method with B-spline curves and surfaces.
However, fairing-PIA is valid for other blending curves and surfaces,
    including B\'{e}zier curves and surfaces, 
    NURBS curves and surface, T-spline surfaces.

The paper is organized as follows: 
    In Section 1.1, we present a brief overview of the related work. 
Section 2 contains the details of our algorithm for curve fairing. 
    Convergence is discussed in Section 3. 
Then, we propose a fairing algorithm for surfaces in Section 4. 
    Moreover, numerical examples are presented in Section 5. 
Finally, this paper is concluded in Section 6.

\subsection{Related work}

In this section, 
    we briefly review the related work on curve and surface fairing 
    and progressive-iterative approximation (PIA).

\textbf{Curve and surface fairing}:
The basic approaches to designing the fairing curve include poor point removal,
    energy minimization, and a combination of both methods.
Kjellander et al.~\cite[]{Kjellander1983} proposed an interactive fairing algorithm 
    to improve the smoothness of the curve by modifying the point where the offset vector is greater than the threshold value.
Farin et al.~\cite[]{FARIN198791} performed the knot removal and reinsertion tricks 
    for the first time to increase the efficiency of replacing points. 
However, such a local fairing method is intrinsically inefficient 
    when numerous abnormal points are encountered.

Energy minimization methods are proposed to overcome this limitation. 
The basic idea is to find the smoothest curve in the sense of minimal energy functional. 
The commonly used energy includes strain energy~\cite[]{Hagen1991}, 
    tension energy~\cite[]{Pottmann1990}, 
    Holladay functional~\cite[]{Campbell1957}, 
    jerk energy~\cite[]{Meier1987}, 
    and the variation of radius of curvature~\cite[]{10.5555/917424}.
The energy methods can perform global optimization for all points; 
    however, the calculation is time-consuming, 
    and the fitting effect cannot be controlled well. 
Hence, another report focuses on integrating the aforementioned two approaches.

Vassilev~\cite[]{Vassilev1996} developed an autonomous fairing technique based on energy minimization and point insertion.
Unlike the previous methods, 
    the technique developed only adds new points when necessary. 
Thus, the number of control points to be solved is greatly decreased. 
Inspired by this work, 
    researchers have presented numerous approaches~\cite[]{Zhang2001,Cali2010,LI2004499} to identify bad points efficiently 
    and update the control points derived from the energy minimization method. 
Wavelet and multiresolution decomposition techniques have been recently employed to 
    fair non-uniform rational B-splines (NURBS)~\cite[]{WANG2010,Aimin2011,Aimin2012}.
For more details on curve and surface fairing,
    please refer to Ref.~\cite[]{Josef1993}.

\textbf{PIA}:
Inspired by Qi and de Boor's work~\cite[]{Qi75,Deboor1979},
    Lin et al.~\cite[]{Lin04} initiated the study of PIA.
Compared with classical methods of curve and surface fitting, 
    PIA avoids the intrinsic lack of direct solving equations, 
    such as high computational cost, 
    by generating a series of fitting curves and surfaces after each iteration 
    to approximate the curve and surface, 
    which interpolate the given data points. 
Then, Lin et al.~\cite[]{LIN2005575} proved that blending curves and surfaces with normalized totally positive bases achieves convergence. 
Shi et al.~\cite[]{Shi06} popularized the result to NURBS curves and surfaces.

Additionally, several accelerating algorithms have been developed. 
Lu~\cite[]{LU2010} proposed the weighted PIA, 
    which can accelerate the convergence rate by introducing optimal weights for difference vectors. 
Wang~\cite[]{Wang2018} developed the Gauss-Seidel PIA to accelerate the traditional PIA method. 
Moreover, Deng and Lin~\cite[]{DENG201432} proposed the least square progressive iterative approximation (LSPIA) 
    and analyzed its convergence. 
In that method, the number of control points is allowed to be less than that of the data points. 
This approach is particularly suitable for handling large data sets.
Since then, numerous interesting applications have been found in diverse fields, 
    including subdivision surfaces~\cite[]{Chen2008,Deng2012,Hamza2021}, 
    implicit curve and surface reconstruction~\cite[]{Hamza2020}, 
    and single image super-resolution~\cite[]{Zhang2019}. 
For more details we refer the readers to~\cite[]{linreview2018}.

\section{Progressive-iterative approximate for curve fairing}
\label{section:curve_fairing}
In this section, we present the fairing-PIA method for B-spline curve fairing.
Given an ordered data set $\{\bm{Q}_i\}_{i=1}^m$,
    with parameters $\left\{t_i\right\}_{i=1}^m$ satisfying $t_1 \leq t_2 \leq \cdots \leq t_m$,
    the initial B-spline curve can be constructed as follows:
\begin{equation}
	\bm{P}^{[0]}(t)=\sum_{j=1}^n N_j(t)\bm{P}_j^{[0]},\quad t\in[t_1,t_m],
	\label{eq:init_curve}
\end{equation}
    where the initial control points $\{\bm{P}_j^{[0]}\}_{j=1}^n$ can be randomly selected,
    and $N_j(t)$ is the B-spline basis function corresponding to the $j$-th control point.
Suppose that the B-spline curve~\pref{eq:init_curve} is defined on the knot vector $\{\bar{t}_1,\bar{t}_2,\cdots,\bar{t}_{n+p+1}\}$,
    where $p$ is 
    \newpage \noindent
    the degree of the B-spline basis function, and $n$ is the number of control points.
In this case, 
    we define the \emph{fairing functional} as follows to generate the fairing curve:
\begin{equation}
	\mathcal{F}_{r,j}(f) = \int_{t_1}^{t_m}N_{r,j}(t)fdt,\quad j=1,2,\cdots,n,\;r=1,2,3,
	\label{eq:def_of_functional}
\end{equation}
    where $N_{r,j}(t)$ represents the $r$-th derivative of $N_j(t)$.
Given the local support property of B-spline functions, 
    the given data points $\{\bm{Q}_i\}_{i=1}^m$ 
    whose parameters fall in the support domain of the $j$-th basis function,
    i.e., $N_j(t) \ne 0$, are classified into the $j$-th group,
    corresponding to the $j$-th control point.
The index set of the data points in the $j$-th group is denoted as $I_j$.

First, we calculate the \emph{difference vectors for data points} in the first step of the iteration, i.e.,
\begin{equation*}
	\bm{d}_i^{[0]} = \bm{Q}_i - \bm{P}^{[0]}(t_i),\quad i=1,2,\cdots,m.
\end{equation*}
The \emph{fitting vectors for control points} are generated by taking a weighted sum of the difference vectors for the $j$-th group data points:
\begin{equation*}
	\bm{\delta}_j^{[0]} = \sum_{h\in I_j}N_j(t_h)\bm{d}_h^{[0]},\quad j=1,2,\cdots,n.
\end{equation*}
Next, the \emph{fairing vectors for control points} of the first iteration can be calculated as
\begin{equation*}
	\bm{\eta}_{j}^{[0]} =\sum_{l=1}^n \mathcal{F}_{r,l}\left(N_{r,j}(t)\right)\bm{P}_l^{[0]},\ j=1,2,\cdots,n.
\end{equation*}
Then, the new control points $\left\{\bm{P}_j^{[1]}\right\}_{j=1}^n$ after the first iteration comprise the contribution of three terms:
\begin{equation*}
	\bm{P}_j^{[1]} = \bm{P}_j^{[0]} + \mu_j
	\left[
	\left(1-\omega_j\right)\bm{\delta}_j^{[0]} - \omega_j\bm{\eta}_{j}^{[0]}\right],
\end{equation*}
    where $\mu_j$ is a \emph{normalization weight},
    and $\omega_j$ is a \emph{smoothing weight} corresponding to the $j$-th control point.
The larger the smoothing weight $\omega_j$ is,
    the smoother the generated curve is.
The new curve is obtained as follows by substituting the new control points into the curve~\eqref{eq:init_curve}:
\begin{equation*}
	\bm{P}^{[1]}(t)=\sum_{j=1}^nN_j(t)\bm{P}_j^{[1]},\quad t\in[t_1,t_m].
\end{equation*}

This procedure is performed iteratively.
After the $k$-th iteration,
   the $k$-th curve $\bm{P}^{[k]}(t)$ is generated:
\begin{equation*}
	\bm{P}^{[k]}(t)=\sum_{j=1}^nN_j(t)\bm{P}_j^{[k]},\quad t\in[t_1,t_m].
\end{equation*}
To construct the $(k+1)$-st curve $\bm{P}^{[k+1]}(t)$,
   we first calculate the $(k+1)$-st difference vectors for data points,
\begin{equation*}
	\bm{d}_i^{[k]} = \bm{Q}_i - \bm{P}^{[k]}(t_i),\quad i=1,2,\cdots,m.
	\label{eq:kst_data_diffvect}
\end{equation*}
Then, the fitting vectors for control points and the fairing vectors for control points are calculated by
\begin{equation*}
	\bm{\delta}_j^{[k]} = \sum_{h\in I_j}N_j(t_h)\bm{d}_h^{[k]},\quad j=1,2,\cdots,n,\ \text{and},
	\label{eq:kst_ctrl_diffvect}
\end{equation*}
\begin{equation*}
	\bm{\eta}_{j}^{[k]} =\sum_{l=1}^n \mathcal{F}_{r,l}\left(N_{r,j}(t)\right)\bm{P}_l^{[k]},\quad j=1,2,\cdots,n.
	\label{eq:kst_fairvect}
\end{equation*}
Finally, the control points of the $(k+1)$-st curve are produced by
\begin{equation}
	\bm{P}_j^{[k+1]} = \bm{P}_j^{[k]} + \mu_j
	\left[
	\left(1-\omega_j\right)\bm{\delta}_j^{[k]} - \omega_j\bm{\eta}_{j}^{[k]}
	\right],\ j=1,2,\cdots,n,
	\label{eq:kst_newctrl}
\end{equation}
   leading to the $(k+1)$-st B-spline curve,
\begin{equation*}
	\bm{P}^{[k+1]}(t)=\sum_{j=1}^nN_j(t)\bm{P}_j^{[k+1]},\quad t\in[t_1,t_m].
\end{equation*}
In this way, we obtain a sequence of curves $\left\{\bm{P}^{[k]}(t),\;k=1,2,3,\cdots\right\}$.
In the following section,
    we prove that the sequence of curves converges.
In the fairing-PIA stated above (Eq.~\eqref{eq:kst_newctrl}),
    $\mu_j$ is a normalization weight that guarantees the convergence of the algorithm,
    and $\omega_j$ is a smoothing weight that balances the fitting vector and fairing vector.
Fairing-PIA is much more flexible for adjusting the fairness than the traditional fairing methods 
    because each control point has an individual smoothing weight.

\section{Convergence analysis}
\label{section:proof_convergence}

In this section, we show that the fairing-PIA for
    curve fairing~\pref{eq:kst_newctrl} is convergent,
    and the traditional energy minimization fairing model is a special case of fairing-PIA.
It is rewritten in the following matrix form
    to show the convergence of the fairing-PIA method~\eqref{eq:kst_newctrl}:
\begin{equation}
	\begin{aligned}
		\bm{P}^{[k+1]} & = \bm{P}^{[k]} + \bm{\Lambda}\left[ \left(\bm{I}-\bm{\Omega}\right)\bm{N}^T\left(\bm{Q}-\bm{NP}^{[k]}\right)-
		\bm{\Omega}\bm{D}_{r}P^{[k]}\right]                                                                                          \\
		& = \left(\bm{I}-\bm{\Lambda A}\right)\bm{P}^{[k]} +
		\bm{\Lambda}\left(\bm{I}-\bm{\Omega}\right)\bm{N}^T\bm{Q},
	\end{aligned}
	\label{eq:kst_newctrl_mtx}
\end{equation}
    in which
\begin{align}
	\bm{P}^{[k]} & = \left[\bm{P}_1^{[k]},\bm{P}_2^{[k]},\cdots,\bm{P}_n^{[k]}\right]^{T}, \;
	\bm{Q}^{[k]}= \left[\bm{Q}_1^{[k]},\bm{Q}_2^{[k]},\cdots,\bm{Q}_m^{[k]}\right]^{T},
	\notag
	\\
	\bm{\Omega}  & = diag \left(\omega_1,\omega_2,\cdots,\omega_n\right),\;
	\bm{\Lambda} = diag \left(\mu_1,\mu_2,\cdots,\mu_n\right),
	\notag
	\\
	\bm{A}       & =\left(\bm{I}-\bm{\Omega}\right)\bm{N}^T\bm{N}+\bm{\Omega}\bm{D}_{r},
	\label{eq:mtx_A}
\end{align}
    and $\bm{N}$ and $\bm{D}_{r}$ denote the collocation matrix and Gram matrix, respectively,
\begin{align}
	\nonumber
	\bm{N}       & = \begin{bmatrix}
		N_1(t_1) & N_2(t_1) & \cdots & N_n(t_1) \\
		N_1(t_2) & N_2(t_2) & \cdots & N_n(t_2) \\
		\vdots   & \vdots   & \ddots & \vdots   \\
		N_1(t_m) & N_2(t_m) & \cdots & N_n(t_m) \\
	\end{bmatrix},                                                                                                                       \\
	\nonumber
  \bm{D}_{r} & =\begin{bmatrix}
    \mathcal{F}_{r,1}\left(N_{r,1}(t)\right) & \mathcal{F}_{r,1}\left(N_{r,2}(t)\right) & \cdots & \mathcal{F}_{r,1}\left(N_{r,n}(t)\right) \\
    \mathcal{F}_{r,2}\left(N_{r,1}(t)\right) & \mathcal{F}_{r,2}\left(N_{r,2}(t)\right) & \cdots & \mathcal{F}_{r,2}\left(N_{r,n}(t)\right) \\
        \vdots                                      & \vdots                                      & \ddots & \vdots                         \\
    \mathcal{F}_{r,n}\left(N_{r,1}(t)\right) & \mathcal{F}_{r,n}\left(N_{r,2}(t)\right) & \cdots & \mathcal{F}_{r,n}\left(N_{r,n}(t)\right)\end{bmatrix}.
\end{align}
Moreover, $diag(\cdot)$ denotes the diagonal matrix, 
    and $r$ can be chosen from the set $\left\{1,2,3\right\}$.

\begin{lemma}
	Let $\bm{A}=[a_{ij}]$ \eqref{eq:mtx_A} be a strict diagonally dominant matrix.
	If the diagonal elements of matrix $\bm{\Lambda}$ satisfy $0<\mu_{i}<\frac{2}{\sum_{j=1}^{n}|a_{ij}|},\;i=1,2,\cdots,n$,
	    then we obtain $0<\left\| \bm{I-\Lambda A} \right\|_{\infty}<1$.
	\label{lemma:norm}
\end{lemma}

\begin{proof}
	Given that $\bm{A}$ is a strict diagonally dominant matrix with nonnegative diagonal entries,
	    we obtain	$\sum_{j=1}^{n}\left|a_{ij}\right|<2a_{ii}$.
  Consequently, 
	\begin{equation}
		0<\frac{1}{a_{ii}}<\frac{2}{\sum_{j=1}^{n}\left|a_{ij}\right|},\;i=1,2,\cdots,n.
		\label{eq:diag_dominant}
	\end{equation}
	On the other hand, 
      according to the definition of $\infty$-norm,
      we obtain
	\begin{equation}
		\left\| \bm{I-\Lambda A} \right\|_{\infty}
		= \underset{i}{\text{max}}\left(|1-\mu_ia_{ii}|+\sum_{j=1,\;j\ne i}^{n}\mu_i|a_{ij}|\right).
		\label{eq:nrm_def}
	\end{equation}
	We discuss below the range of $\left\| \bm{I-\Lambda A} \right\|_{\infty}$ under three cases:
	\begin{enumerate}
		\item[(i)] If $\mu_i = \frac{1}{a_{ii}},\;i=1,2,\cdots,n$,
		then Eq. \eqref{eq:nrm_def} can be simplified as
		\begin{equation*}
			\left\| \bm{I-\Lambda A} \right\|_{\infty}
			= \underset{i}{\text{max}}\left(\frac{\sum_{j=1,\;j\ne i}^{n}|a_{ij}|}{a_{ii}}\right)<1.
			\label{eq:lemma_case1}
		\end{equation*}
		
		\item[(ii)] If $0<\mu_i<\frac{1}{a_{ii}},\;i=1,2,\cdots,n$,
		then we obtain the following by substituting Eq. \eqref{eq:diag_dominant} into Eq. \eqref{eq:nrm_def}:
		\begin{equation*}
			\begin{aligned}
				\left\| \bm{I-\Lambda A} \right\|_{\infty}
				\nonumber
				& = \underset{i}{\text{max}}\left(1-\mu_ia_{ii}+\sum_{j=1,\;j\ne i}^{n}\mu_i|a_{ij}|\right)  \\
				\nonumber
				& = \underset{i}{\text{max}}\left(1+\mu_i\left( \sum_{j=1}^{n}|a_{ij}|-2a_{ii}\right)\right) \\
				& <1.
				\label{eq:lemma_case2}
			\end{aligned}
		\end{equation*}
		
		\item[(iii)] If $\frac{1}{a_{ii}}<\mu_i<\frac{2}{\sum_{j=1}^{n}|a_{ij}|},\;i=1,2,\cdots,n$,
		we can compute Eq.~\eqref{eq:nrm_def} in the following way:
		\begin{equation*}
			\begin{aligned}
				\left\| \bm{I-\Lambda A} \right\|_{\infty}
				\nonumber
				& = \underset{i}{\text{max}}\left(\mu_ia_{ii}-1+\sum_{j=1,\;j\ne i}^{n}\mu_i|a_{ij}|\right) \\
				\nonumber
				& = \underset{i}{\text{max}}\left(\mu_i\sum_{j=1}^{n}|a_{ij}|-1\right)                      \\
				& <1.
				\label{eq:lemma_case3}
			\end{aligned}
		\end{equation*}
	\end{enumerate}
	When $0<\mu_i<\frac{2}{\sum_{j=1}^{n}|a_{ij}|},\;i=1,2,\cdots,n$,
	    we obtain $0<\left\| \bm{I-\Lambda A} \right\|_{\infty}<1$ by
      combining the results of case (i)-(iii).
\end{proof}

\begin{theorem}
	When $\bm{A}$ is a strict diagonally dominant matrix,
	fairing-PIA~\pref{eq:kst_newctrl_mtx} is convergent.
	\label{thrm:diff_weight_convergent}
\end{theorem}
\begin{proof}
	We obtain the iterative method from Eq. \eqref{eq:kst_newctrl_mtx} as follows:
	\begin{equation}
		\begin{aligned}
			\bm{P}^{[k+1]} & = \bm{P}^{[k]} + \bm{\Lambda}\left[ \left(\bm{I}-\bm{\Omega}\right)\bm{N}^T\left(\bm{Q}-\bm{NP}^{[k]}\right)-
			\bm{\Omega}\bm{D}_{r}\bm{P}^{[k]}\right]                                                                                                                                                                                 \\
			& = \left(\bm{I}-\bm{\Lambda A}\right)\bm{P}^{[k]} +
			\bm{\Lambda}\left(\bm{I}-\bm{\Omega}\right)\bm{N}^T\bm{Q}                                                                                                                                                             \\
			& = \left(\bm{I-\Lambda A}\right)^{2}\bm{P}^{[k-1]}+\bm{\Lambda}\left(\bm{I}-\bm{\Omega}\right)\bm{N}^T\bm{Q} + \left(\bm{I-\Lambda A}\right)\bm{\Lambda}\left(\bm{I}-\bm{\Omega}\right)\bm{N}^T\bm{Q} \\
			& = \cdots                                                                                                                                                                                             \\
			& = \left(\bm{I-\Lambda A}\right)^{k+1}\bm{P}^{[0]} + \sum_{l=0}^k\left(\bm{I-\Lambda A}\right)^l\bm{\Lambda}\left(\bm{I}-\bm{\Omega}\right)\bm{N}^T\bm{Q}.
		\end{aligned}
		\label{eq:kst_ctrk_iter}
	\end{equation}
	
	Based on Lemma \ref{lemma:norm}, the spectral radius of $\bm{I}-\bm{\Lambda A}$ satisfies
	\begin{equation*}
		0<\rho\left(\bm{I}-\bm{\Lambda A}\right)<\left\|\bm{I}-\bm{\Lambda A}\right\|_{\infty}<1,
	\end{equation*}
	and we obtain
	\begin{equation*}
		\lim_{k \to \infty}\left(\bm{I-\Lambda A}\right)^k = 0 ,\;
		\sum_{l=0}^{\infty}\left(\bm{I-\Lambda A}\right)^l = \left(\bm{\Lambda A}\right)^{-1}.
	\end{equation*}
	Hence, when $k\to\infty$, Eq. \eqref{eq:kst_ctrk_iter} tends to
	\begin{equation*}
		\begin{aligned}
			\bm{P}^{[\infty]} & = \left(\bm{\Lambda A}\right)^{-1}\bm{\Lambda}\left(\bm{I}-\bm{\Omega}\right)\bm{N}^T\bm{Q} \\
			& = \bm{A}^{-1}\left(\bm{I}-\bm{\Omega}\right)\bm{N}^T\bm{Q},
		\end{aligned}
	\end{equation*}
	which is the solution of the linear system as follows,
	\begin{equation*}
		\bm{AP}^{[\infty]}=\left(\bm{I}-\bm{\Omega}\right)\bm{N}^T\bm{Q}.
		\label{eq:energy_sys_mtx}
	\end{equation*}
\end{proof}

\begin{remark}
	In practice, we select the diagonal elements of matrix $\bm{\Lambda}$ as
	$\mu_{i}=\frac{1}{\sum_{j=1}^{n}|a_{ij}|},\;i=1,2,\cdots,n$.
	\label{rmk1:diff_weight_mu}
\end{remark}

In a special case of
   $\bm{\Omega} = \omega\bm{I}$,
   i.e., all smoothing weights
   $\omega_j, j=1,2,\cdots,n$~\pref{eq:kst_newctrl} are equal to $\omega$,
   the fairing-PIA \eqref{eq:kst_newctrl_mtx} degenerates to
\begin{align}
	\nonumber
	\bm{P}^{[k+1]} & = \bm{P}^{[k]} + \bm{\Lambda}\left[ \left(1-\omega\right)\bm{N}^T\left(\bm{Q}-\bm{NP}^{[k]}\right)-
	\omega\bm{D}_{r}P^{[k]}\right]                                                                                     \\
	& = \left(\bm{I}-\bm{\Lambda B}\right)\bm{P}^{[k]}+\left(1-\omega\right)\bm{\Lambda R}^T\bm{Q},
	\label{eq:equal_weight_ctrlpts}
\end{align}
   where
\begin{equation} \label{eq:mtx_B}
   \bm{B} = \left(1-\omega\right)\bm{N}^T\bm{N}+\omega \bm{{D}}_{r}.
\end{equation}
The matrices $\bm{N}^T\bm{N}$ and $\bm{{D}}_{r}$ are real symmetric.
The matrix $\bm{B}$ is also a real symmetric matrix.
In the following,
    we will prove that the special case \eqref{eq:equal_weight_ctrlpts} of fairing-PIA is equivalent to the traditional energy minimization fairing model.

In the traditional energy minimization fairing method,
    the goal is to find the minimum of the \emph{energy functional},
    which is defined as follows:
\begin{equation}
	E = \frac{\left(1-\omega\right)}{2}f_1 + \frac{\omega}{2}f_2,
	\label{eq:total_energy}
\end{equation}
where $f_1$ and $f_2$ represent the fitting term and the fairing term, respectively, defined as
\begin{align}
	f_1 & = \sum_{i=1}^m\left\|\bm{P}(t_i)-\bm{Q}_i\right\|^2, \label{eq:obj_fit} \\
  f_2 & = \int_{t_1}^{t_m} \|\bm{P}^{(r)}(t)\|^2dt,\quad r=1,2,3,
	\label{eq:obj_eng}
\end{align}
where $\bm{P}^{(r)}(t)$ represents the $r$-th derivative of $\bm{P}(t)$,
    and $\omega$ is the smoothing weight.
The larger the smoothing weight $\omega$ is,
    the smoother the generated curve is.

\begin{remark}
	When $r = 1, 2, 3$, 
	    the fairing term $f_2$ \eqref{eq:obj_eng} corresponds to three types of energy:
      stretch energy, strain energy, and jerk energy \cite[]{veltkamp1995modeling,zhang2001fairing,meier1987interpolating}.
	\label{rmk2:equal_weight_energymeaning}
\end{remark}

The substitution of Eq.~\eqref{eq:obj_eng} into Eq.~\eqref{eq:total_energy}
    makes the optimization problem become
\begin{equation} \label{eq:energy_min_model}
	\min_{\bm{P}_j}E=\min_{\bm{P}_j}
	\left[
	\frac{\left(1-\omega\right)}{2}\sum_{i=1}^m\left\|\bm{P}(t_i)-\bm{Q}_i\right\|^2
	+ \frac{\omega}{2}\int_{t_1}^{t_m} \|\bm{P}^{(r)}(t)\|^2dt
	\right].
\end{equation}
We obtain the following by $\frac{\partial E}{\partial \bm{P}_j} =0$, with $j=1,2,\cdots,n$: 
\begin{equation}
	\begin{aligned}
		\frac{\partial E}{\partial \bm{P}_j}
		& =\left(1-\omega\right)\sum_{i=1}^mN_j(t_i)\left(\sum_{l=1}^nN_l(t_i)\bm{P}_l-\bm{Q}_i\right)\\
		& +\omega\int_{t_1}^{t_m}\left(\sum_{l=1}^nN_{r,l}(t)N_{r,j}(t)\bm{P}_l\right)dt
		=0,
	\end{aligned}
	\label{eq:energ_partial}
\end{equation}
where $\omega$ is a smoothing parameter.

Rearranging Eq.~\eqref{eq:energ_partial} yields
\begin{equation}
	\begin{aligned}
		\left(1-\omega\right)\sum_{i=1}^mN_j(t_i)\bm{Q}_i
		& =\left(1-\omega\right)\sum_{i=1}^mN_j(t_i)\left(\sum_{l=1}^nN_l(t_i)\bm{P}_l\right)\\
		& +\omega\int_{t_1}^{t_m}\left(\sum_{l=1}^nN_{r,l}(t)N_{r,j}(t)\bm{P}_l\right)dt.
	\end{aligned}
	\label{eq:simp_energ_partial}
\end{equation}
Then, Eq.~\eqref{eq:simp_energ_partial} can be represented by a matrix as
\begin{equation}
	\bm{BP}=\left(1-\omega\right)\bm{N}^T\bm{Q},
	\label{eq:equal_weight_classical}
\end{equation}
    where $\bm{B}=\left(1-\omega\right)\bm{N}^T\bm{N}+\omega \bm{D}_{r}$~\eqref{eq:mtx_B}.
A fairing curve that minimizes the energy model~\pref{eq:total_energy} can be obtained 
    after the equation is solved and the obtained control points are substituted into the curve equation.

\begin{theorem}
	Let the diagonal elements of matrix $\bm{\Lambda}$ satisfy $0<\mu_i<\frac{2}{\lambda_{\text{max}}\left(\bm{B}\right)},\;i=1,2,\cdots,n$,
      where $\lambda_{\text{max}}$ is the largest eigenvalue of matrix $\bm{B}$~\pref{eq:mtx_B}.
	If $\bm{B}$ is a positive definite matrix,
	    then the special case of the fairing-PIA \eqref{eq:equal_weight_ctrlpts} converges to
	    the solution of the traditional energy minimization fairing model~\pref{eq:energy_min_model}.
	\label{thrm:equal_weight__definite_convergent}
\end{theorem}

\begin{proof}
	According to Eq. \eqref{eq:equal_weight_ctrlpts}, we have
	\begin{equation}
		\begin{aligned}
			\bm{P}^{[k+1]} & = \left(\bm{I}-\bm{\Lambda B}\right)\bm{P}^{[k]}+\left(1-\omega\right)\bm{\Lambda R}^T\bm{Q}                                                          \\
			& = \left(\bm{I}-\bm{\Lambda B}\right)^{k+1}\bm{P}^{[0]} + \left(1-\omega\right)\sum_{l=0}^k\left(\bm{I}-\bm{\Lambda B}\right)^l\bm{\Lambda R}^T\bm{Q},
		\end{aligned}
		\label{eq:equal_weight_kst_ctrk_iter}
	\end{equation}
	where $\bm{B}$ is a positive definite matrix.
	
	A real orthogonal matrix $\bm{U}$ and a positive diagonal matrix $\bm{S}$ exist,
	    such that $\bm{B} = \bm{USU}^T=\bm{US}^{\frac{1}{2}}\bm{S}^{\frac{1}{2}}\bm{U}^T$.
	We obtain $\bm{B} = \bm{C}^T\bm{C}$ by denoting $\bm{C}=\bm{S}^{\frac{1}{2}}\bm{U}^T$.
	Suppose that $\beta$ is an arbitrary eigenvalue of the matrix
    $\bm{\Lambda B} = \bm{\Lambda C}^T\bm{C}$ w.r.t. eigenvector $\bm{v}$, i.e.,
	\begin{equation}
		\bm{\Lambda B} \bm{v} = \bm{\Lambda C}^T\bm{Cv} = \beta\bm{v}.
		\label{eq:equal_weight_symm}
	\end{equation}
	In this case, 
      we obtain the following by multiplying both sides of Eq. \eqref{eq:equal_weight_symm} by $\bm{C}$:
	\begin{equation*}
		\bm{C\Lambda C}^T\left(\bm{Cv}\right) = \beta\left(\bm{Cv}\right),
	\end{equation*}
	    which means that $\beta$ is also an eigenvalue of the matrix $\bm{C\Lambda C}^T$ w.r.t. eigenvector $\bm{Cv}$.
	Moreover, for all nonzero vectors $\bm{x}\in R^{n}$, it holds
	\begin{equation*}
		\bm{x}^T\bm{C\Lambda C}^T\bm{x} =
		\left(\bm{x}^T\bm{C\Lambda}^{\frac{1}{2}}\right)\left(\bm{x}^T\bm{C\Lambda}^{\frac{1}{2}}\right)^T> 0.
	\end{equation*}
	Therefore, we conclude that the matrix $\bm{C\Lambda C}^T$ is a positive definite matrix,
 	    and all its eigenvalues are positive, i.e., $\beta>0$.
  Thus, the eigenvalues of matrix
      $\bm{\Lambda B}=\bm{\Lambda C}^T\bm{C}$
      are all positive real numbers.
	
	On the other hand, given that $0<\mu_i<\frac{2}{\lambda_{\text{max}}\left(\bm{B}\right)},\;i=1,2,\cdots,n$,
	    the eigenvalues of $\bm{\Lambda B}$ satisfy
	\begin{equation*}
		0 < \lambda\left(\bm{\Lambda B}\right)
		< \left\| \bm{\Lambda B}\right\|_2 
        \leq \left\| \bm{\Lambda}\right\|_2 \left\| \bm{B}\right\|_2
        = \lambda_{\text{max}}\left(\bm{\Lambda}\right)\lambda_{\text{max}}\left(\bm{B}\right)
		< 2.
	\end{equation*}
	Thus, the eigenvalues of $\left(\bm{I}-\bm{\Lambda B}\right)$ fulfill the condition
	$-1<\lambda\left(\bm{I}-\bm{\Lambda B}\right)<1$, leading to
	\begin{equation}
		\rho\left(\bm{I}-\bm{\Lambda B}\right)<1.
		\label{eq:equal_weight_spec_radius}
	\end{equation}
	We obtain the following from Eq. \eqref{eq:equal_weight_spec_radius}: 
	\begin{equation*}
		\lim_{k \to \infty}\left(\bm{I}-\bm{\Lambda B}\right)^k = 0 ,\;
		\sum_{l=0}^{\infty}\left(\bm{I}-\bm{\Lambda B}\right)^l = \left(\bm{\Lambda B}\right)^{-1}.
	\end{equation*}
	Hence, when $k\rightarrow\infty$, Eq. \eqref{eq:equal_weight_kst_ctrk_iter} tends to
	\begin{equation}
		\begin{aligned}
			\bm{P}^{[\infty]} & = \left(1-\omega\right)\left(\bm{\Lambda}\bm{B}\right)^{-1}\bm{\Lambda}\bm{N}^T\bm{Q} \\
			& = \left(1-\omega\right)\bm{B}^{-1}\bm{N}^T\bm{Q}.
		\end{aligned}
		\label{eq:equal_weight_limit}
	\end{equation}
	
	The combination of Eqs.~\eqref{eq:equal_weight_limit} and~\eqref{eq:equal_weight_classical}
      indicate that the fairing-PIA~\eqref{eq:equal_weight_ctrlpts} is convergent to
	    the solution of the traditional energy minimization fairing model~\pref{eq:energy_min_model}.
\end{proof}

 We deal with the case that matrix $\bm{B}$ is positive definite in the above Theorem~\ref{thrm:equal_weight__definite_convergent}.
 When matrix $B$ is positive semidefinite,
    the convergence proof is complicated and deferred to Appendix for clarity of presentation.

\begin{remark}
	For the case when matrix $\bm{\Omega}$~\eqref{eq:kst_newctrl_mtx} degrades to $\omega \bm{I}$~\eqref{eq:equal_weight_ctrlpts},
	    calculating the largest eigenvalue of the coefficient matrix $\bm{B}$ is time-consuming.
	Therefore, we choose $\mu_i=\mu=\frac{2}{\left\|\bm{B}\right\|_{\infty}},\;i=1,2,\cdots,n$,
	    which satisfies the condition in Theorem~\ref{thrm:equal_weight__definite_convergent}.
	\label{rmk3:equal_weight_mu}
\end{remark}

\section{Progressive-iterative approximate for surface fairing}
\label{section:surface_fairing}
The fairing-PIA method for curve fairing can be easily extended to the surface case.
In the following, we present the details of fairing-PIA for surface fairing.

Given an ordered data set $\{\bm{Q}_{ij}\}_{i=1,j=1}^{m_1,m_2}$,
    with parameters $\{s_i,t_j\}_{i=1,j=1}^{m_1,m_2}$ 
    satisfying $s_1\le s_2\le\cdots\le s_{m_1}$ and $t_1\le t_2\le\cdots\le t_{m_2}$,
    we rearrange the data points and their parameters in lexicographical order,
    i.e., $\{\bm{Q}_i\}_{i=1}^m$ with $\{(s_i,t_i)\}_{i=1}^{m}$,
    where $m = m_1m_2$.
The initial B-spline surface $\bm{P}^{[0]}\left(s,t\right)$ can be constructed as follows:
\begin{equation*}
\bm{P}^{[0]}\left(s,t\right) = \sum_{h=1}^{n_1}\sum_{l=1}^{n_2}N_h(s)N_l(t)\bm{P}_{hl}^{[0]},\quad 0\le s,t\le1,
\end{equation*}
where  $\{\bm{P}_{hl}\}_{h=1,l=1}^{n_1,n_2}$ are the initial control points,
   and $N_h(s)$ and $N_l(t)$ are the B-spline basis.
Similarly, the basis functions and control points are arranged in lexicographical order as follows:
\begin{equation*}
	\begin{aligned}
		\left [ N_1(s,t),N_2(s,t),\cdots,N_{n_1n_2}(s,t) \right ]
		& =\left[N_1(s)N_1(t),N_1(s)N_2(t),\cdots,N_{n_1}(s)N_{n_2}(t)\right], \\
		\left [ \bm{P}_1,\bm{P}_2,\cdots,\bm{P}_{n_1n_2} \right ]
		& =\left[\bm{P}_{11},\bm{P}_{12},\cdots,\bm{P}_{n_1,n_2}\right].
	\end{aligned}
\end{equation*}
Then, the initial surface can be rewritten as
\begin{equation}
	\bm{P}^{[0]}\left(s,t\right) = \sum_{j=1}^{n_1n_2}N_j(s,t)\bm{P}_{j}^{[0]}.
	\label{eq:init_surf}
\end{equation}

We define the fairing functionals analogous to Eq.~\eqref{eq:def_of_functional}.
Let
\begin{equation*}
	\begin{aligned}
		\mathcal{F}_{ss,j}(f) & = \int_{s_1}^{s_m}\int_{t_1}^{t_m}N_{ss,j}(s,t)fdsdt,   \\
		\mathcal{F}_{st,j}(f) & = \int_{s_1}^{s_m}\int_{t_1}^{t_m}N_{st,j}(s,t)fdsdt,\ \quad j=1,2,\cdots,n_1n_2,    \\
		\mathcal{F}_{tt,j}(f) & = \int_{s_1}^{s_m}\int_{t_1}^{t_m}N_{tt,j}(s,t)fdsdt,
	\end{aligned}
\end{equation*}
where $N_{ss,j}(s,t)$, $N_{st,j}(s,t)$ and $N_{tt,j}(s,t)$ represent the second-order partial derivatives of $N_j(s,t)$ with respect to $s$ and $t$.
Here, we take the fairing energy functional defined using
    the second-order derivatives as examples.
The definition of the fairing functional can be easily extended to other cases~\cite{wang1997energy}.

Similar to the curve case,
    the $k$-th surface $\bm{P}^{[k]}(s,t)$ is assumed to be generated after the $k$-th iteration.
In this case, 
    we first calculate the $(k+1)$-st difference vectors for data points as follows
    to construct the $(k+1)$-st surface $\bm{P}^{[k+1]}(s,t)$:
\begin{equation*}
	\bm{d}_i^{[k]} = \bm{Q}_i-\bm{P}^{[k]}(s_i,t_i),\quad i=1,2,\cdots,m.
\end{equation*}
Then, the fitting vectors for control points and the fairing vectors for control points are calculated by
\begin{equation*}
\begin{aligned}
	\bm{\delta}_j^{[k]} &= \sum_{h\in \hat{I_j}}N_j(s_h,t_h)\bm{d}_h^{[k]},\quad j=1,2,\cdots,n_1n_2,\;\text{and},\\
	\bm{\eta}_{j}^{[k]} &=\sum_{l=1}^{n_1n_2}
	\left[\mathcal{F}_{ss,l}\left(N_{ss,j}(s,t)\right)
	+ 2\mathcal{F}_{st,l}\left(N_{st,j}(s,t)\right)
	+ \mathcal{F}_{tt,l}\left(N_{tt,j}(s,t)\right)
	\right]\bm{P}_l^{[k]},
	\label{eq:kst_fairvect_surf}
\end{aligned}
\end{equation*}
where $\hat{I_j}$ is the index set of all parameters that fall in the support region of the $j$-th basis function.
Finally, the new control points of the $(k+1)$-st surface are produced by
\begin{equation}
	\bm{P}_j^{[k+1]} = \bm{P}_j^{[k]} + \mu_j
	\left[
	\left(1-\omega_j\right)\bm{\delta}_j^{[k]} - \omega_j\bm{\eta}_{j}^{[k]}
	\right],
	\label{eq:kst_newctrl_surf}
\end{equation}
leading to the $(k+1)$-st B-spline surface:
\begin{equation*}
	\bm{P}^{[k+1]}(s,t)=\sum_{j=1}^{n_1n_2}N_j(s,t)\bm{P}_j^{[k+1]}.
	\label{eq:kst_surf}
\end{equation*}
Here, the parameters $\omega_j$ and $\mu_j$ \eqref{eq:kst_newctrl_surf} are similar to those in Section~\ref{section:curve_fairing}.

In this way, 
    we obtain a sequence of surfaces $\left\{\bm{P}^{[k]}(s,t),\;k=1,2,3,\cdots\right\}$.
The convergence analysis of the fairing-PIA for the surface
    case~\pref{eq:kst_newctrl_surf} is analogous to that of the curve case in Section~\ref{section:proof_convergence}.

\begin{figure*}[!htb]
  \centering
  \subfigure[]{
    \includegraphics[width=1.5in]{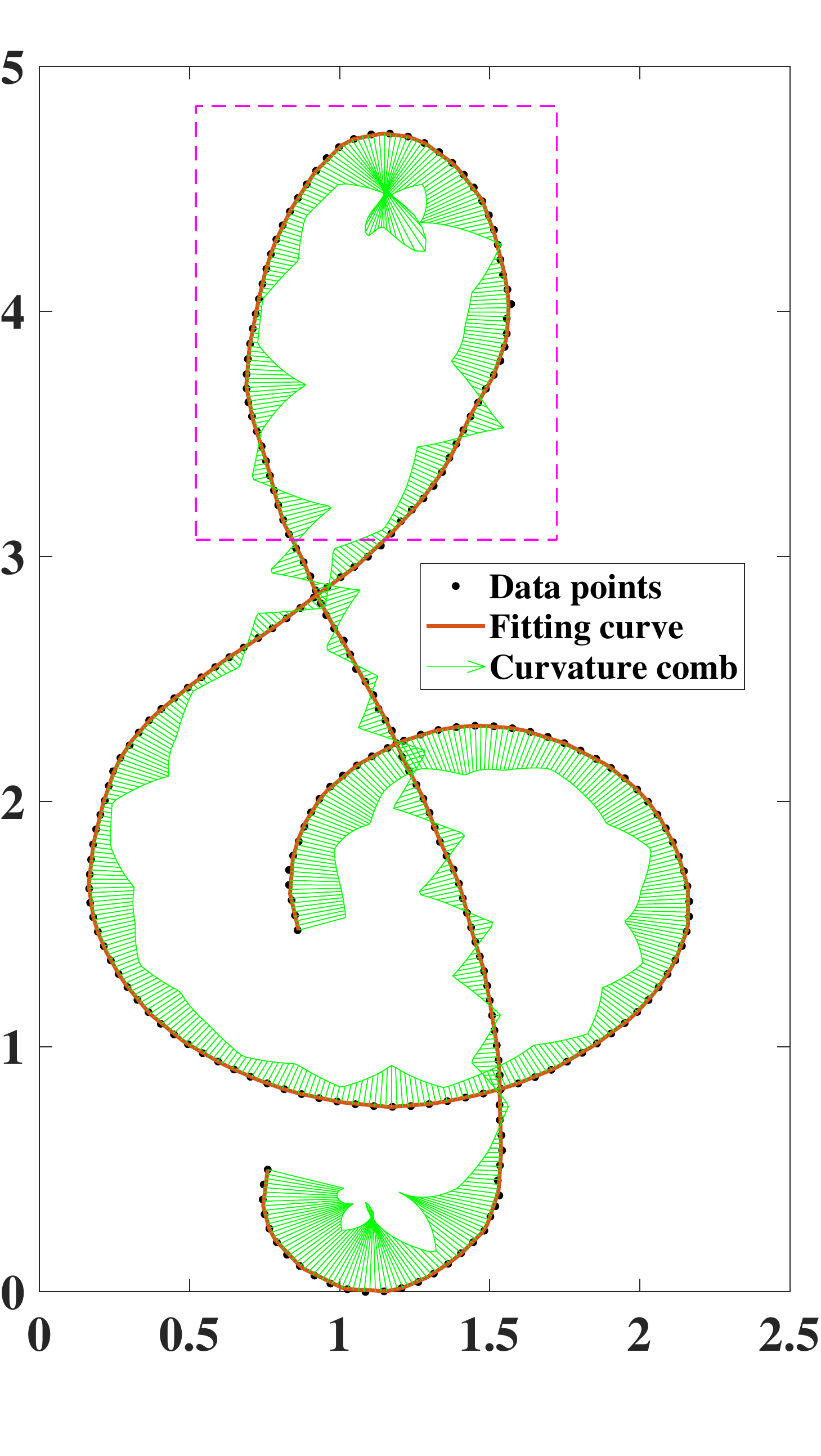}
    \label{fig:Note_fit_comb}
  }
  \subfigure[]{
    \includegraphics[width=1.5in]{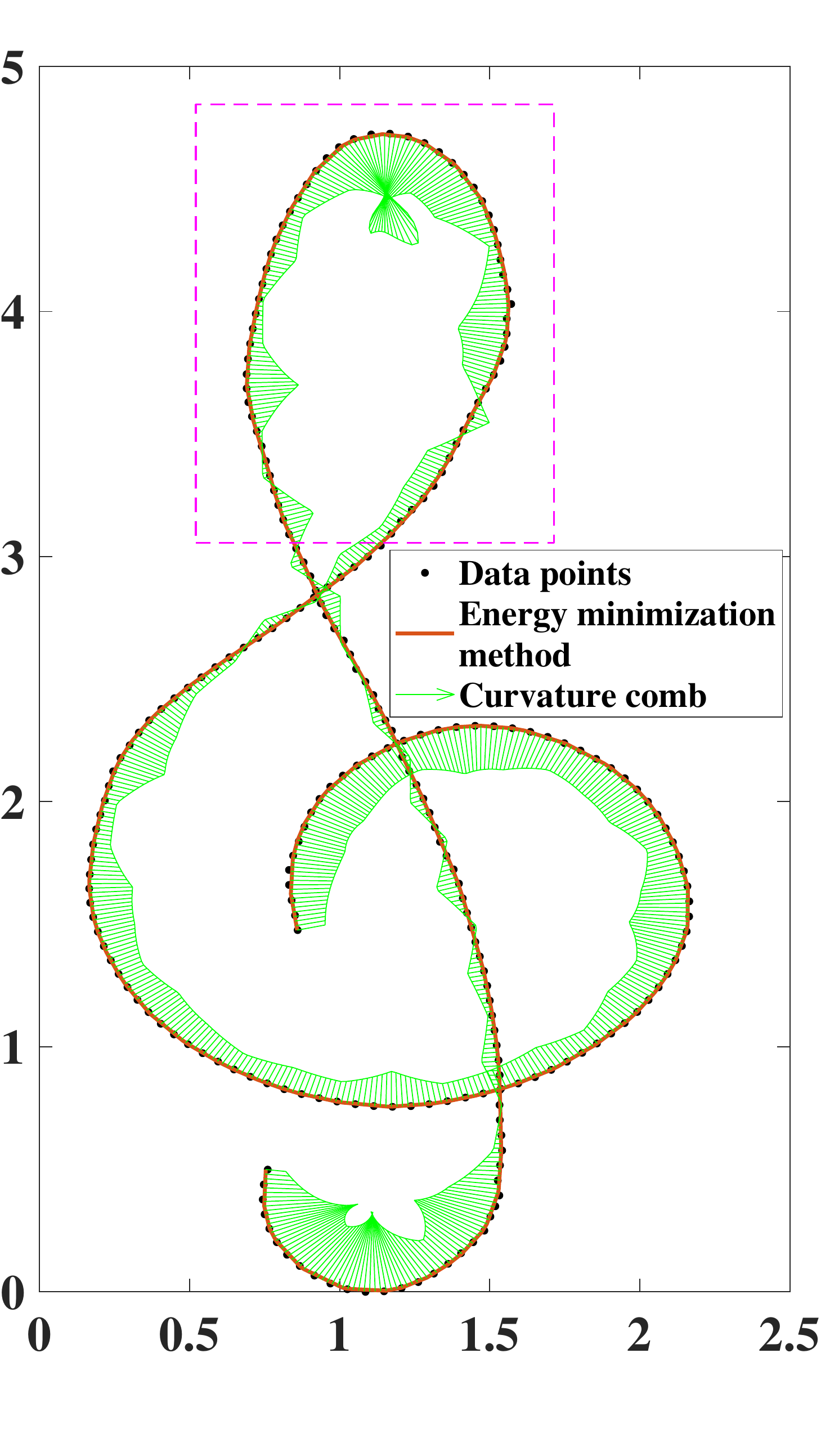}
    \label{fig:Note_energy_comb}
  }
  \subfigure[]{
    \includegraphics[width=1.5in]{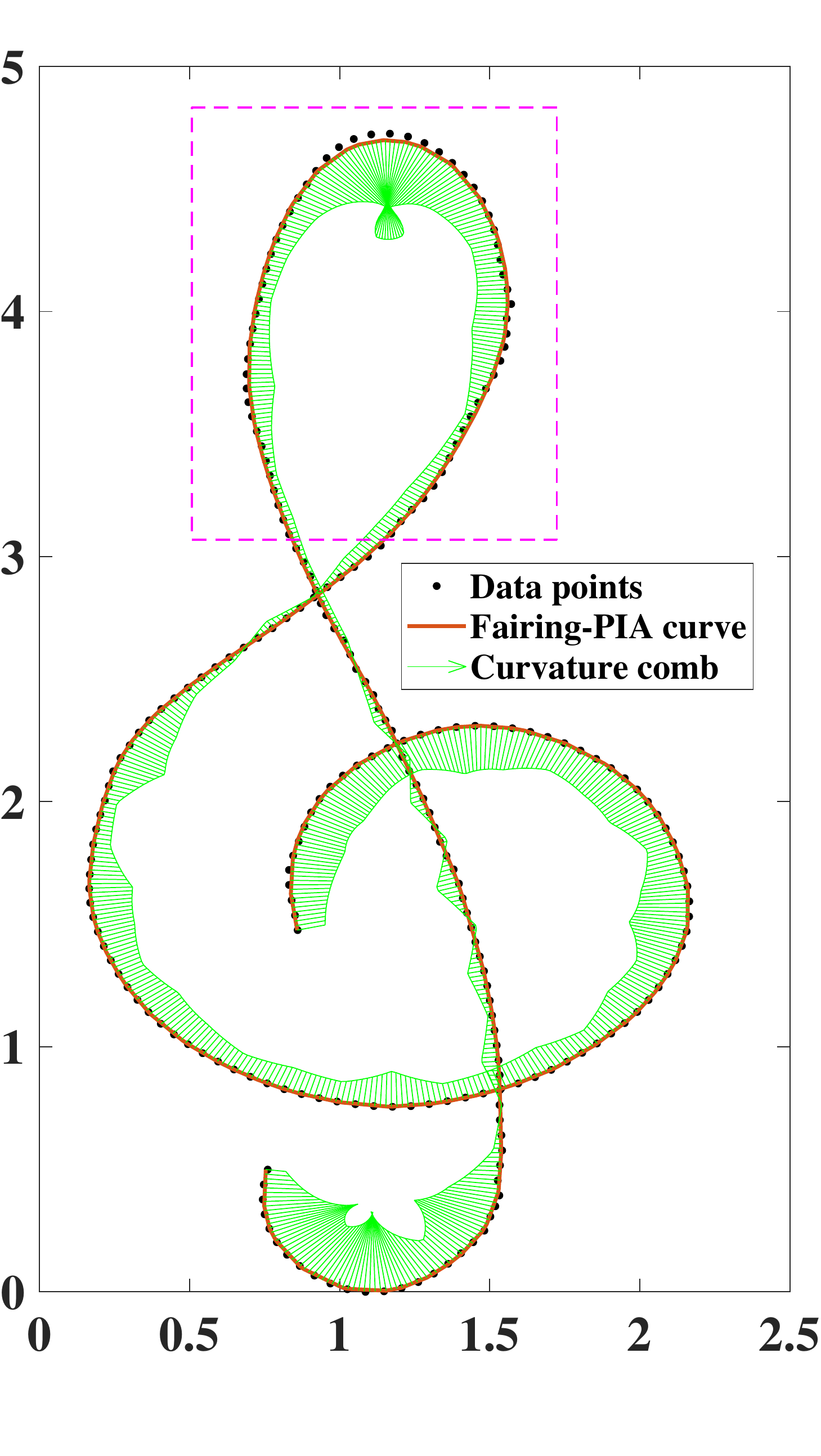}
    \label{fig:Note_FPIA_comb}
  }
  \\
  \subfigure[]{
    \includegraphics[width=1.5in]{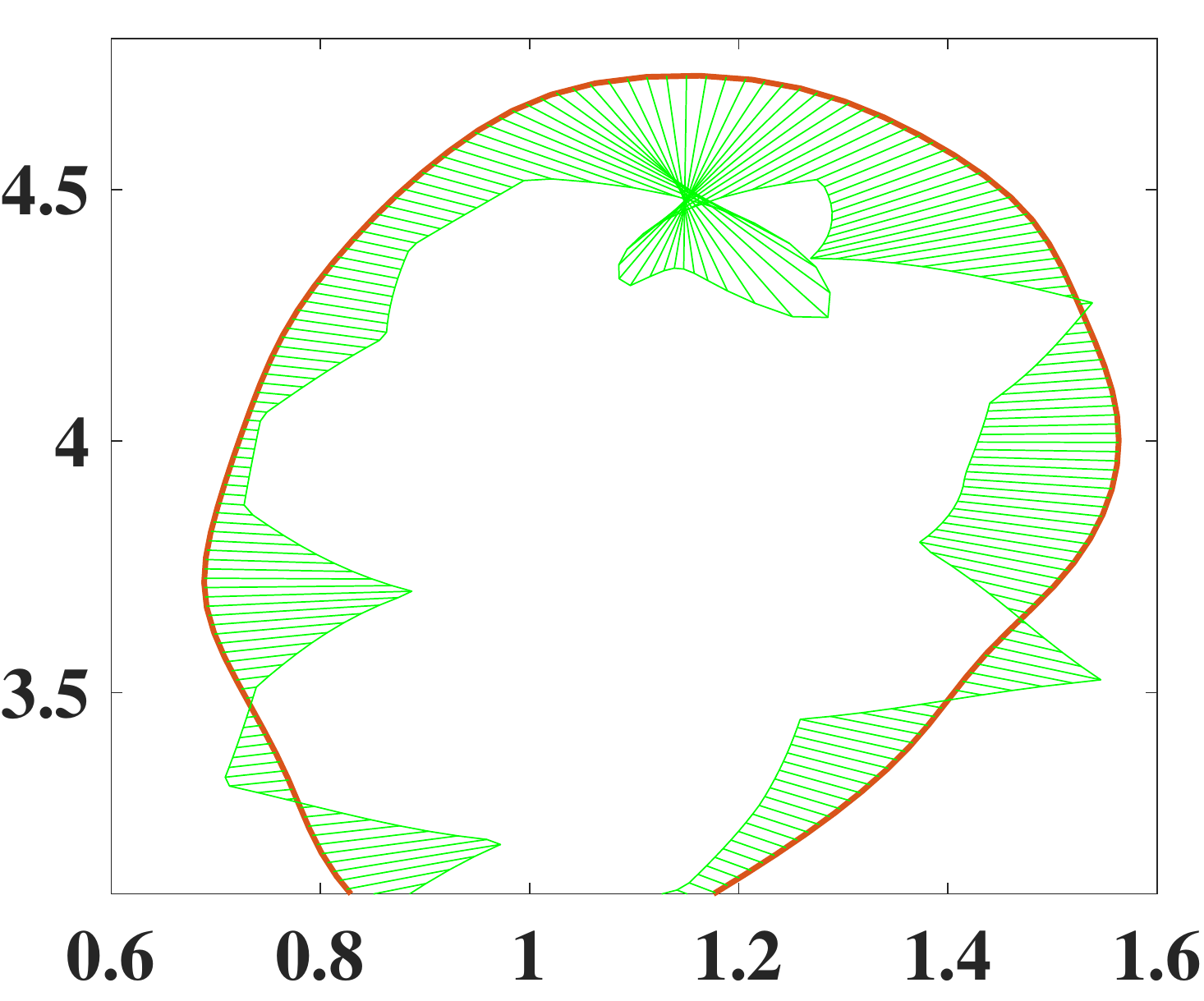}
    \label{fig:Note_fit_comb_zoom}
  }
  \subfigure[]{
    \includegraphics[width=1.5in]{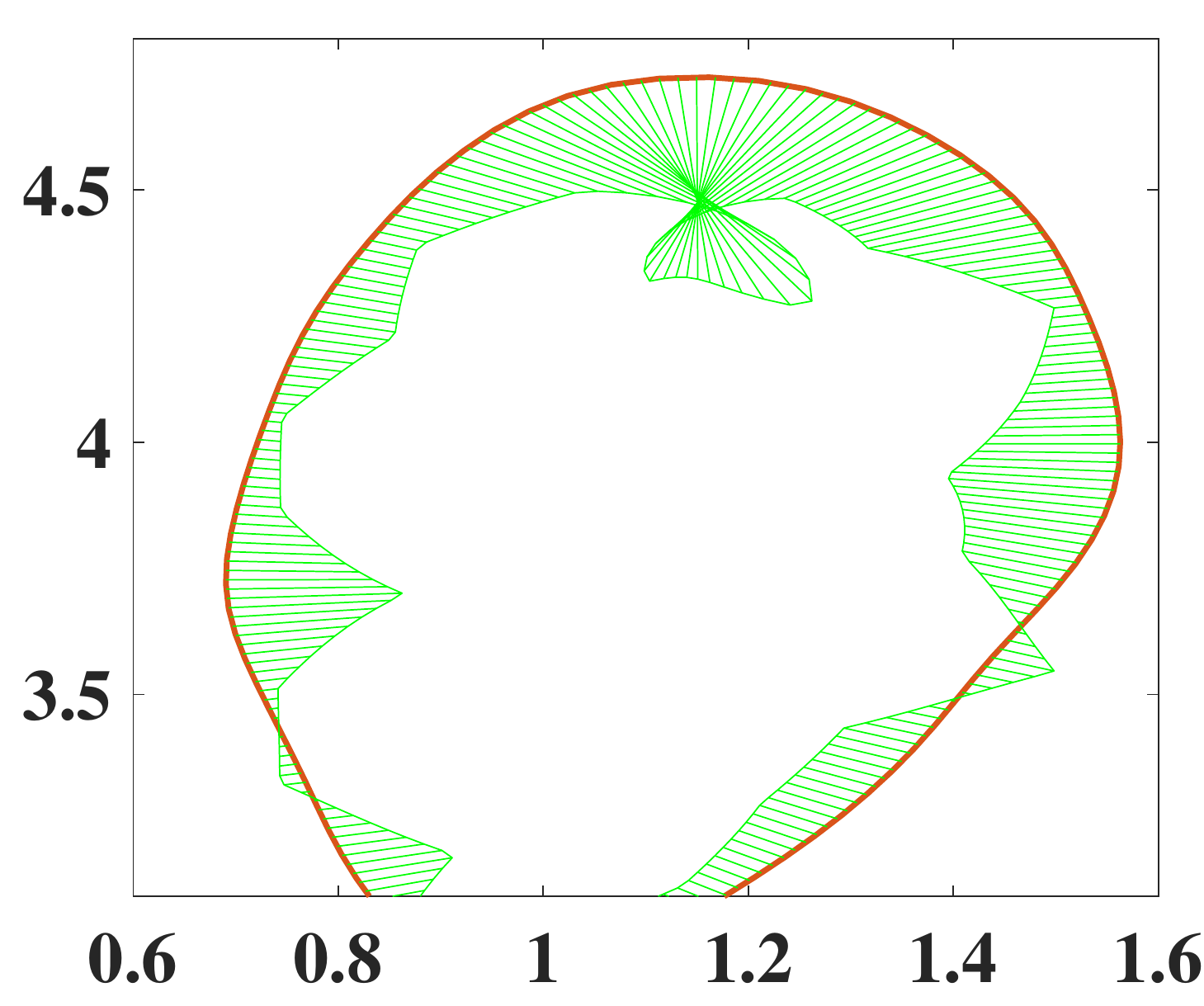}
    \label{fig:Note_energy_comb_zoom}
  }
  \subfigure[]{
    \includegraphics[width=1.5in]{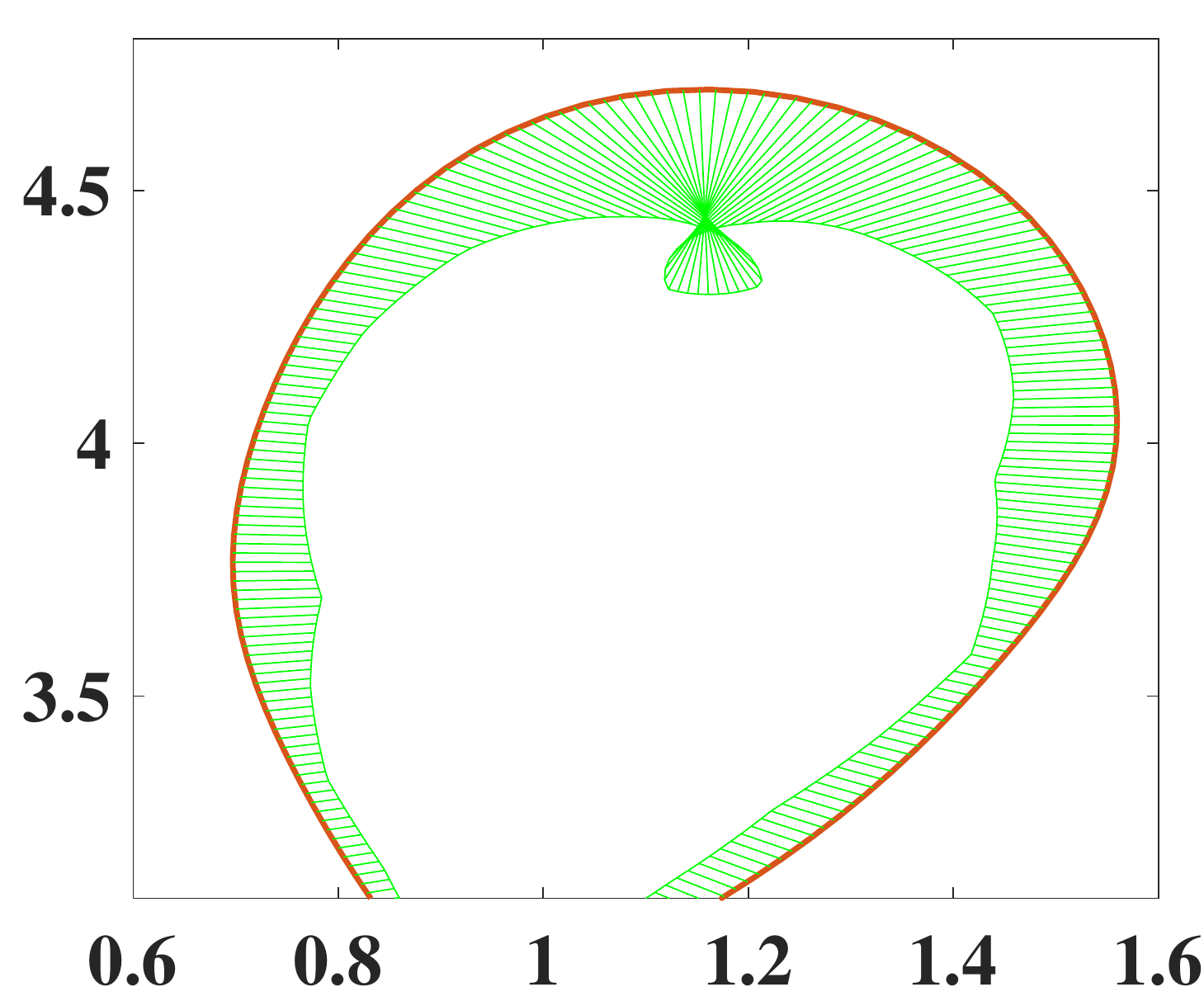}
    \label{fig:Note_FPIA_comb_zoom}
  }
  \caption
  { The \textit{Note} example.
    (a, d) The fitting curve and its curvature comb;
    (b, e) The fairing curve generated by the traditional energy minimization method and its curvature comb.
           The smoothing weight
           $\omega = 1\times 10^{-7}$~\eqref{eq:total_energy} and
           $r=2$~\eqref{eq:obj_eng};
    (c, f) The fairing curve generated by the fairing-PIA method and its
        curvature comb.
        The smoothing weights of the $38^{th}$ to $61^{st}$ control points are set to $3\times 10^{-6}$,
        and those of the other control points are set to $1\times 10^{-7}$.
    }
  \label{fig:Note_comb}
\end{figure*}

\section{Experiments and Discussions}
\label{section:results}
In this section,
    some numerical examples are presented to demonstrate the effectiveness and efficiency of the developed fairing-PIA method.
Moreover, we compare the fairing-PIA method with the traditional energy
    minimization method~\eqref{eq:total_energy},
    which is a special case of the fairing-PIA method,
    i.e., the fairing-PIA method with equal weights.

In the numerical examples,
    cubic B-spline curves are considered the fitting curves.
We choose the initial control points $\{\bm{P}_j^{[0]}\}_{j=1}^n$ from the data points, i.e.,
\begin{equation*}
  P^{[0]}_j=Q_{i_j},\ 1=i_1 < i_2 < \cdots < i_n = m.
\end{equation*}
Then, we construct a knot vector by using the following form:
\begin{equation*}
  \{0,0,0,0,\bar{t}_5,\cdots,\bar{t}_{n},1,1,1,1\},
\end{equation*}
    where $\bar{t}_i=\left(t_{i_{j-3}}+t_{i_{j-2}}+t_{i_{j-1}}\right)/3,\;j=5,\cdots,n$.
 Moreover, bi-cubic B-spline surfaces are considered the fitting surfaces,
    and the initial control points and knot vectors are formed similarly.

 On the one hand, 
    we employ the root mean square as the \emph{absolute fitting error},
    \begin{equation*}
        \bm{E}_{fit,a}^{[k]} = 
            \sqrt{\frac{\sum_{i=1}^m \left ( \bm{Q}_i - \bm{P}^{[k]}(t_i)\right )^2}{m}},
    \end{equation*}
    \newpage \noindent
    and Eq.~\pref{eq:obj_eng} as the \emph{absolute energy},
    \begin{equation*}
        \bm{E}_{eng,a} = \int_{t_1}^{t_m} \|\bm{P}^{(r)}(t)\|^2dt,\quad r=1,2,3,
    \end{equation*}
    to measure the fitting error and the energy of the curve and surface.
 On the other hand, 
    we define the \emph{relative fitting error} $E_{fit,r}^{[k]}$, 
    \emph{relative energy} $E_{eng,r}^{[k]}$, i.e.,
    \begin{equation*}
        E_{fit,r}^{[k]} = \frac{E_{fit,a}^{[k]}}{E_{fit,a}^{[0]}}, \quad
        E_{eng,r}^{[k]} = \frac{E_{eng,a}^{[k]}}{E_{eng,a}^{[0]}},
    \end{equation*}
    and the \emph{relative iteration error},
\begin{equation*}
    \bm{E}_{iter,r}^{[k]}  = \sqrt{\frac{\sum_{j=1}^n \left\|\left(\bm{I}-\bm{\Omega}\right)\bm{N}^T\bm{Q}(j,:)-\bm{AP}^{[k]}(j,:)\right\|^2}
    {\sum_{j=1}^n   \left\|\left(\bm{I}-\bm{\Omega}\right)\bm{N}^T\bm{Q}(j,:)-\bm{AP}^{[0]}(j,:)\right\|^2}},\; k=0,1,2,\cdots,
  \end{equation*}
    where matrices $\bm{\Omega}$, $\bm{N}$ and $\bm{A}$ are given in Eq.~\eqref{eq:kst_newctrl_mtx}
    to measure the iteration procedure of the fairing-PIA.
 In our implementation,
    the fairing-PIA iteration stops   
    when $\left | \bm{E}_{iter,r}^{[k+1]}-\bm{E}_{iter,r}^{[k]} \right | <10^{-6}$.
 We implement the algorithms in MATLAB and run them on a PC with an Intel Core i7-10700 2.90GHz CPU and 16 GB RAM.

\subsection{Planar curves}

In this section, 
    three examples of planar curve fairing are presented 
    to illustrate the effectiveness of the fairing-PIA algorithm. 
The three models are as follows:

\begin{itemize}[itemsep=0pt,topsep=0pt,parsep=0pt]
  \item \textit{Note}: 245 points sampled from a note-shaped curve
  \item \textit{Airfoil}: 49 points measured from a supercritical airfoil
  \item \textit{Starfish}: 100 points sampled uniformly from an analytical curve expressed by
        \begin{equation*}
          \left\{
          \begin{aligned}
            x(t) & = \left( 1+\frac{1}{5}\cos(5t) \right)\cos(t),\quad t\in [0,2\pi] \\
            y(t) & = \left( 1+\frac{1}{5}\cos(5t) \right)\sin(t).
          \end{aligned}
          \right.
        \end{equation*}
\end{itemize}
 In the three examples,
    we take $r = 2$ in the fairing functional~\pref{eq:def_of_functional},
    corresponding to the strain energy in the traditional method~\pref{eq:obj_eng}.

 In the example of the \textit{Note} model,
    we first generate the fitting curve with the curvature comb 
    (that is, the smoothing weight $\omega = 0$ is set in the energy functional~\pref{eq:total_energy},
    illustrated in Fig.~\ref{fig:Note_fit_comb}).
 Then, the fairing curve using the traditional method is produced by setting
    $\omega = 1 \times 10^{-7}$ and $r=2$ (strain energy) in the energy functional~\pref{eq:total_energy} 
    (refer to Fig.~\ref{fig:Note_energy_comb}).
 The fairness of the curve segment in the red box
    (refer to Figs.~\ref{fig:Note_energy_comb} and~\ref{fig:Note_energy_comb_zoom}) is undesirable
    and needs to be improved.
 Each control point can be assigned with an individual smoothing weight in the fairing-PIA.
 Thus, we set the smoothing weights of the $38^{th}$ to $61^{st}$ control points as $3 \times 10^{-6}$,
    thereby affecting the curve segment with undesirable fairness.
 The smoothing weights of the other control points are still set as $1 \times 10^{-7}$,
    and the fairing curve is generated using fairing-PIA.
 As illustrated in Figs.~\ref{fig:Note_FPIA_comb} and~\ref{fig:Note_FPIA_comb_zoom},
    the fairness of the curve segment in the red box is clearly improved with different smoothing weights.
 Moreover, the diagrams of iterations v.s. relative iterative error,
    iterations v.s. relative fitting error,
    and iterations v.s. relative strain energy are shown in Fig.~\ref{fig:Note_ErrandEnergyPlot}.
 The relative iterative error, 
    relative fitting error, 
    and relative strain energy are rapidly reduced at the early iterations.

\begin{figure*}[t]
  \centering
  \subfigure[]{
    \includegraphics[width=2.5in]{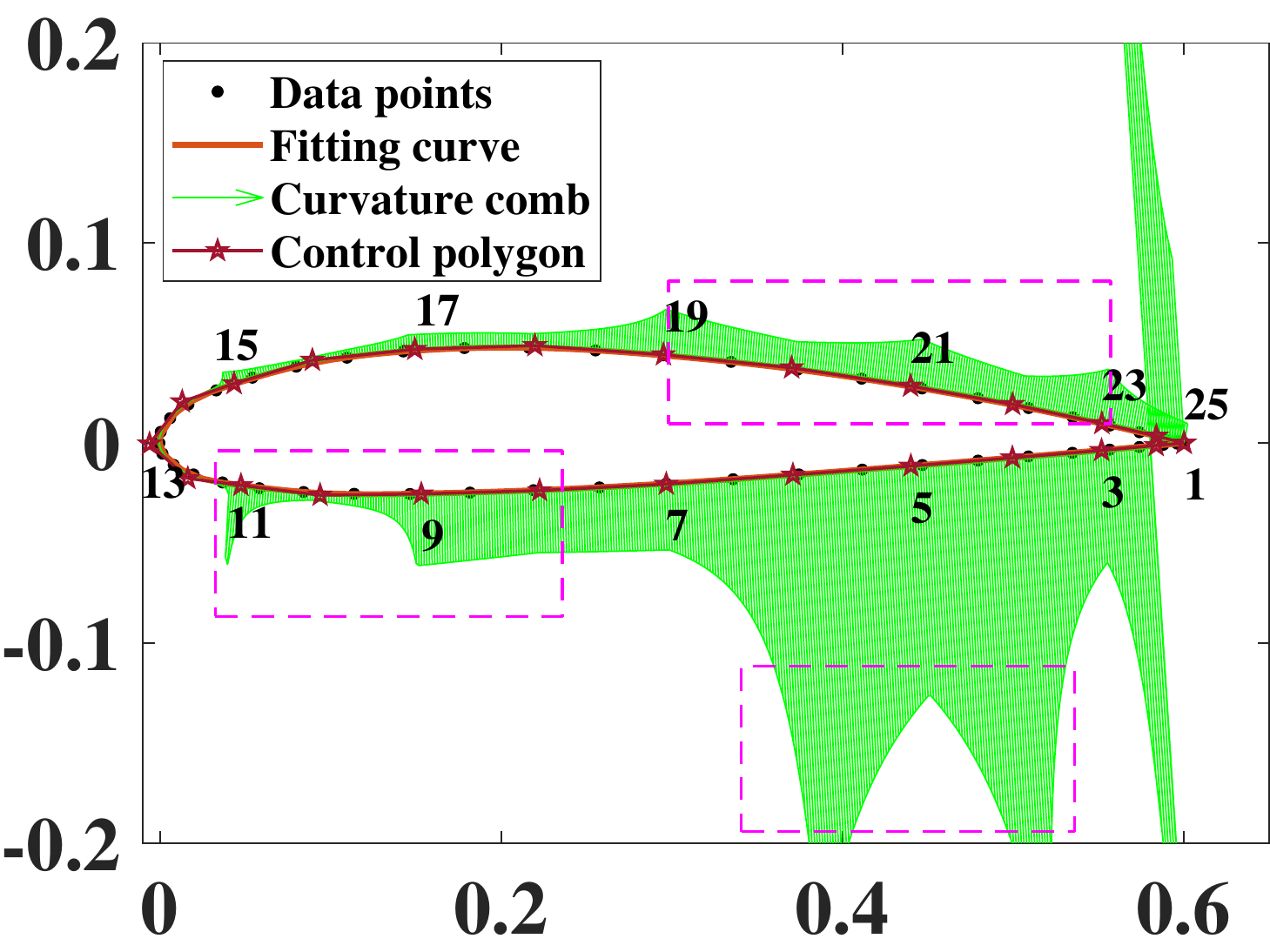}
    \label{fig:Airfoil_fit_comb}
  }
  \subfigure[]{
    \includegraphics[width=2.5in]{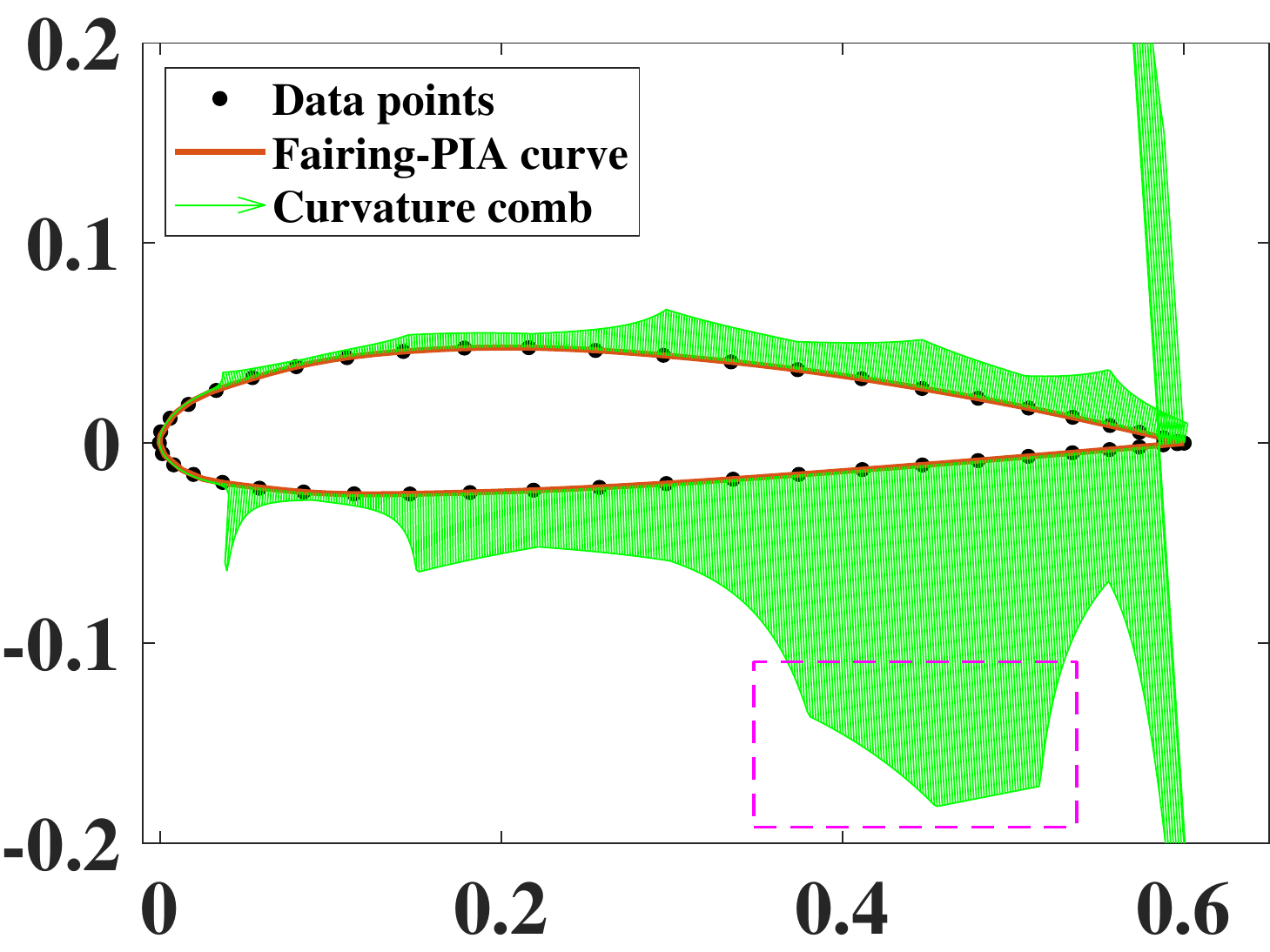}
    \label{fig:Airfoil_FPIA_comb_1}
  }

  \subfigure[]{
    \includegraphics[width=2.5in]{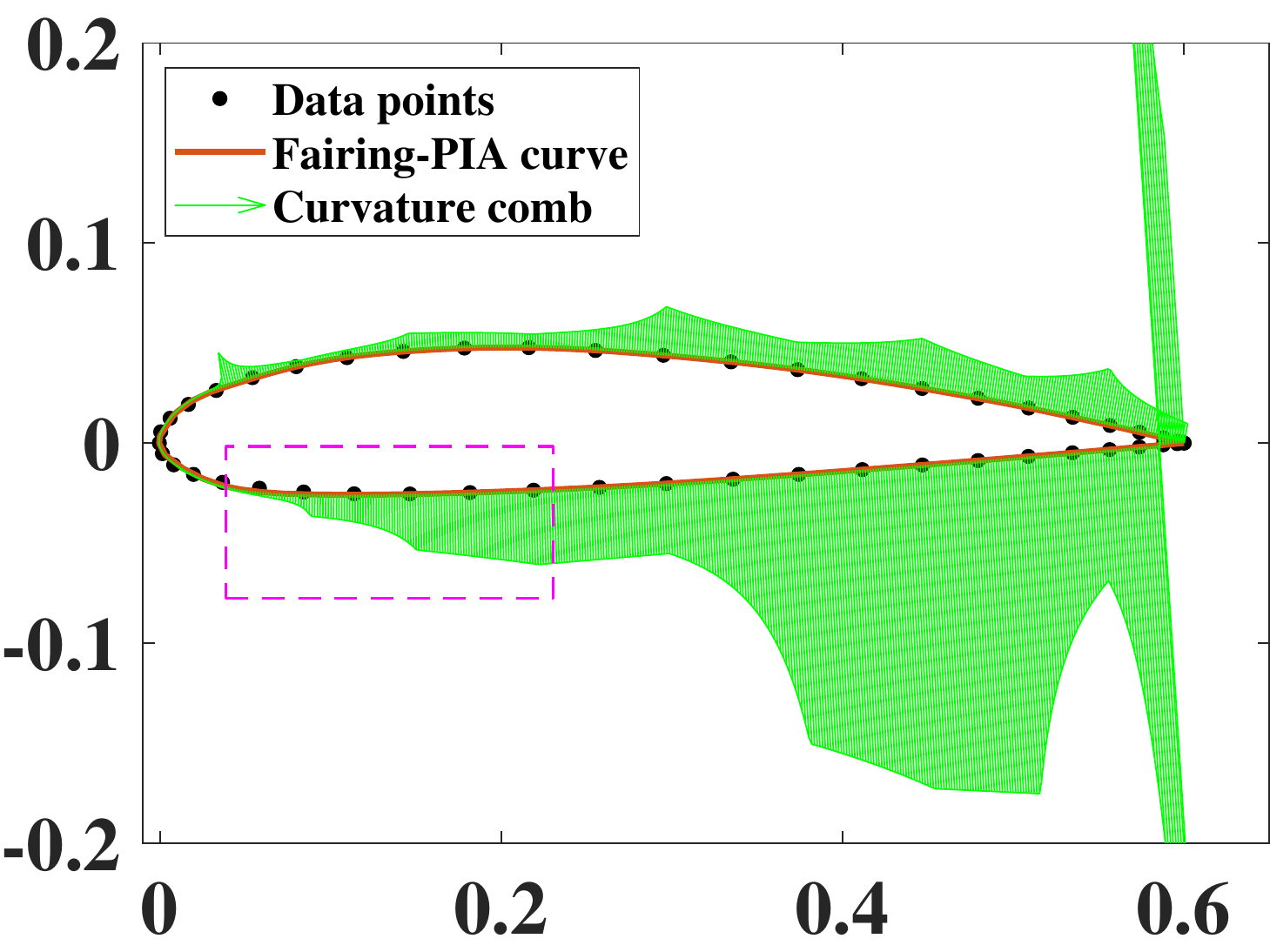}
    \label{fig:Airfoil_FPIA_comb_2}
  }
  \subfigure[]{
    \includegraphics[width=2.5in]{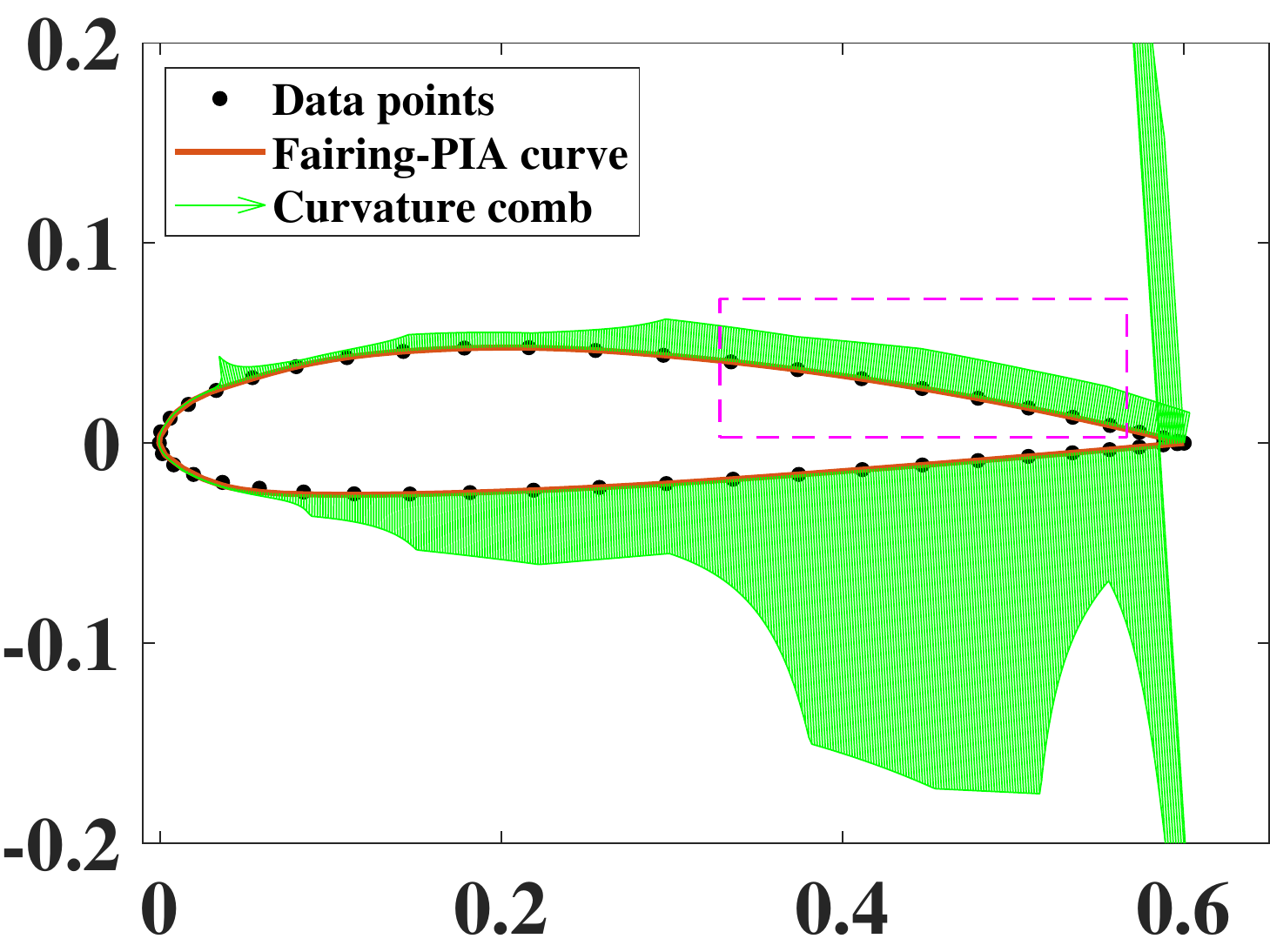}
    \label{fig:Airfoil_FPIA_comb_3}
  }

  \caption{ The \textit{Airfoil} example.
  (a) Data sets, the fitting curve with its curvature comb, 
    and its control polygon.
  (b) The first local fairing by adjusting the smoothing weights of the $4^{th}$-$7^{th}$ control points;
  (c) The second local fairing by adjusting the smoothing weights of the $8^{th}$-$11^{th}$ control points;
  (d) The third local fairing by adjusting the smoothing weights of the $19^{th}$-$23^{rd}$ control points.}
  \label{fig:Airfoil_comb}
\end{figure*}

\begin{figure*}[!htb]
  \centering
  \subfigure[]{
    \includegraphics[width=2.0in]{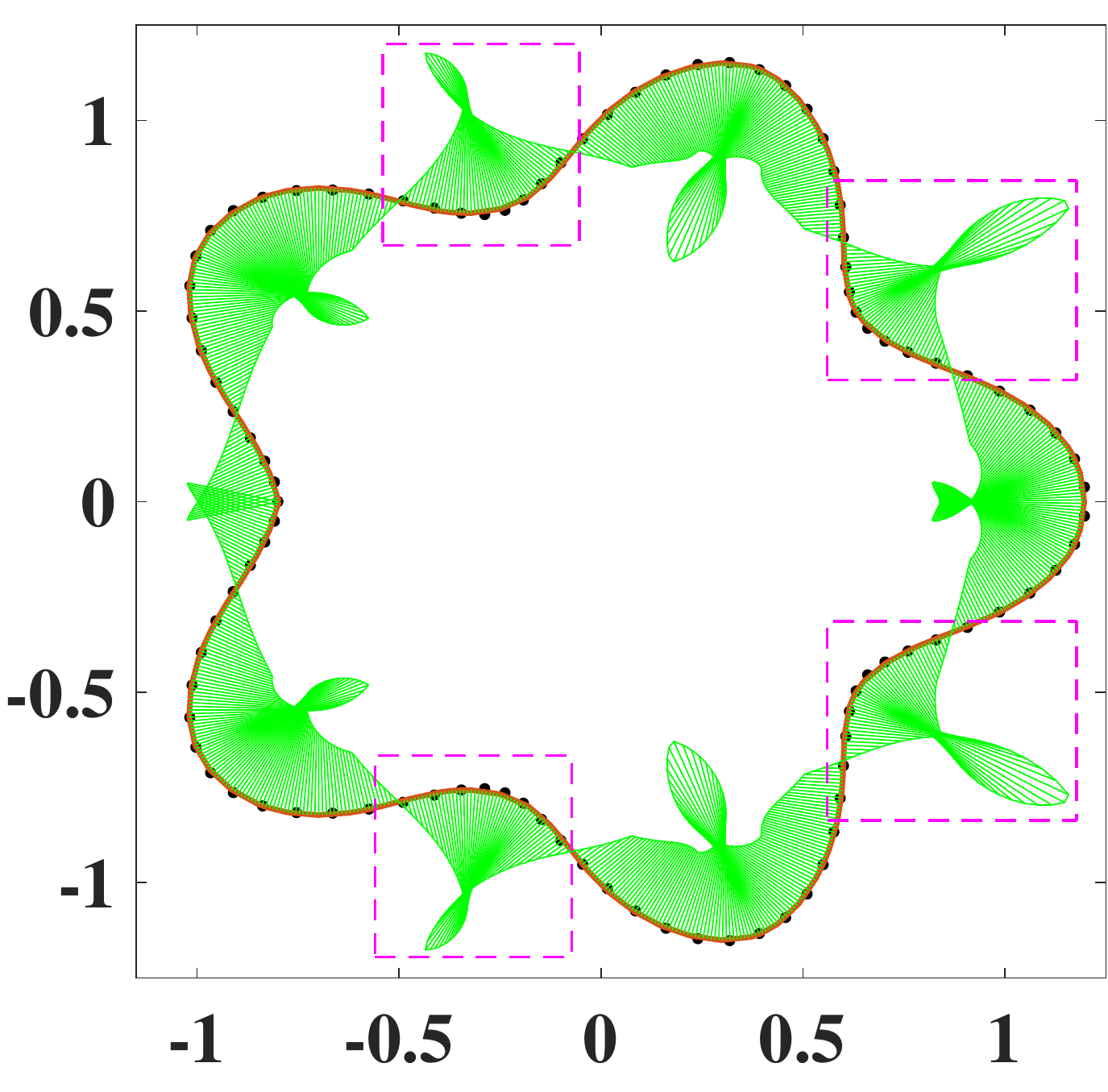}
    \label{fig:Star_fit_comb}
  }
  \subfigure[]{
    \includegraphics[width=2.0in]{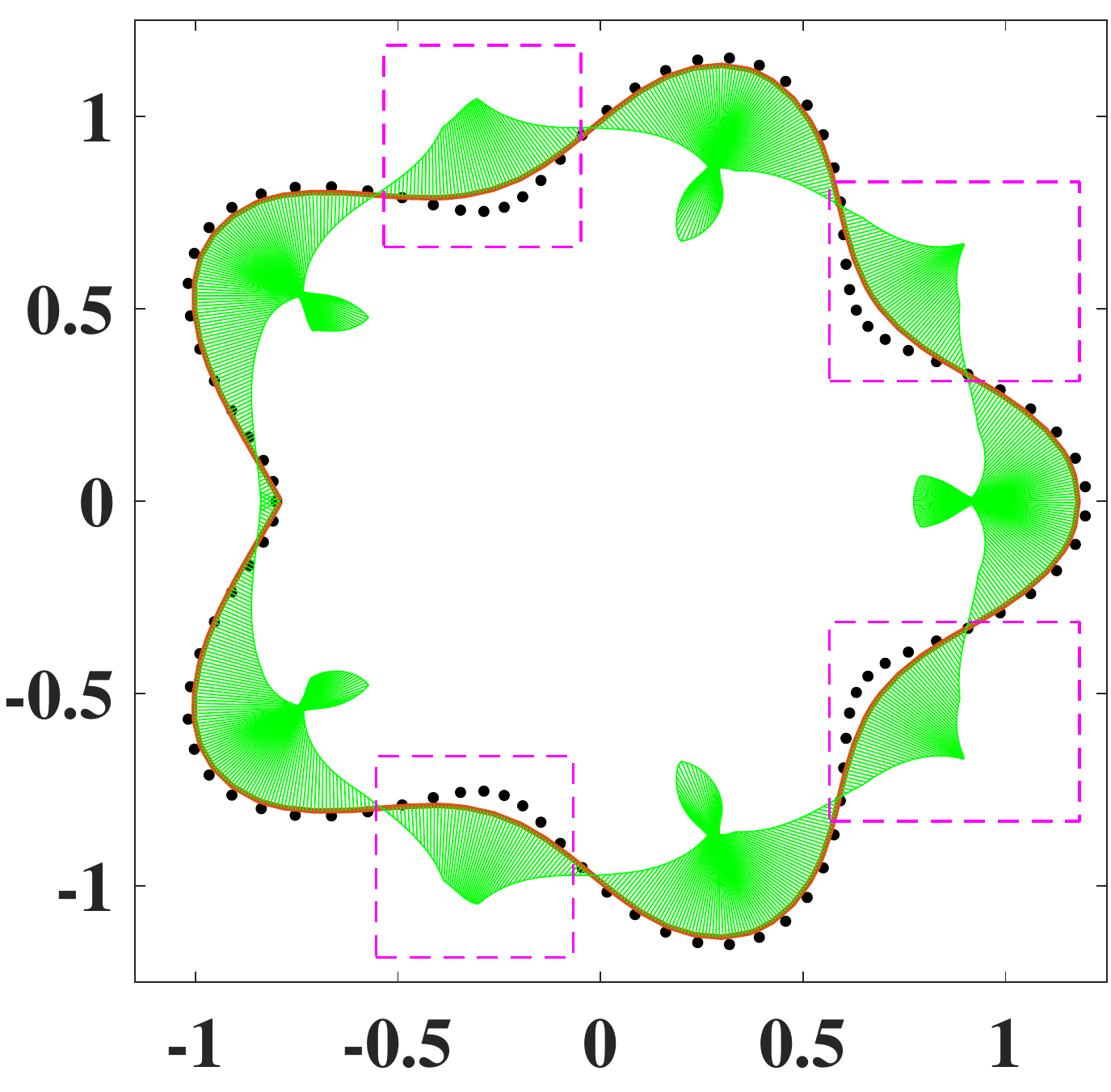}
    \label{fig:Star_FPIA_knotNoInsert_comb}
  }
  \subfigure[]{
    \includegraphics[width=2.0in]{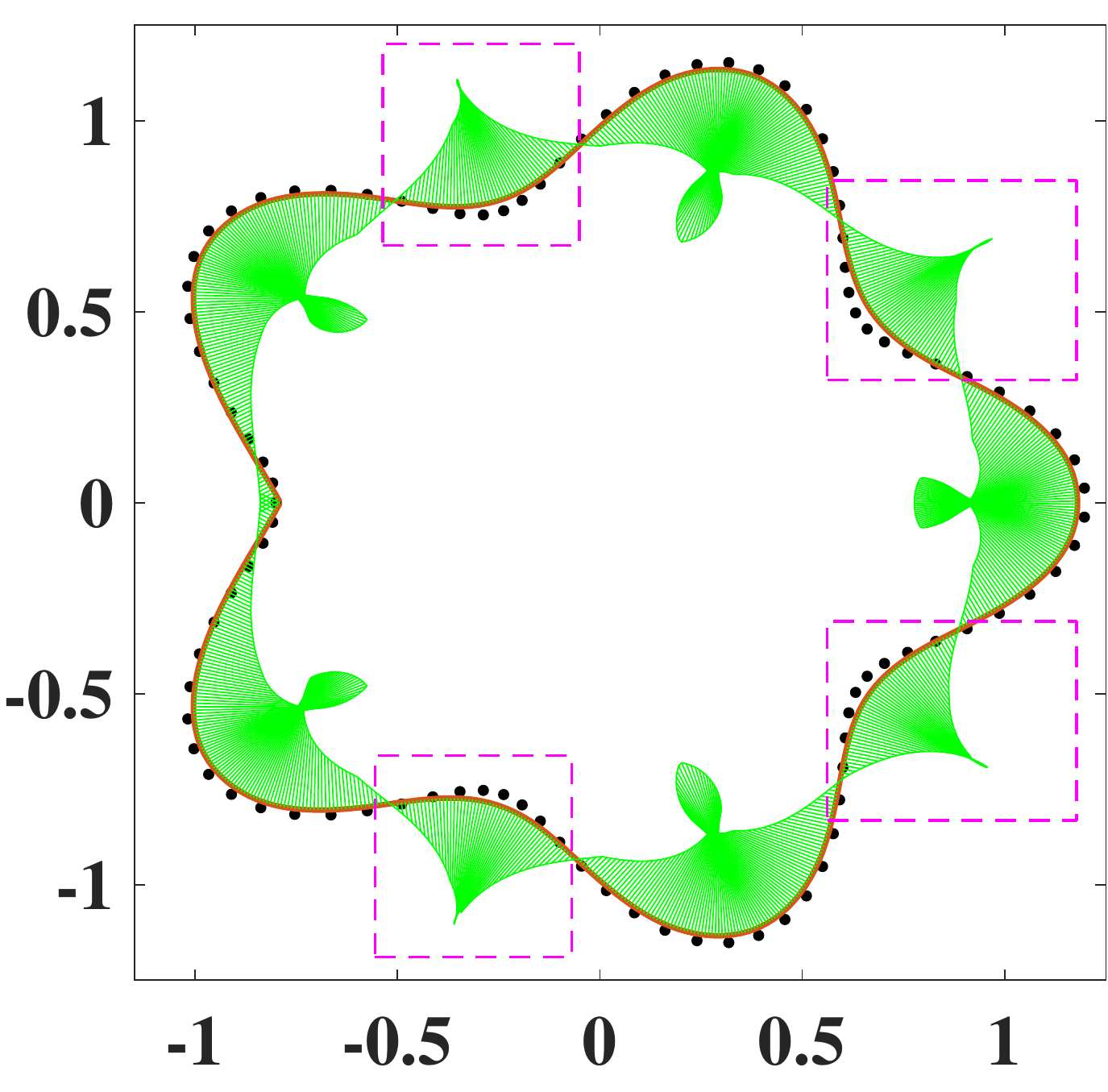}
    \label{fig:Star_FPIA_knotInsert_comb}
  }
  \caption{ The \textit{Starfish} example.
    (a) The fitting curve and its curvature comb;
    (b) The fairing curve generated by fairing-PIA method
        and its curvature comb,
        where the fairness of the curve segments in the red boxes is improved, but the fitting errors are increased;
    (c) The fairing curve generated by fairing-PIA method by inserting knots and decreasing corresponding smoothing weights,
     where the fitting errors of the curve segments in the red boxes are reduced.}
    \label{fig:Star_comb}
\end{figure*}

 In practice, the fairness of a curve is usually adjusted progressively and locally,
    i.e., segment by segment,
    using local fairing methods.
 In the example of the \textit{Airfoil} model,
    we demonstrate the segment-by-segment fairing and local fairing capabilities of the fairing-PIA method.
 The fitting curve with curvature comb is illustrated in Fig.~\ref{fig:Airfoil_fit_comb},
    where the smoothing weights for all control points are set as $0$.
 We aim to improve the fairness of the three segments in the red boxes segment by segment.
 In Fig.~\ref{fig:Airfoil_FPIA_comb_1},
    we first improve the fairness of the curve segment in the lower right box 
    by setting the smoothing weights of the $4^{th}$ to $7^{th}$ control points as $1 \times 10^{-5}$.
 The fairing-PIA is performed by taking the control points of the fitting curve 
    in Fig.~\ref{fig:Airfoil_fit_comb} as the initial control points.
 The fairness of the curve segment in the lower right box is improved 
    after the fairing-PIA iteration stops (Fig.~\ref{fig:Airfoil_FPIA_comb_1}).
 Second, Fig.~\ref{fig:Airfoil_FPIA_comb_2} illustrates that the fairness of the curve segment in the left box is improved
    using fairing-PIA by setting the smoothing weights of the $8^{th}$ to $11^{th}$ control points as $1 \times 10^{-5}$
    and taking the control points of curve in Fig.~\ref{fig:Airfoil_FPIA_comb_1} as the initial control points.
 Finally, the fairness of the curve segment in the upper right box can be improved similarly 
    using fairing-PIA by setting the smoothing weights of the $19^{th}$ to $23^{rd}$ as
    $1 \times 10^{-5}$ and taking the control points of the curve in Fig.~\ref{fig:Airfoil_FPIA_comb_2} as the initial control points (refer to Fig.~\ref{fig:Airfoil_FPIA_comb_3}).

 Moreover, the flexibility of fairing-PIA is demonstrated in the example of
    the \textit{Starfish} model (Fig.~\ref{fig:Star_comb}),
    where the new knots are inserted, and the smoothing weights are renewed in the iterations of fairing-PIA.
 In Fig.~\ref{fig:Star_fit_comb},
    the fitting curve with the curvature comb is presented, 
    where the fairness of the curve segments in the red boxes is undesirable and needs to be improved.
 The fairness of the curve segments in the red boxes is improved (Fig.~\ref{fig:Star_FPIA_knotNoInsert_comb})
    with fairing-PIA by setting the smoothing weights of the control points as $3\times 10^{-5}$,
    thereby affecting the curve segments in the red boxes.
 The smoothing weights of the other control points are set as $1 \times 10^{-5}$.
 However, the fitting errors of the curve segments in the red boxes increase.
 We insert two knots at each curve segment in the red box 
    to decrease the fitting error. 
 We also reduce the smoothing weights of the control points to $6\times 10^{-6}$,
    thereby affecting the curve segments in the red boxes.
 Then, a new round of fairing-PIA iterations is invoked 
    by taking the curve generated after the last round of fairing-PIA iterations as the initial curve (refer to Fig.~\ref{fig:Star_FPIA_knotNoInsert_comb}).
 After the new round of iterations stops,
    a new curve is produced (Fig.~\ref{fig:Star_FPIA_knotInsert_comb}),
    where the fitting errors of the curve segments in the red boxes are reduced.
 Moreover, the fairness of these curve segments is comparable with the fairness of the curve in Fig.~\ref{fig:Star_FPIA_knotNoInsert_comb}.
 The diagrams of iteration v.s. relative iterative error,
    iteration v.s. relative fitting error,
    and iteration v.s. the relative strain energy of the two rounds of iterations are shown in Fig.~\ref{fig:Star_ErrandEnergyPlot},
    where the fitting error is progressively decreased in the second round of iterations.

 Additionally, the statistics of the above three examples are listed in
    Table~\ref{TABLE:iteration_data},
    including the number of data points, number of control points,
    number of iterations, and running time.
 The absolute fitting error and absolute strain energy 
    are given in Table~\ref{TABLE:FitErrorandEnergy}.

 \begin{figure*}[!htb]
   \centering
  \subfigure[]{
    \includegraphics[width=2.6in]{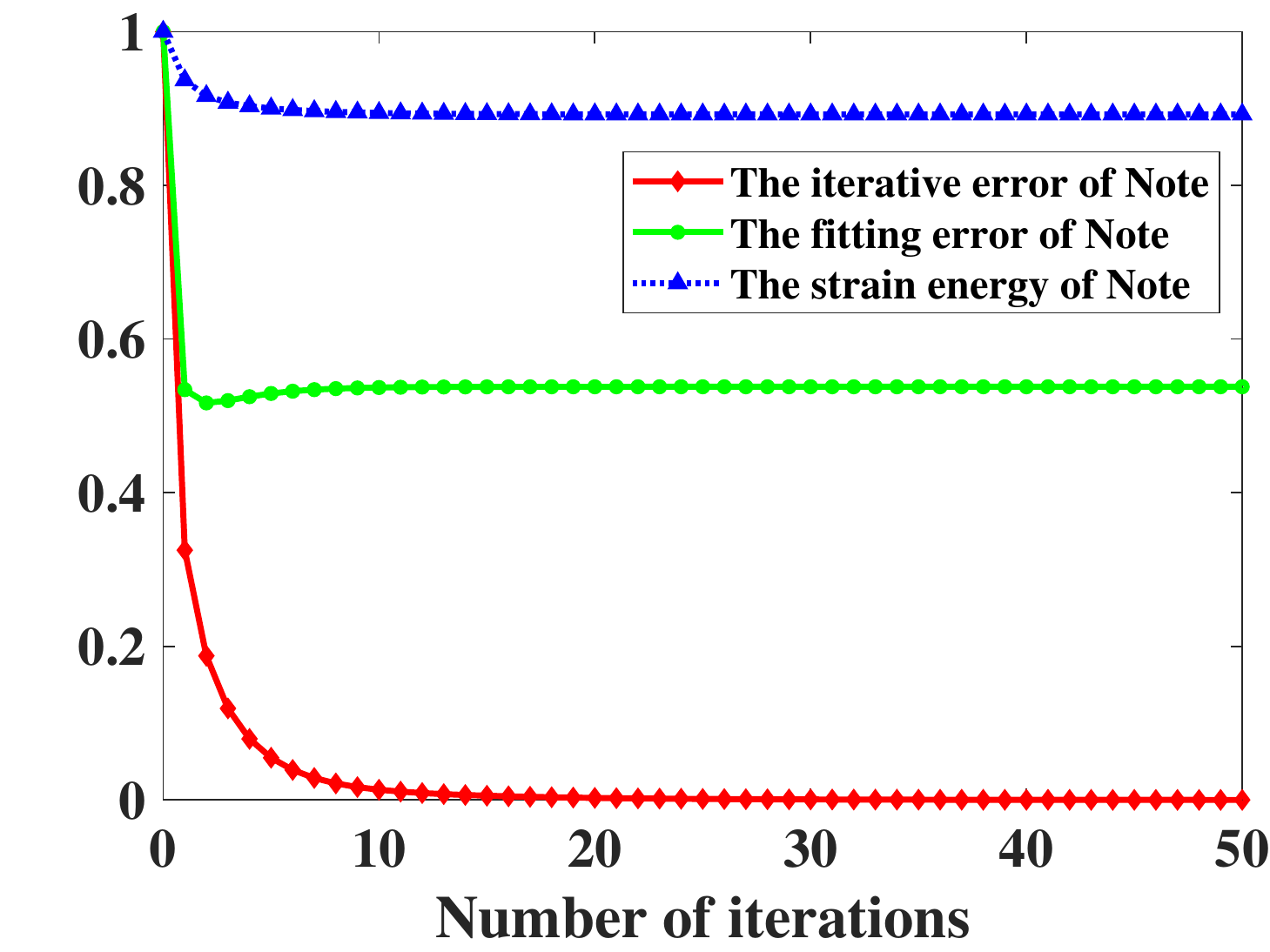}
    \label{fig:Note_ErrandEnergyPlot}
  }
  \subfigure[]{
    \includegraphics[width=2.6in]{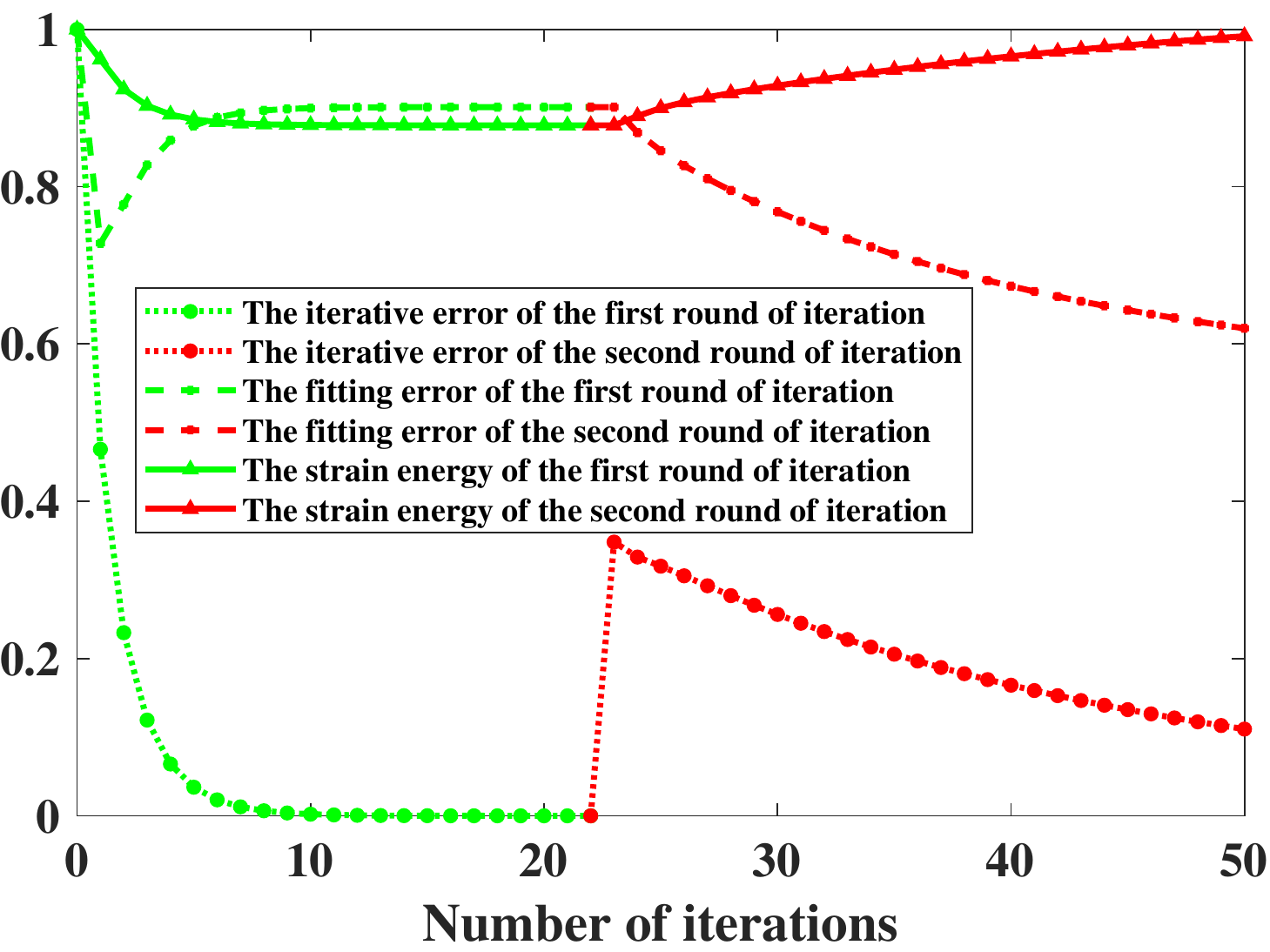}
    \label{fig:Star_ErrandEnergyPlot}
  }
   \caption{ The diagrams of iterations v.s. relative fitting error, relative iterative error and relative strain energy, respectively, for
   \textit{Note} model (a), and
   \textit{Starfish} model (b).
   }
   \label{fig:CurvePlots}
 \end{figure*}

\begin{table*}[!htb] 
	\centering
  \caption{Statistical data of fairing-PIA on curve and surface examples.}
  \begin{threeparttable}
  \begin{tabular}{llllll}
    \toprule
    Model &
    \# Data points\tnote{a} &
    \# Control points\tnote{b}&
    \# Iterations\tnote{c} &
    \multicolumn{2}{c}{Time(seconds)}
    \\
    \cmidrule(l){5-6} 
    &  &  & & \makecell[c]{Energy\\minimization\\ method} & Fairing-PIA  \\
    \midrule
    \textit{Note}    & 245 & 83 &  64 & 0.1509 & $\mathbf{0.1035}$ \\
    \textit{Airfoil} (The $1^{st}$ segment) & 49  & 25 & 32 & 0.01144 & $\mathbf{0.01032}$ \\
    \textit{Starfish} (The first round) & 100 & 34 & 23 & 0.02133 & $\mathbf{0.01911}$  \\
    \textit{Tooth}    & $81\times61$ & $32\times24$ & 248 & 39.3429 & $\mathbf{35.4474}$   \\
    \textit{Fan\_disk} & $41\times61$ & $16\times24$ & 85 & 7.4511 & $\mathbf{7.2305}$   \\
    \textit{Mannequin} & $121\times161$ & $48\times64$ & 143 & $\mathbf{1554.949}$6 & 1560.6648  \\
    \bottomrule
  \end{tabular}
  \begin{tablenotes}
    \item[a] Number of the data points.
    \item[b] Number of the control points.
    \item[c] Number of iterations.
  \end{tablenotes}
\end{threeparttable}
  \label{TABLE:iteration_data}
\end{table*}

\begin{table*}
  \centering
  \caption
  {
    The absolute fitting error and absolute energy of curve and surface examples.
  }
  \begin{threeparttable}
  \begin{tabular}{lllll}
   \toprule
    Model & Initial fitting error & Final fitting error\tnote{a} & Initial energy & Final energy\tnote{b}\\
   \midrule
    \textit{Note} & 0.01074 & $\mathbf{0.005778}$ & 90438.9606 & $\mathbf{80448.1437}$ \\
    \textit{Starfish} (The first round)     & 0.02725  & $\mathbf{0.02456}$   & 15090.0488 & $\mathbf{13275.4486}$ \\
    \textit{Viviani}\tnote{c} ($r=1$) & 0.02895 & $\mathbf{0.02210}$ & 1455.8978 & $\mathbf{1442.3715}$\\
    \textit{Viviani} ($r=2$) &  0.02895 & $\mathbf{0.01498}$ & $2.4161\times{10^5}$  & $\mathbf{1.9570\times{10^5}}$ \\
    \textit{Viviani} ($r=3$) &  0.02895 & $\mathbf{0.01712}$ & $7.7733\times{10^8}$  & $\mathbf{3.8296\times{10^7}}$ \\
    \textit{Tooth}      & 0.006903 & $\mathbf{0.006180}$ & 236.9693 & $\mathbf{70.3654}$ \\
    \textit{Fan\_disk}   & 0.01002  & $\mathbf{0.005743}$ & 73.3475  & $\mathbf{15.9637}$ \\
    \textit{Mannequin}  & 0.003143 & $\mathbf{0.001457}$ & 787.0678 & $\mathbf{350.0640}$ \\
   \bottomrule
  \end{tabular}
  \begin{tablenotes}
    \item[a] The fitting error obtained after the iteration stops.
    \item[b] The energy obtained after the iteration stops.
    \item[c] The stretch energy, strain energy and jerk energy correspond to $r=1, 2, 3$, respectively.
  \end{tablenotes}
  \end{threeparttable}
  \label{TABLE:FitErrorandEnergy}
\end{table*}

\begin{figure*}[!htb]
  \centering
  \subfigure{
    \rotatebox{90}{~~~~~~~~~~~~~~~~$r=1$}
    \begin{minipage}[t]{0.2\linewidth}
      \centering
      \includegraphics[width=0.7in]{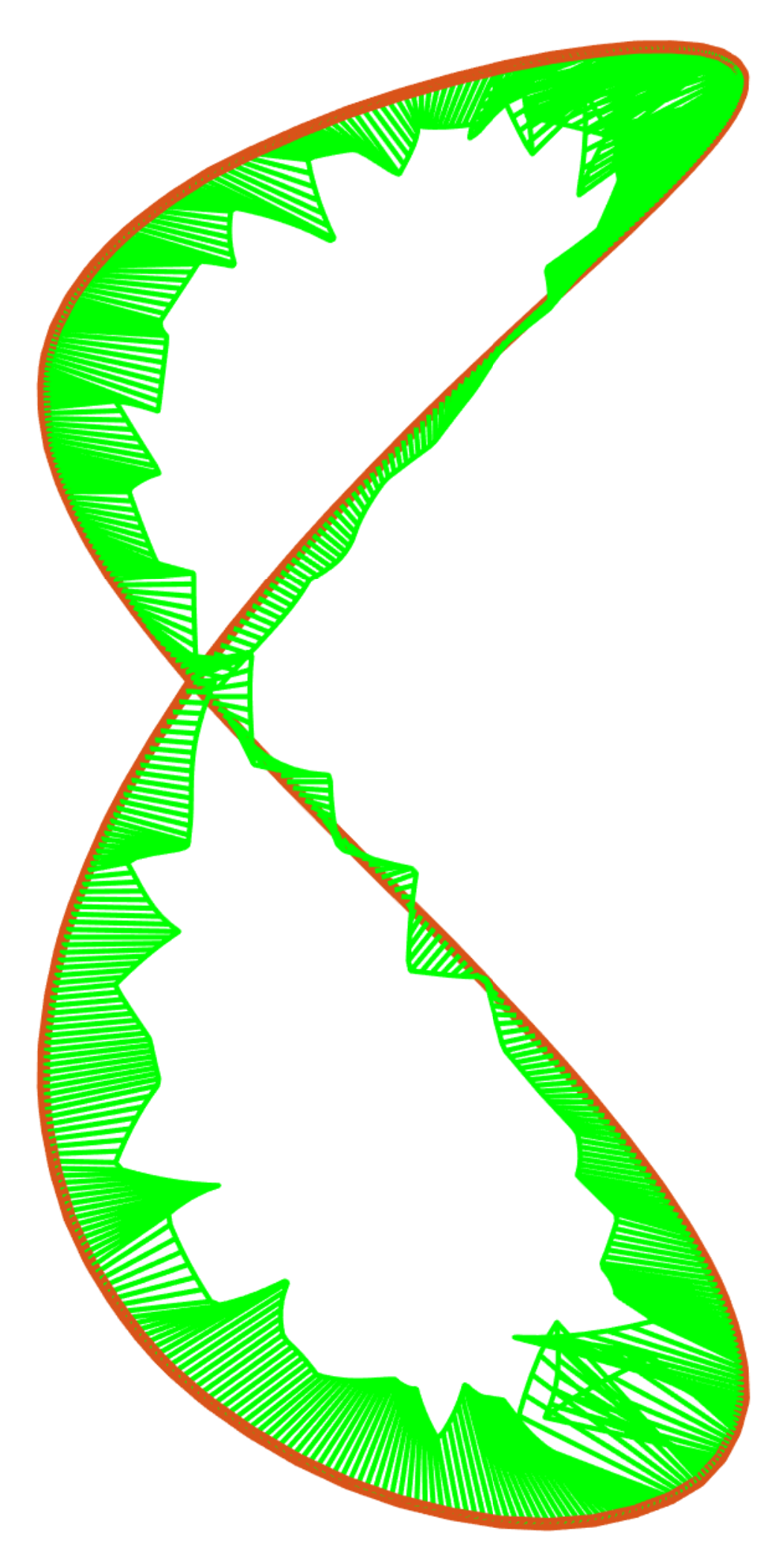}
      \label{fig:Vivi_r1_iter0_comb}
    \end{minipage}
  }
  \hspace{-5mm}
  \subfigure{
    \begin{minipage}[t]{0.2\linewidth}
      \centering
      \includegraphics[width=0.7in]{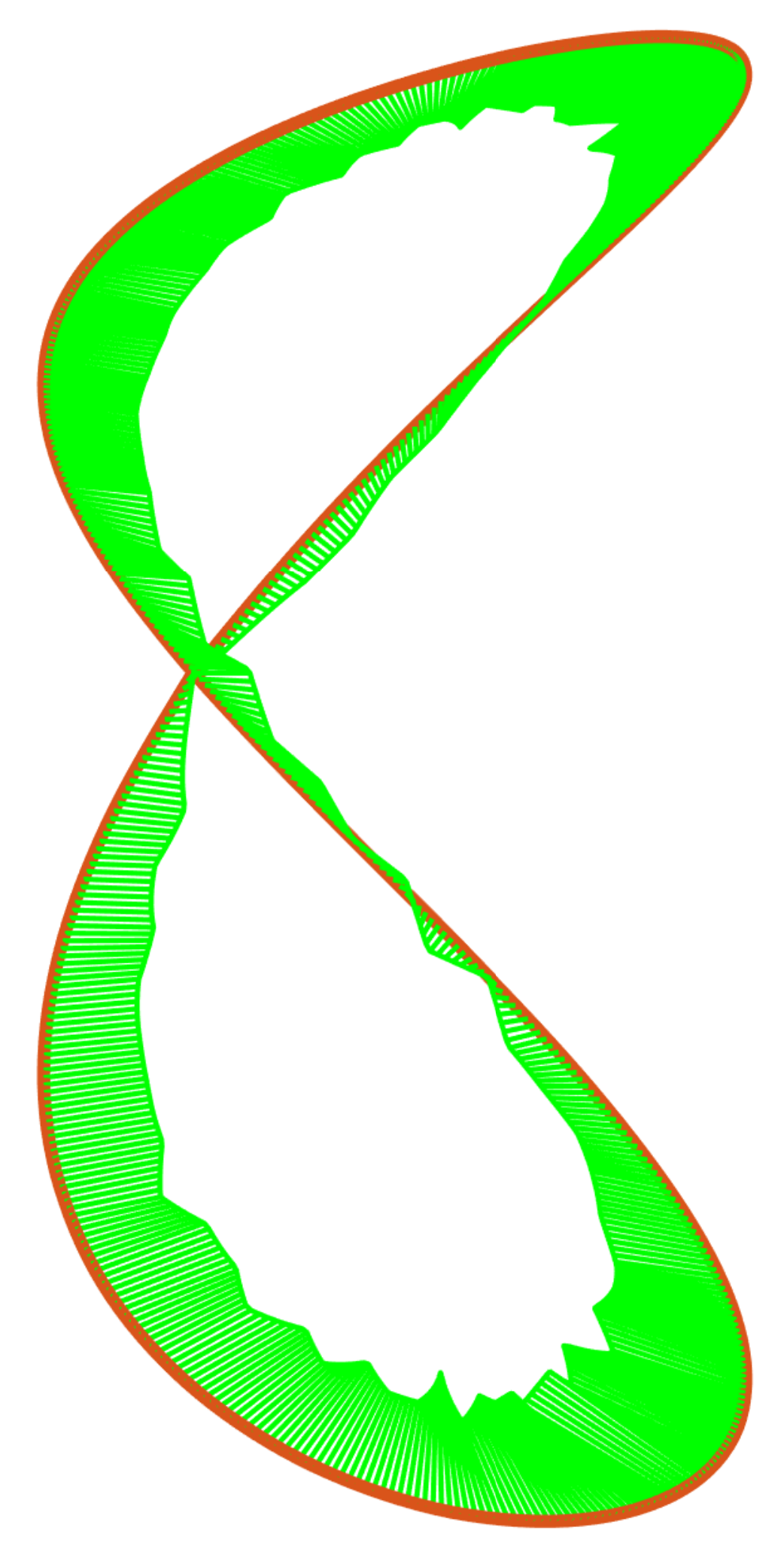}
      \label{fig:Vivi_r1_iter5_comb}
    \end{minipage}
  }
  \hspace{-5mm}
  \subfigure{
    \begin{minipage}[t]{0.2\linewidth}
      \centering
      \includegraphics[width=0.7in]{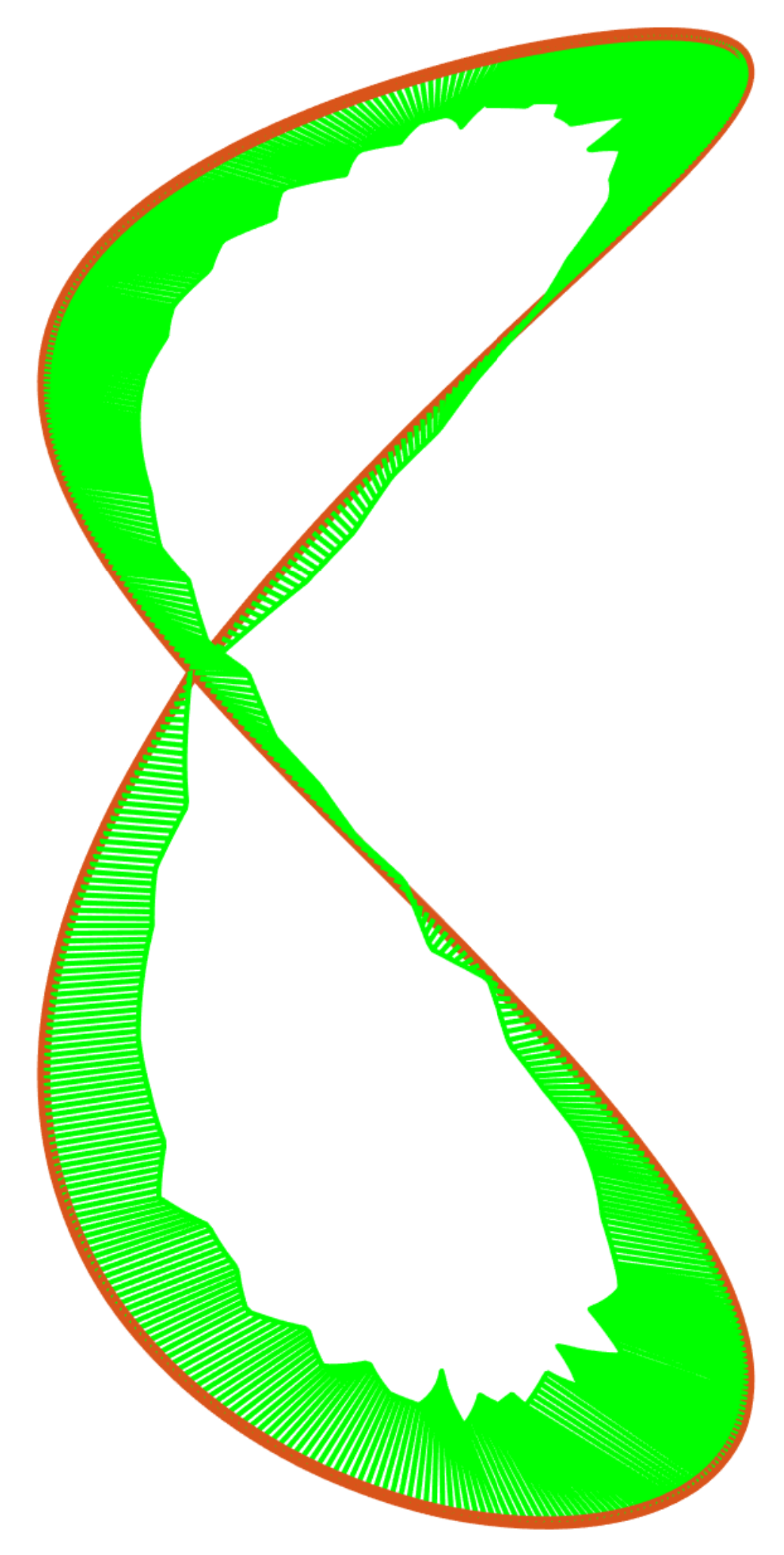}
      \label{fig:Vivi_r1_iter15_comb}
    \end{minipage}
  }
  \hspace{-5mm}
  \subfigure{
    \begin{minipage}[t]{0.2\linewidth}
      \centering
      \includegraphics[width=0.7in]{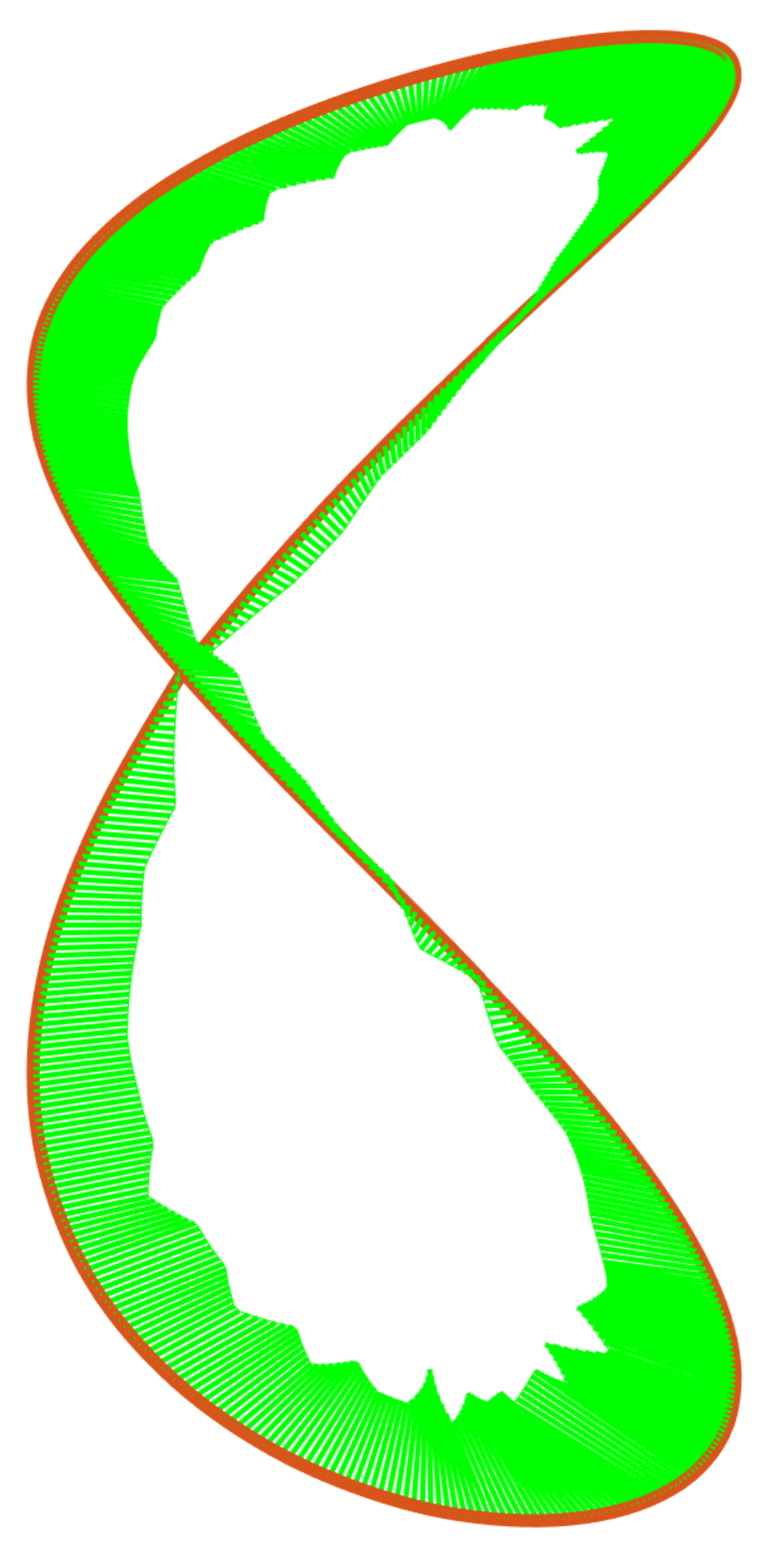}
      \label{fig:Vivi_r1_fairing_comb}
    \end{minipage}
  }
  \vspace{-4mm}

  \setcounter{subfigure}{0}
  \subfigure{
    \rotatebox{90}{~~~~~~~~~~~~~~~~$r=2$}
    \begin{minipage}[t]{0.2\linewidth}
      \centering
      \includegraphics[width=0.7in]{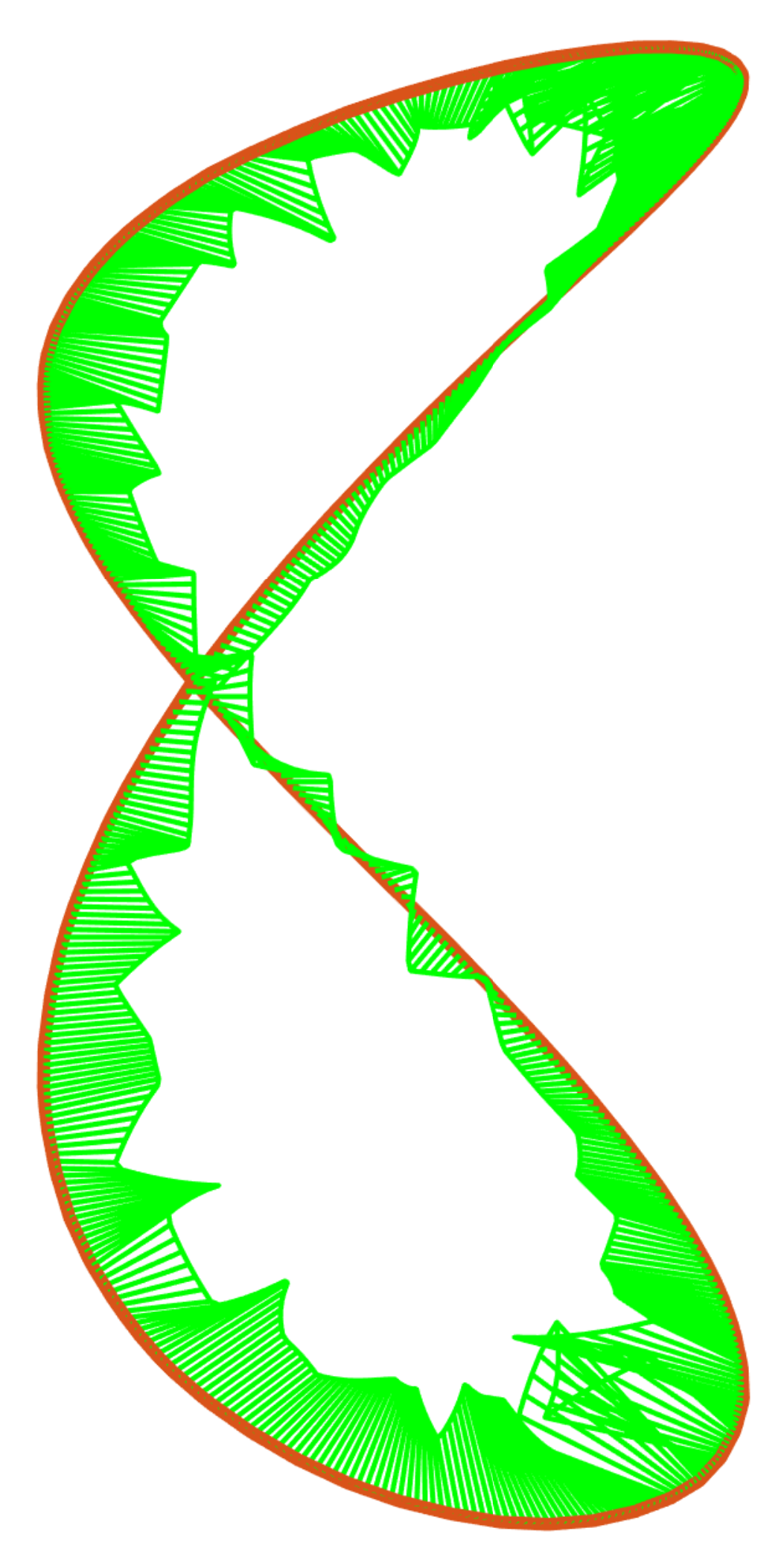}
      \label{fig:Vivi_r2_iter0_comb}
    \end{minipage}
  }
  \hspace{-5mm}
  \subfigure{
    \begin{minipage}[t]{0.2\linewidth}
      \centering
      \includegraphics[width=0.7in]{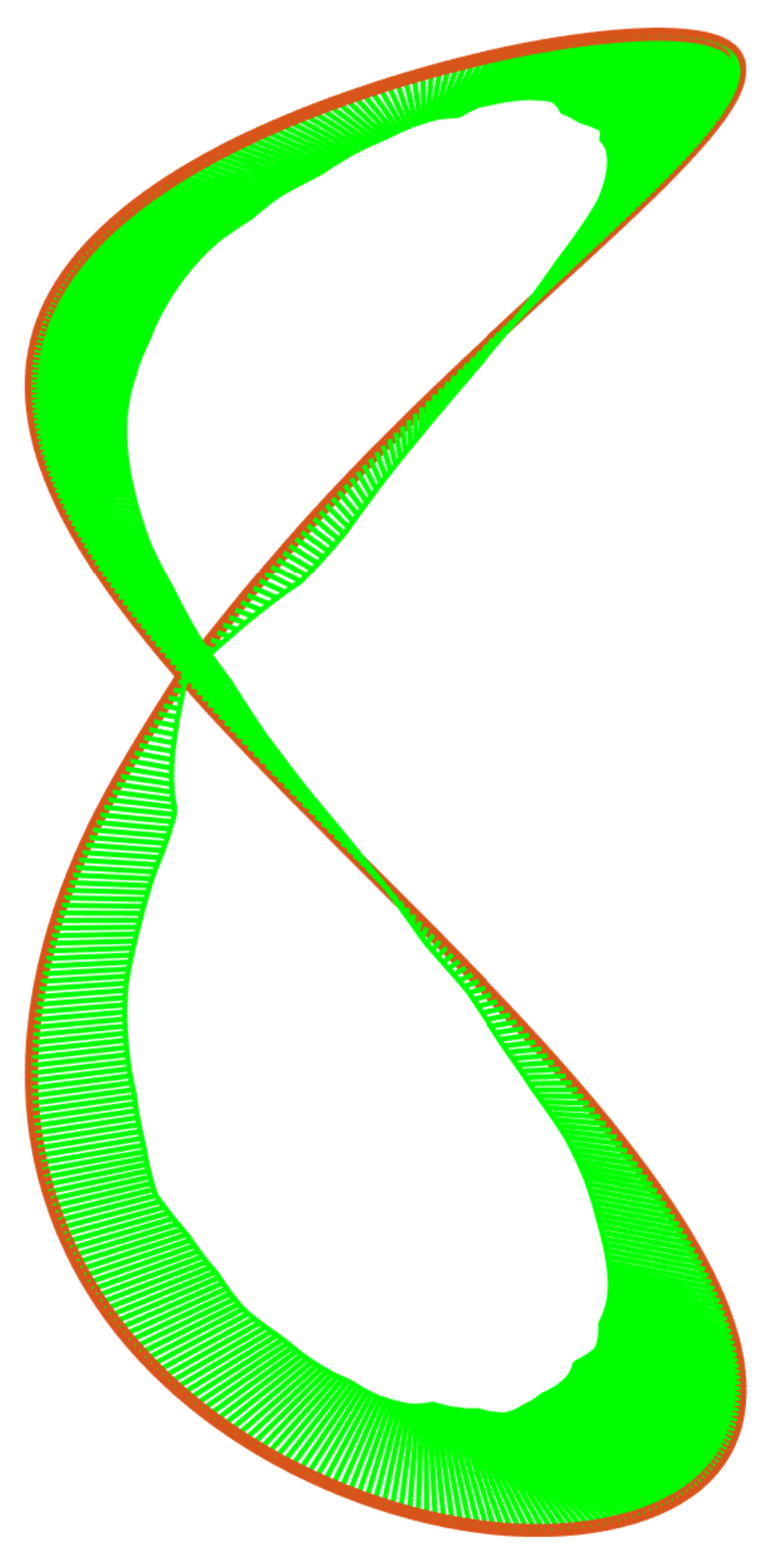}
      \label{fig:Vivi_r2_iter5_comb}
    \end{minipage}
  }
  \hspace{-5mm}
  \subfigure{
    \begin{minipage}[t]{0.2\linewidth}
      \centering
      \includegraphics[width=0.7in]{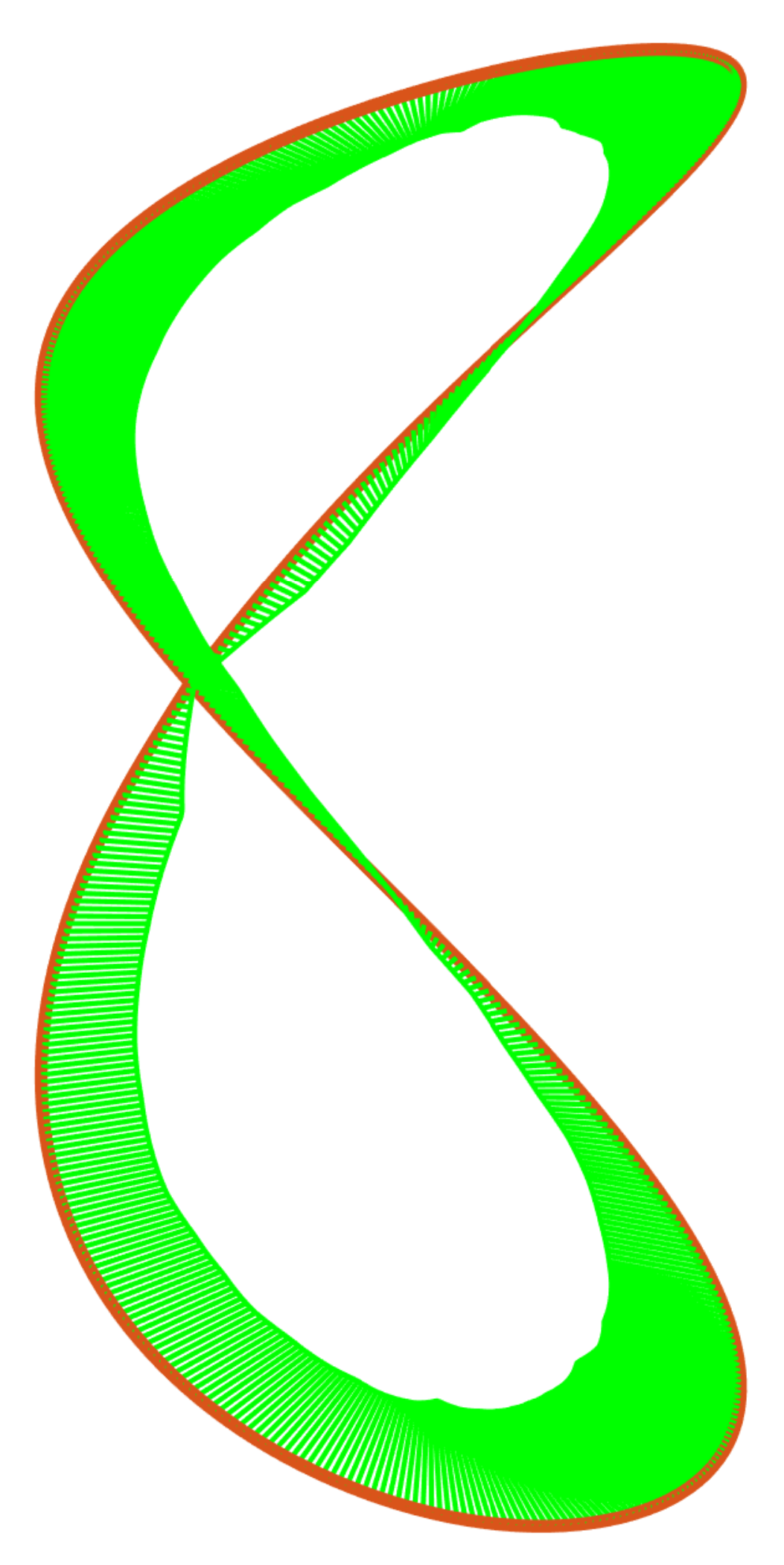}
      \label{fig:Vivi_r2_iter15_comb}
    \end{minipage}
  }
  \hspace{-5mm}
  \subfigure{
    \begin{minipage}[t]{0.2\linewidth}
      \centering
      \includegraphics[width=0.7in]{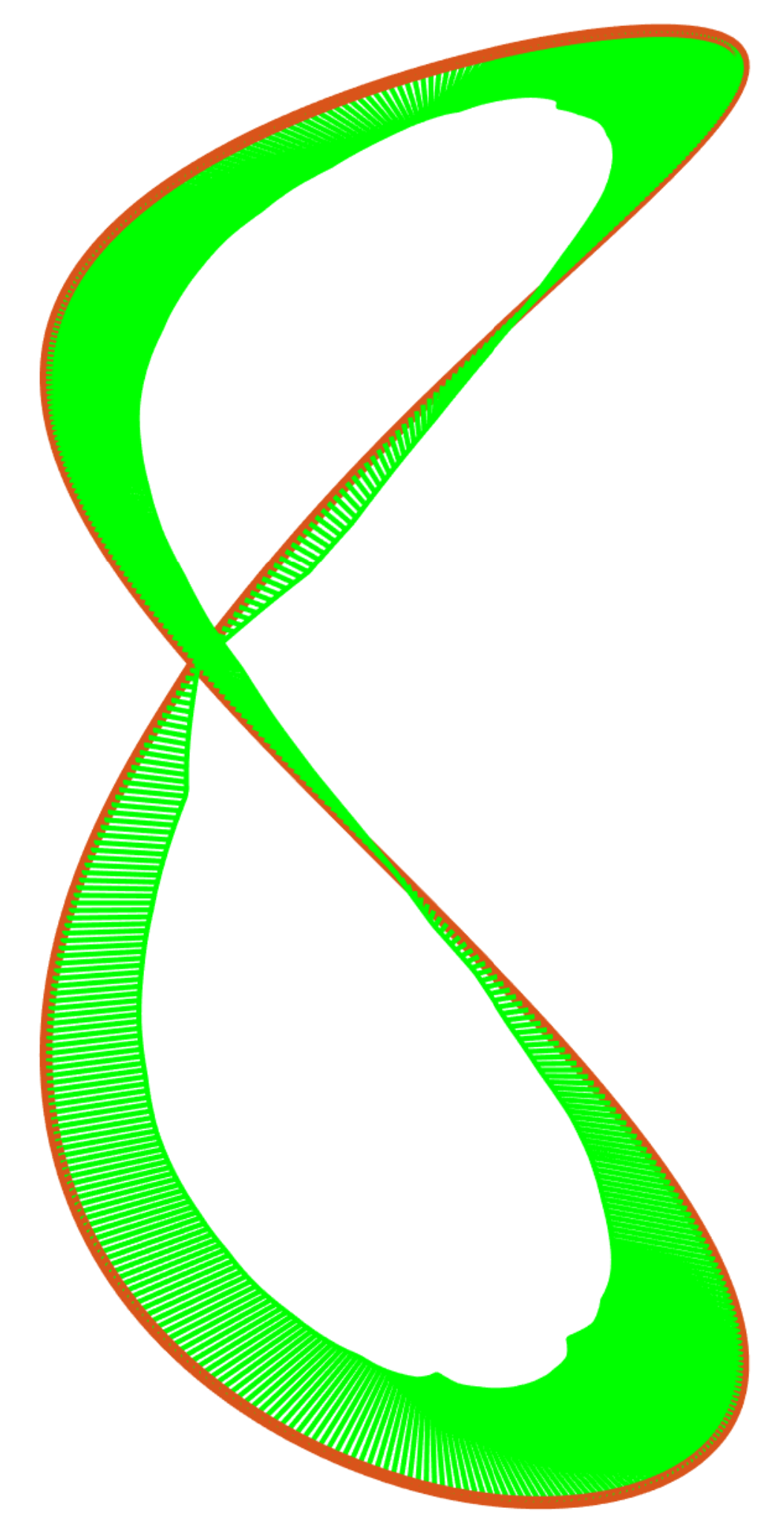}
      \label{fig:Vivi_r2_fairing_comb}
    \end{minipage}
  }
  \vspace{-4mm}

  \setcounter{subfigure}{0}
  \subfigure[Iteration 0]{
    \rotatebox{90}{~~~~~~~~~~~~~~~~$r=3$}
    \begin{minipage}[t]{0.2\linewidth}
      \centering
      \includegraphics[width=0.7in]{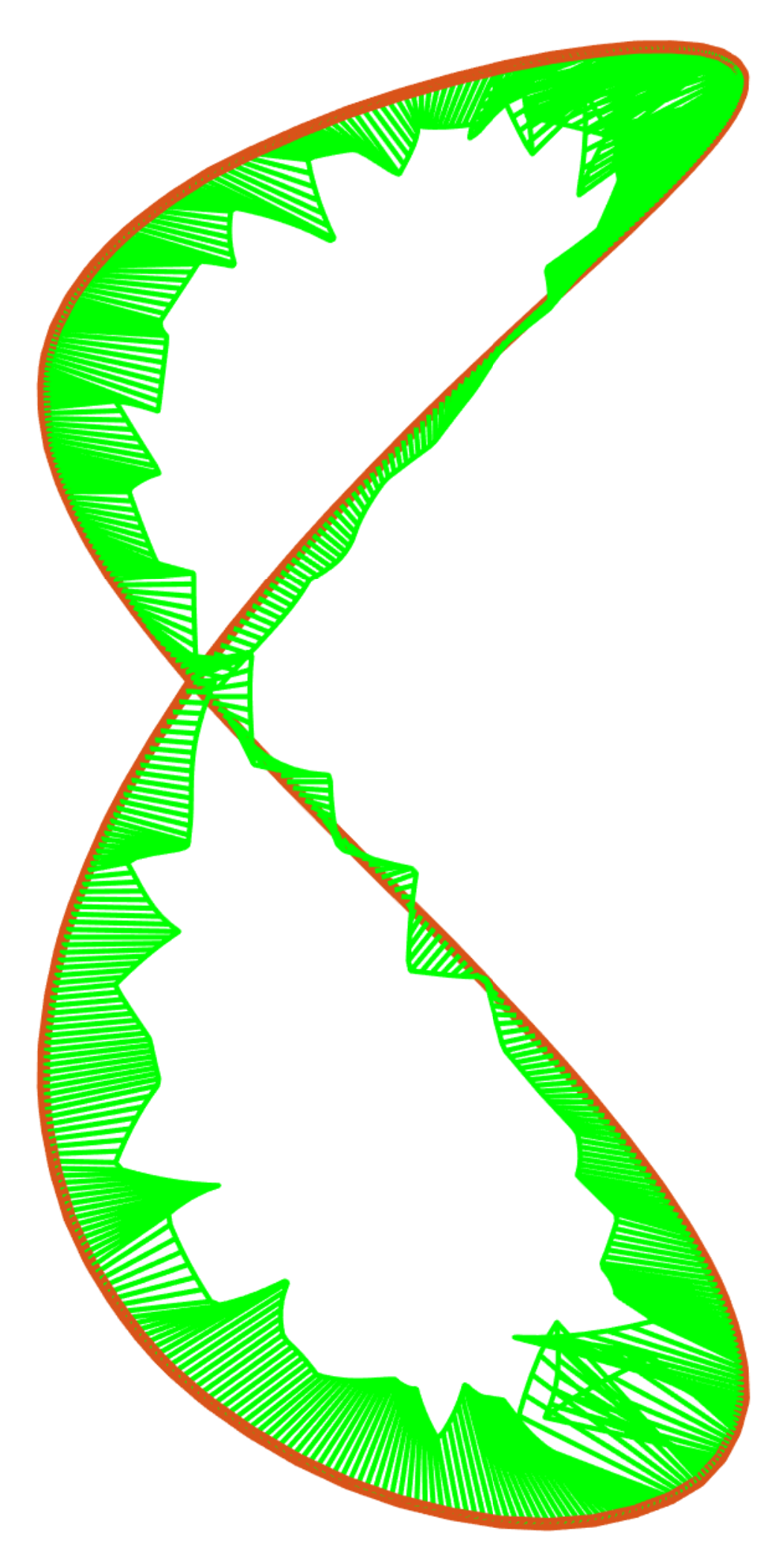}
      \label{fig:Vivi_r3_iter0_comb}
    \end{minipage}
  }
  \hspace{-5mm}
  \subfigure[Iteration 5]{
    \begin{minipage}[t]{0.2\linewidth}
      \centering
      \includegraphics[width=0.7in]{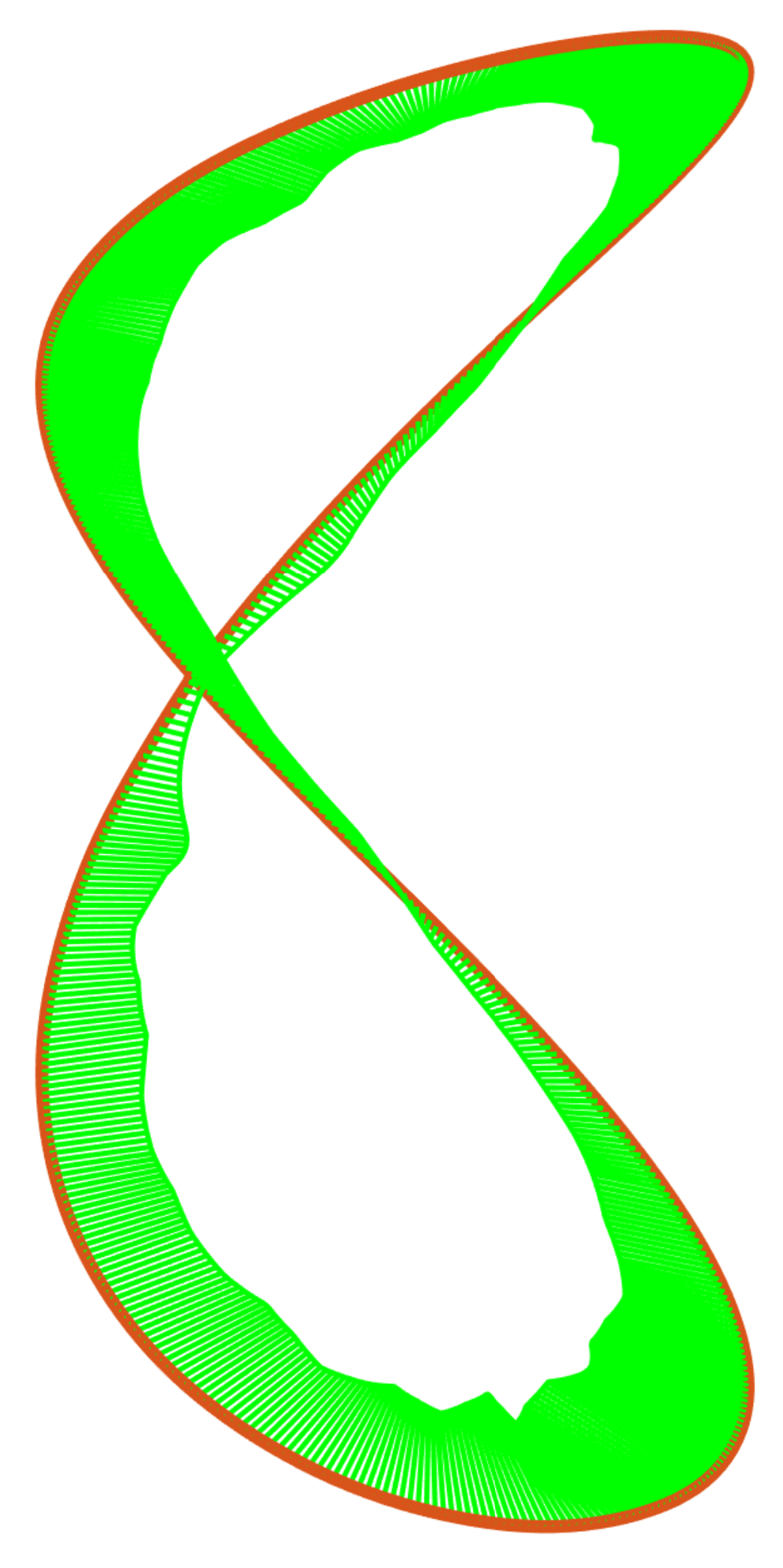}
      \label{fig:Vivi_r3_iter5_comb}
    \end{minipage}
  }
  \hspace{-5mm}
  \subfigure[Iteration 15]{
    \begin{minipage}[t]{0.2\linewidth}
      \centering
      \includegraphics[width=0.7in]{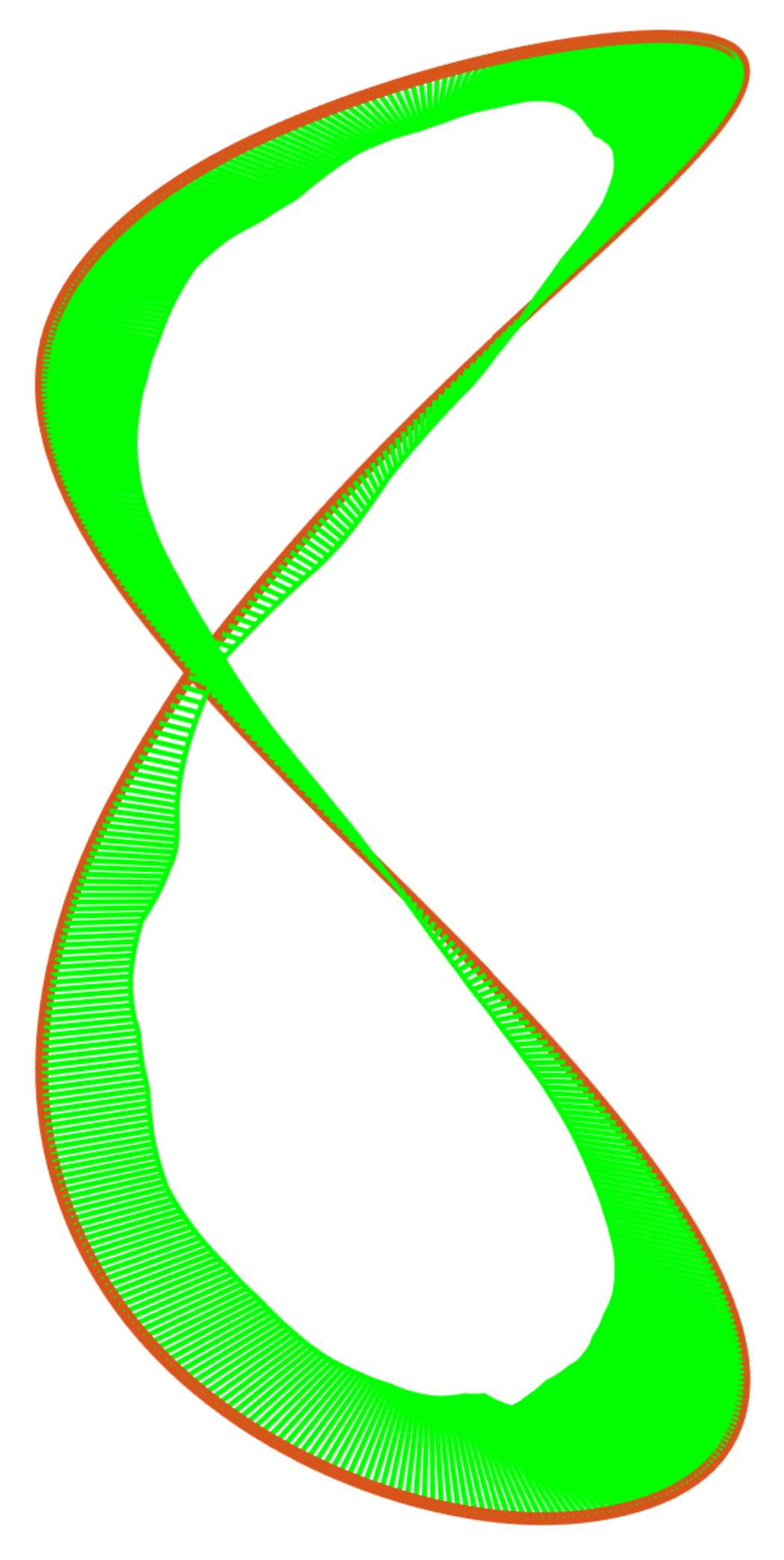}
      \label{fig:Vivi_r3_iter15_comb}
    \end{minipage}
  }
  \hspace{-5mm}
  \subfigure[Iteration stops]{
    \begin{minipage}[t]{0.2\linewidth}
      \centering
      \includegraphics[width=0.7in]{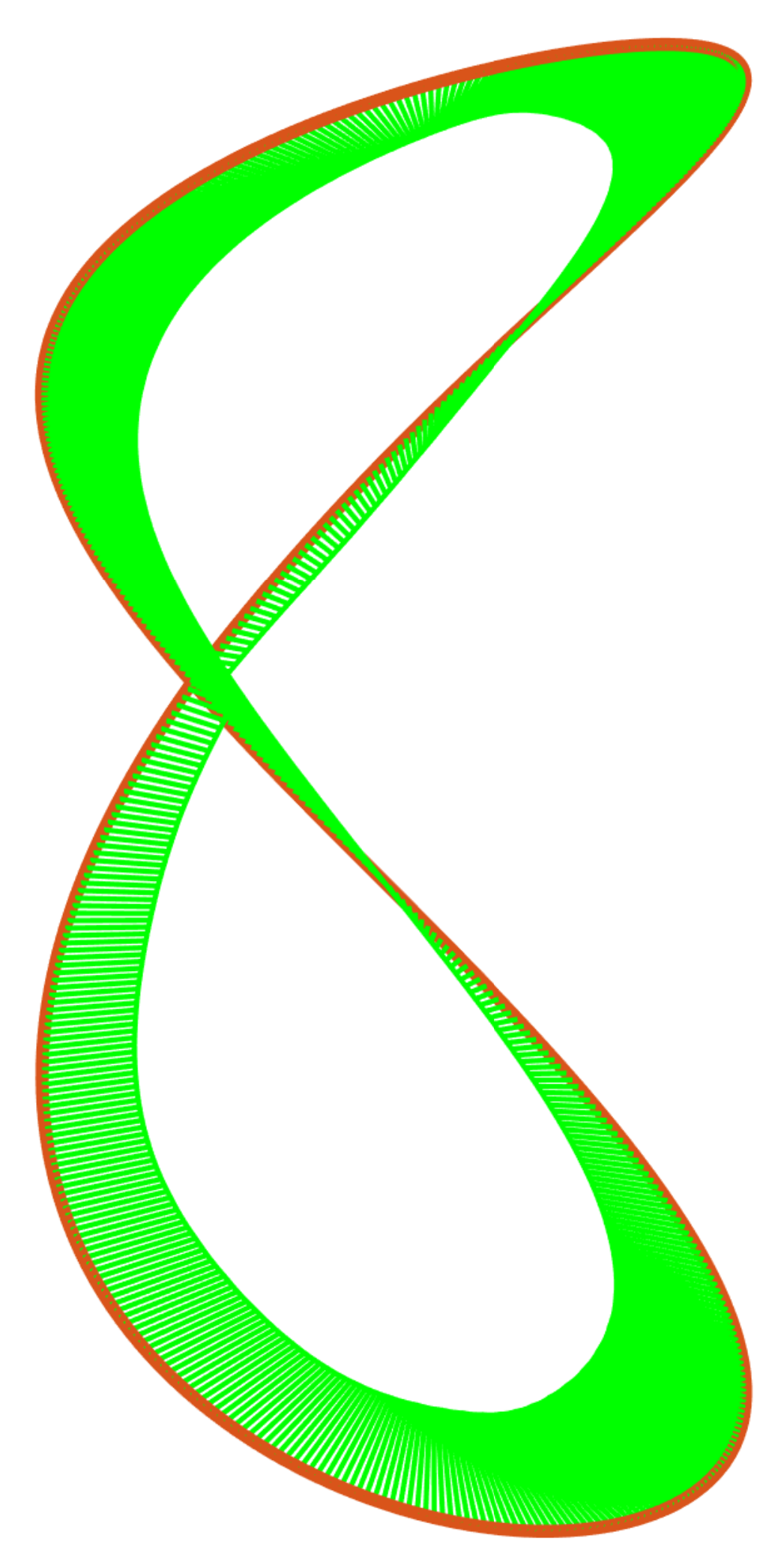}
      \label{fig:Vivi_r3_fairing_comb}
    \end{minipage}
  }
  \caption{
    The \textit{Viviani} example.
    Comparison of different selections of $r$
    (refer to Eq.~\eqref{eq:def_of_functional})
    in the fairing-PIA method.
    The fairing curve is in red,
        and the curvature comb is in green.
    }
  \label{fig:Vivi_Comb}
\end{figure*}

\begin{figure*}[!htb]
  \centering
  \subfigure[]{
    \includegraphics[width=2in]{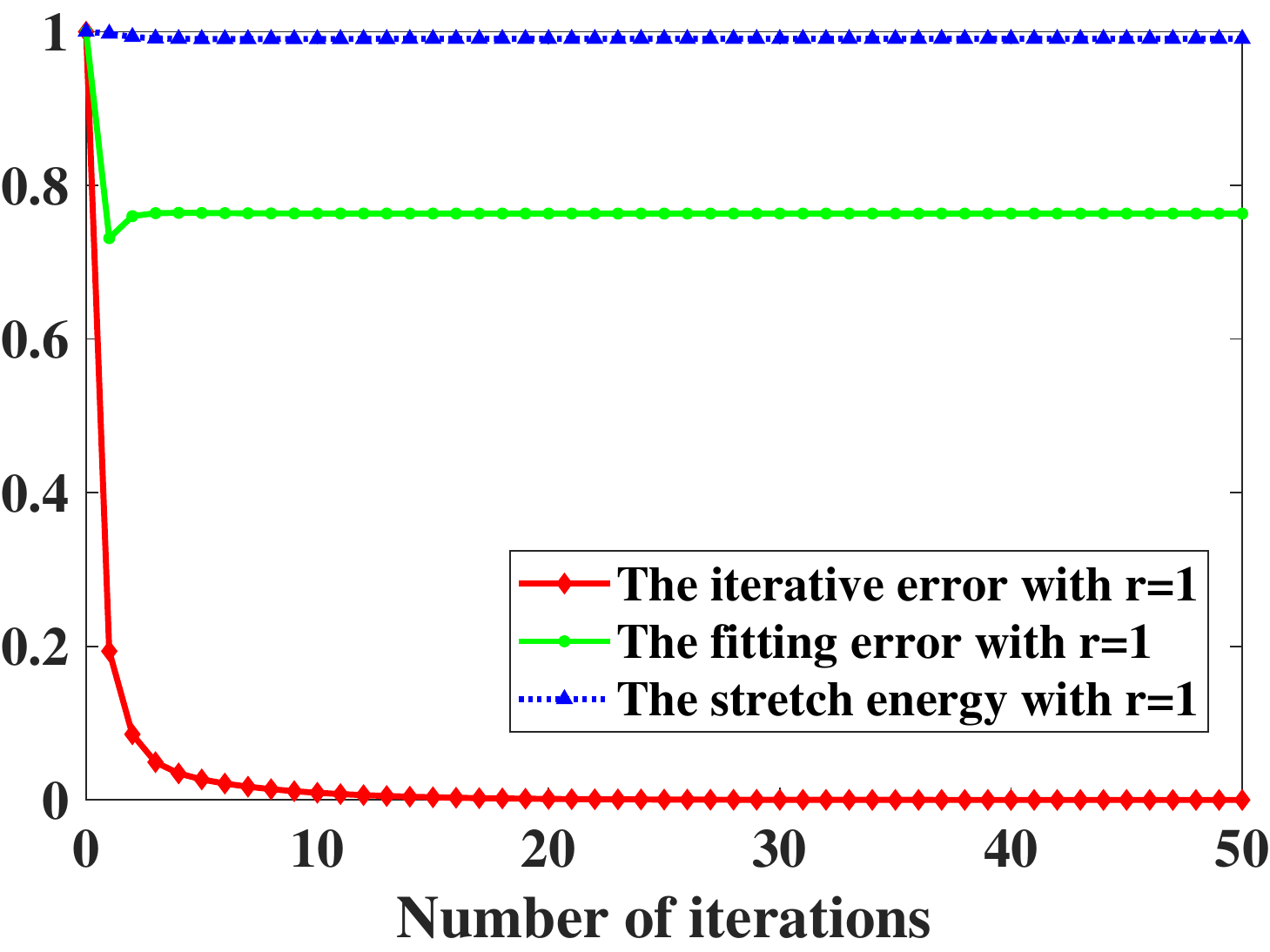}
    \label{fig:Vivi_ErrandEnergyPlot_r1}
  }
  \subfigure[]{
    \includegraphics[width=2in]{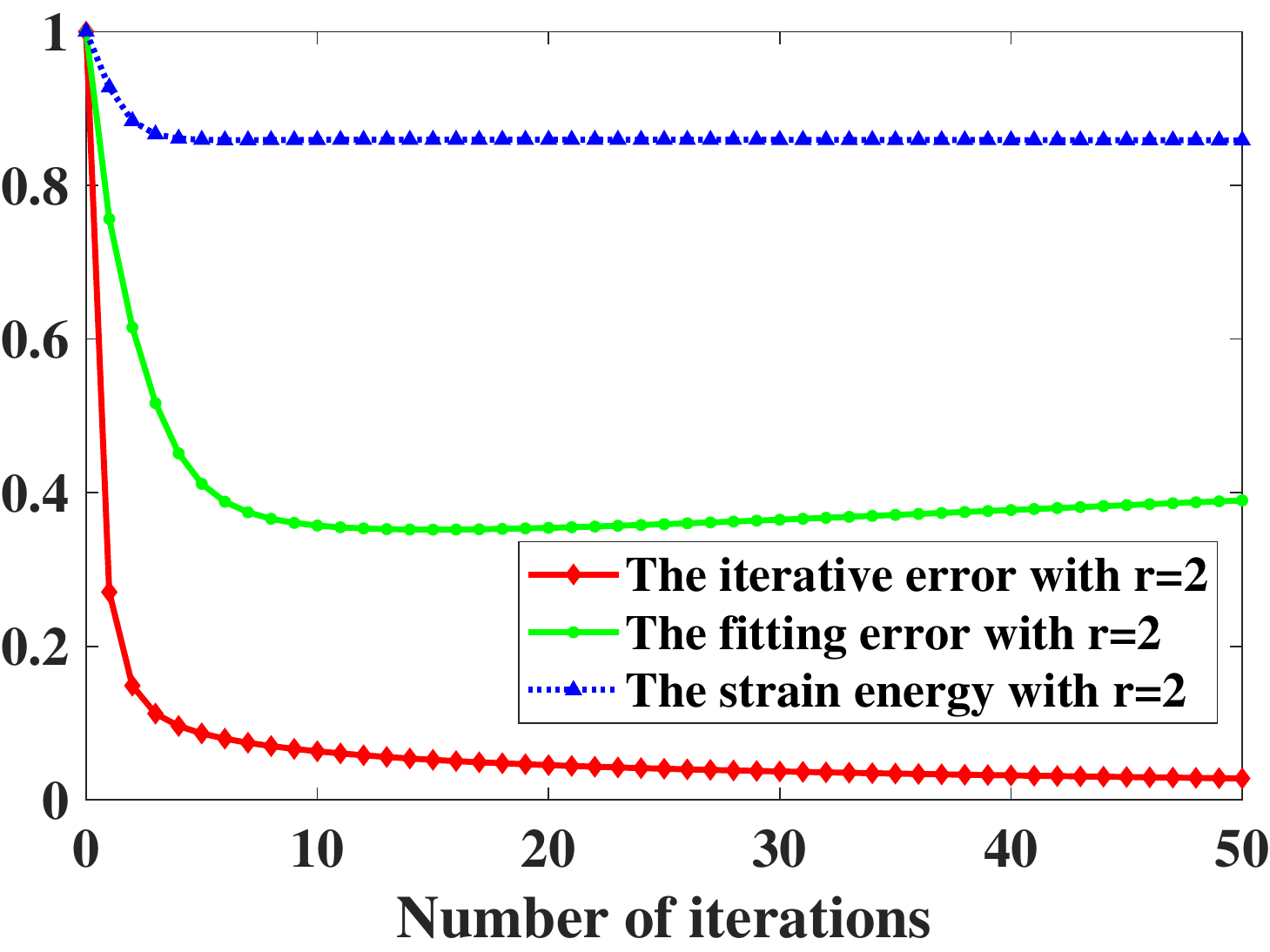}
    \label{fig:Vivi_ErrandEnergyPlot_r2}
  }
  \subfigure[]{
    \includegraphics[width=2in]{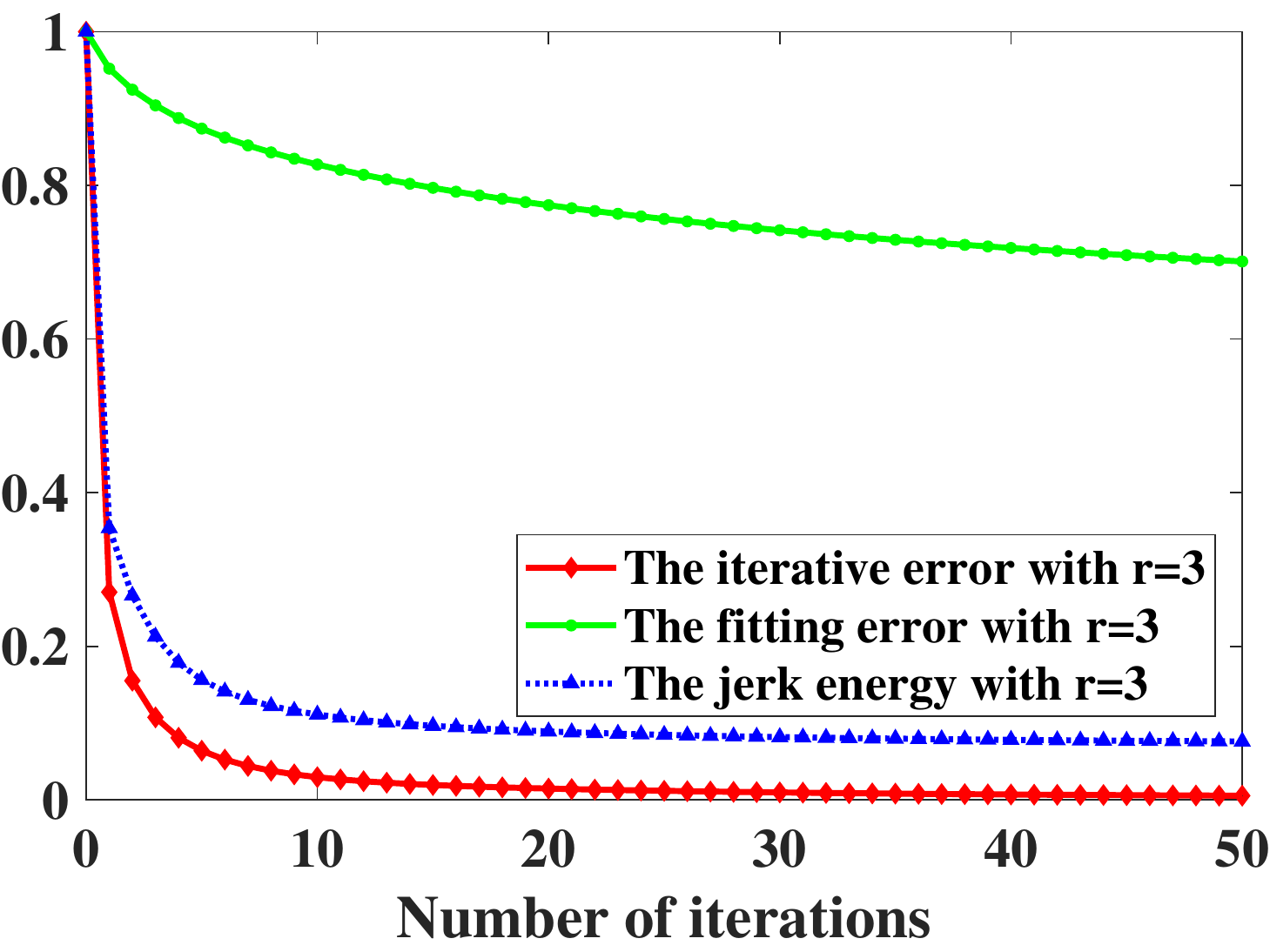}
    \label{fig:Vivi_ErrandEnergyPlot_r3}
  }
  \caption
  { 
    Diagrams of iteration v.s. relative fitting error, relative iterative error, relative energy of the \textit{Viviani} example for $r=1$ (a),
    $r=2$ (b), and $r=3$ (c). 
  }
  \label{fig:Vivi_Plot}
\end{figure*}

\subsection{Space Curves}

 In this section,
    an example of a space curve is presented to show the effects of the parameter $r$ in the fairing functional~\pref{eq:def_of_functional}.
 As stated in Remark~\ref{rmk2:equal_weight_energymeaning},
    different values of $r$, i.e., $r = 1, 2, 3$, 
    correspond to different energies, i.e., stretch energy, strain energy and jerk energy.
 In the space curve example (Fig.~\ref{fig:Vivi_Comb}),
    the $420$ data points are sampled uniformly from the Viviani curve,
    \begin{equation*}
          \left\{
          \begin{aligned}
             &x^2+y^2+z^2=25,\\
             &x^2+y^2-5x=0.
          \end{aligned}
          \right.
        \end{equation*}
 Moreover, we add Gaussian noise with $0$ mean and $0.005$ variance into
    the data set.
 Then, we select $85$ control points from the data points 
    as the initial control points to perform the fairing-PIA iterations with $r = 1,2,3$.
 The smoothing weights for the control points are assigned according to the following strategy.
 First, we calculate \emph{the absolute value of the second order difference}
    (ASOD) of each data point.
 The initial control points are chosen from the data points. 
 Thus, each control point has its own ASOD.
 The smoothing weights of the control points with the $20$ highest
    ASODs are set as $2\times10^{-4}$,
    and the rest are set as $1\times10^{-5}$ for $r=2$ and $3$.
 When $r=1$, the smoothing weights of the control points with the $20$ highest
    ASODs are set as $0.022$,
  and the rest are set as $0.02$.

 The fairing-PIA iteration procedures are illustrated in
    Fig.~\ref{fig:Vivi_Comb} for $r = 1,2,3$.
 As shown in Fig.~\ref{fig:Vivi_Comb},
    the fairness of curves is improved with iterations, 
    and the fairness for $r=2$ and $3$ is better than that for $r=1$.
 In addition,
    the diagrams of iteration v.s. relative iterative error, 
    relative fitting error,
    and relative energy are presented in Fig.~\ref{fig:Vivi_Plot}.
 Similar to those in the planar curve case, 
    the iterative error, fitting error, and energy rapidly decrease in the early iterations. 
 Moreover, the absolute fitting error and three types of absolute energies are shown in Table~\ref{TABLE:FitErrorandEnergy}.

\begin{figure*}[!htb]
  \centering
  \subfigure[]{
    \includegraphics[width=2.1in]{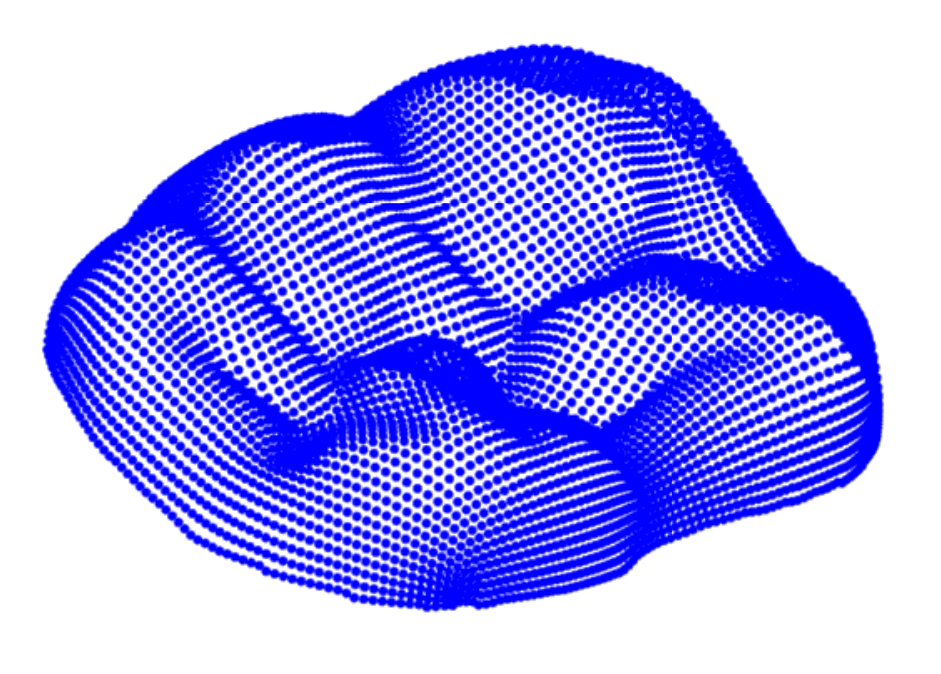}
    \label{fig:tooth_init_surf}
  }
  \subfigure[]{
    \includegraphics[width=1.8in]{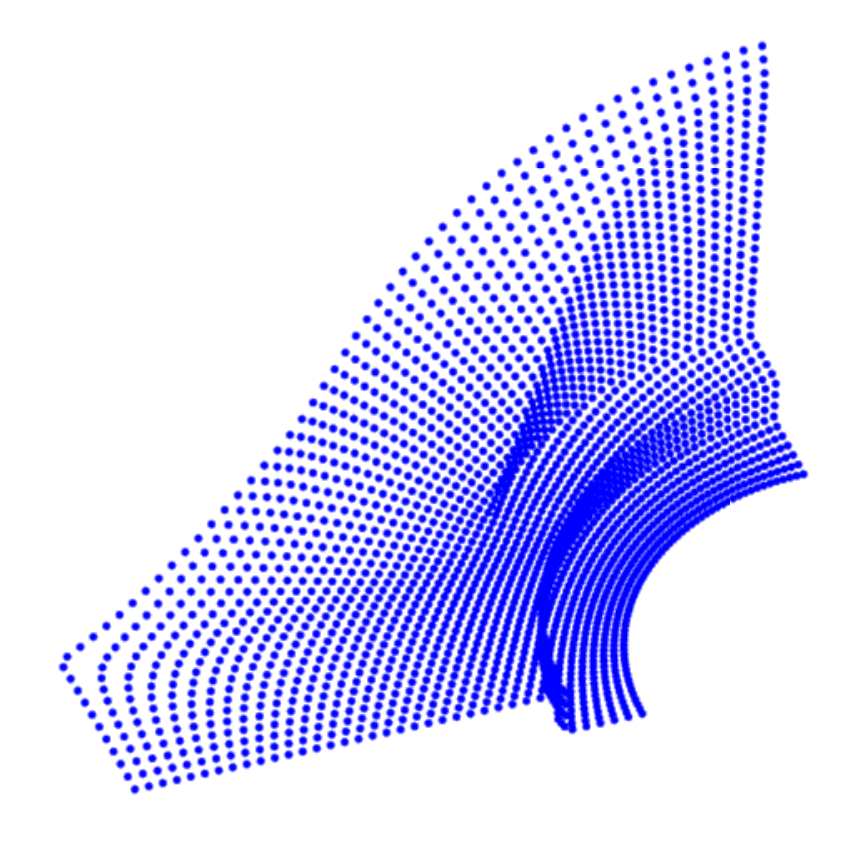}
    \label{fig:Fdisk_init_surf}
  }
  \subfigure[]{
    \includegraphics[width=1.5in]{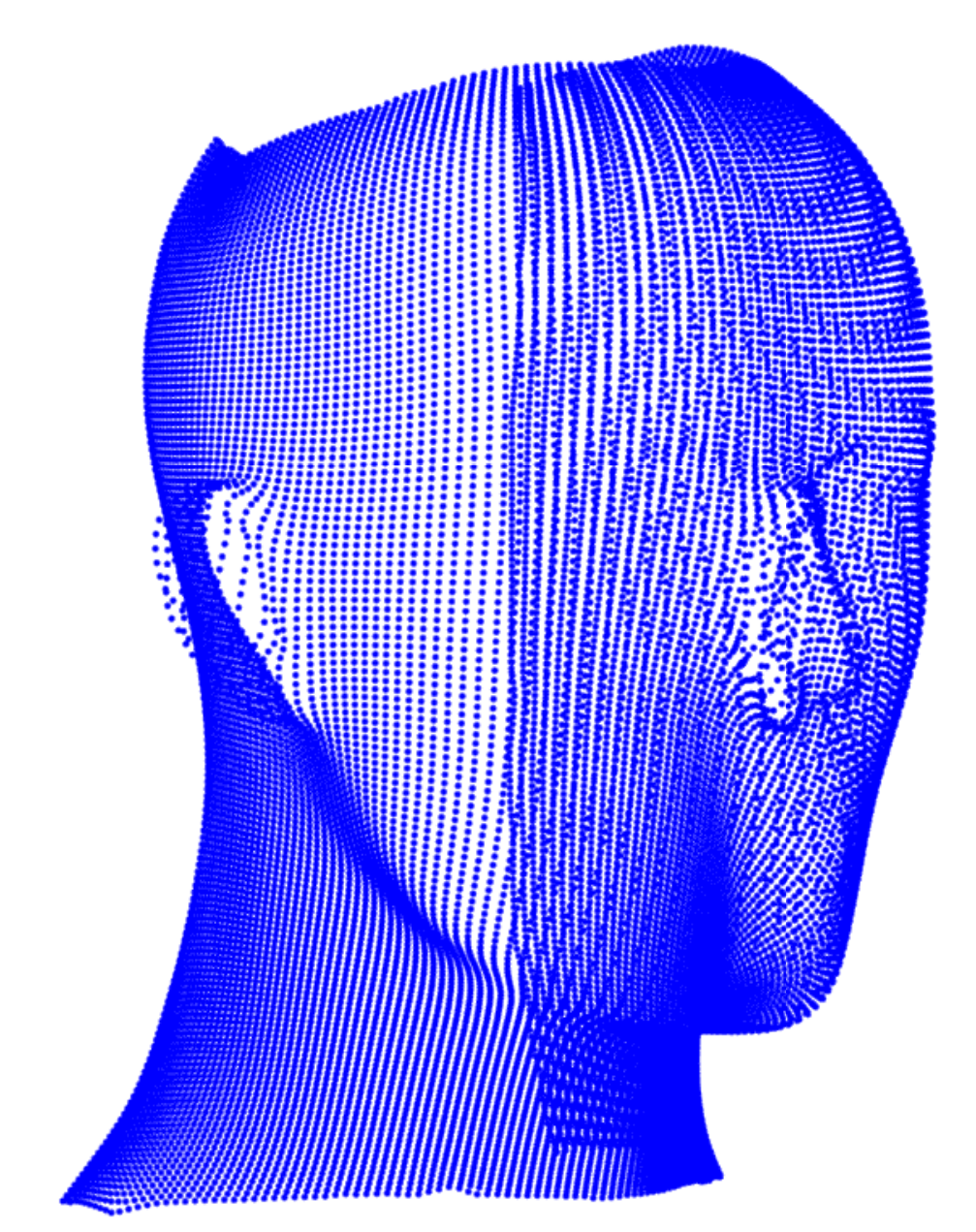}
    \label{fig:Mannequin_init_surf}
  }
  \caption
  {
    The three data models sampled from surfaces.
    (a) The \emph{tooth} model;
    (b) The \emph{fan\_disk} model;
    (c) The \emph{mannequin} model.
  }
  \label{fig:surface_model}
\end{figure*}

\begin{figure*}[!htb]
  \centering
\hspace{-0.5cm}
  \subfigure[]{
    \includegraphics[width=1.8in]{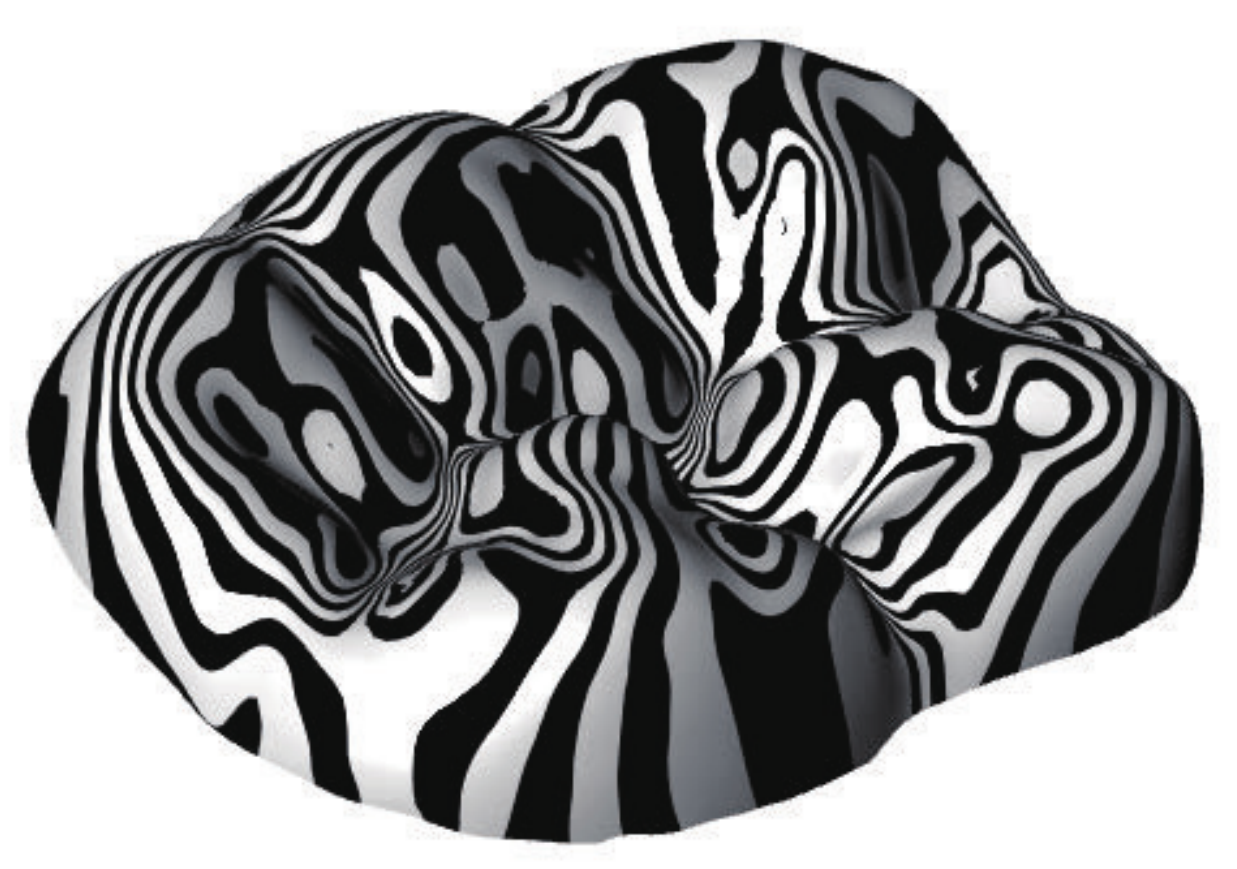}
    \label{fig:tooth_fit_zebra}
  }
  \subfigure[]{
    \includegraphics[width=1.8in]{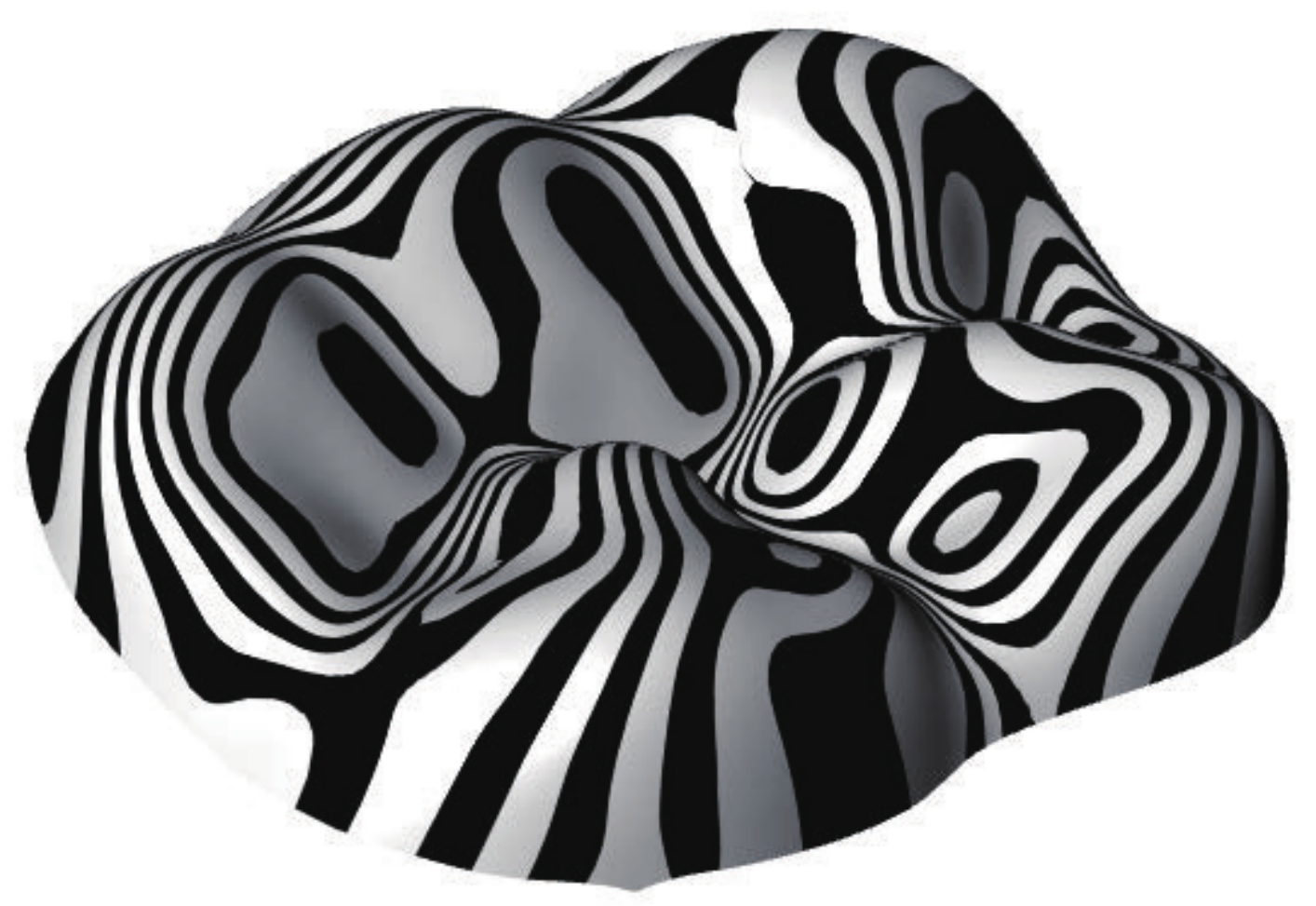}
    \label{fig:tooth_energy_zebra}
  }
  \subfigure[]{
    \includegraphics[width=1.8in]{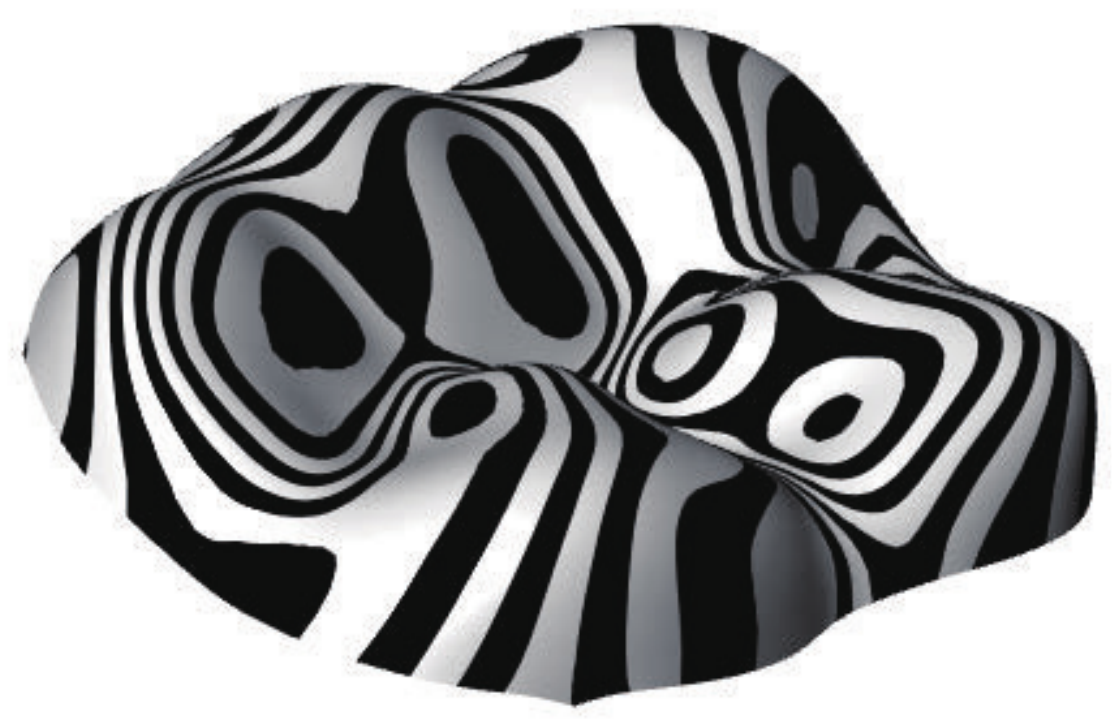}
    \label{fig:tooth_FPIA_zebra}
  }

   \subfigure[]{
    \includegraphics[width=1.8in]{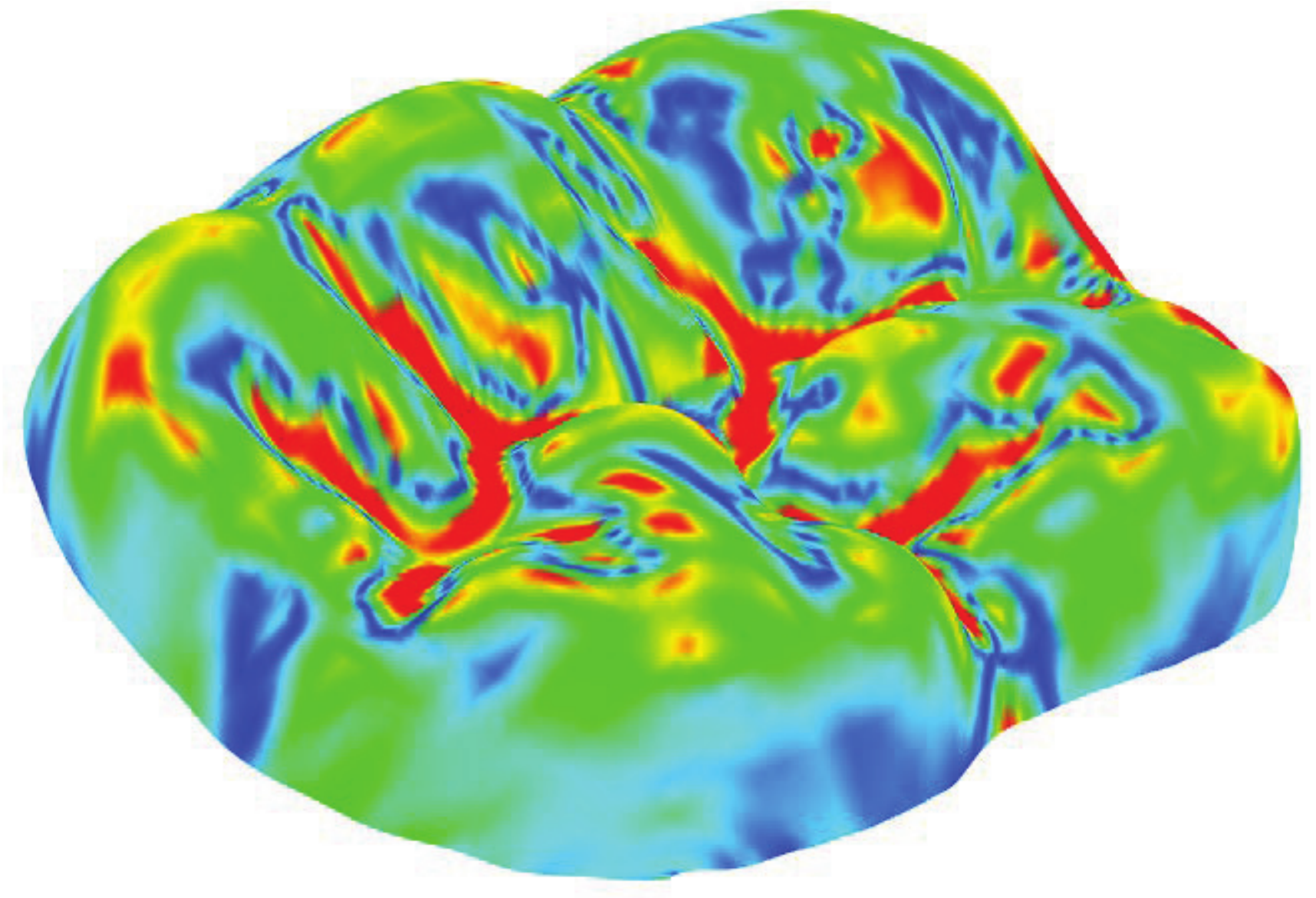}
    \label{fig:tooth_fit_curvature}
  }
\hspace{-0.35cm}
  \subfigure[]{
    \includegraphics[width=1.8in]{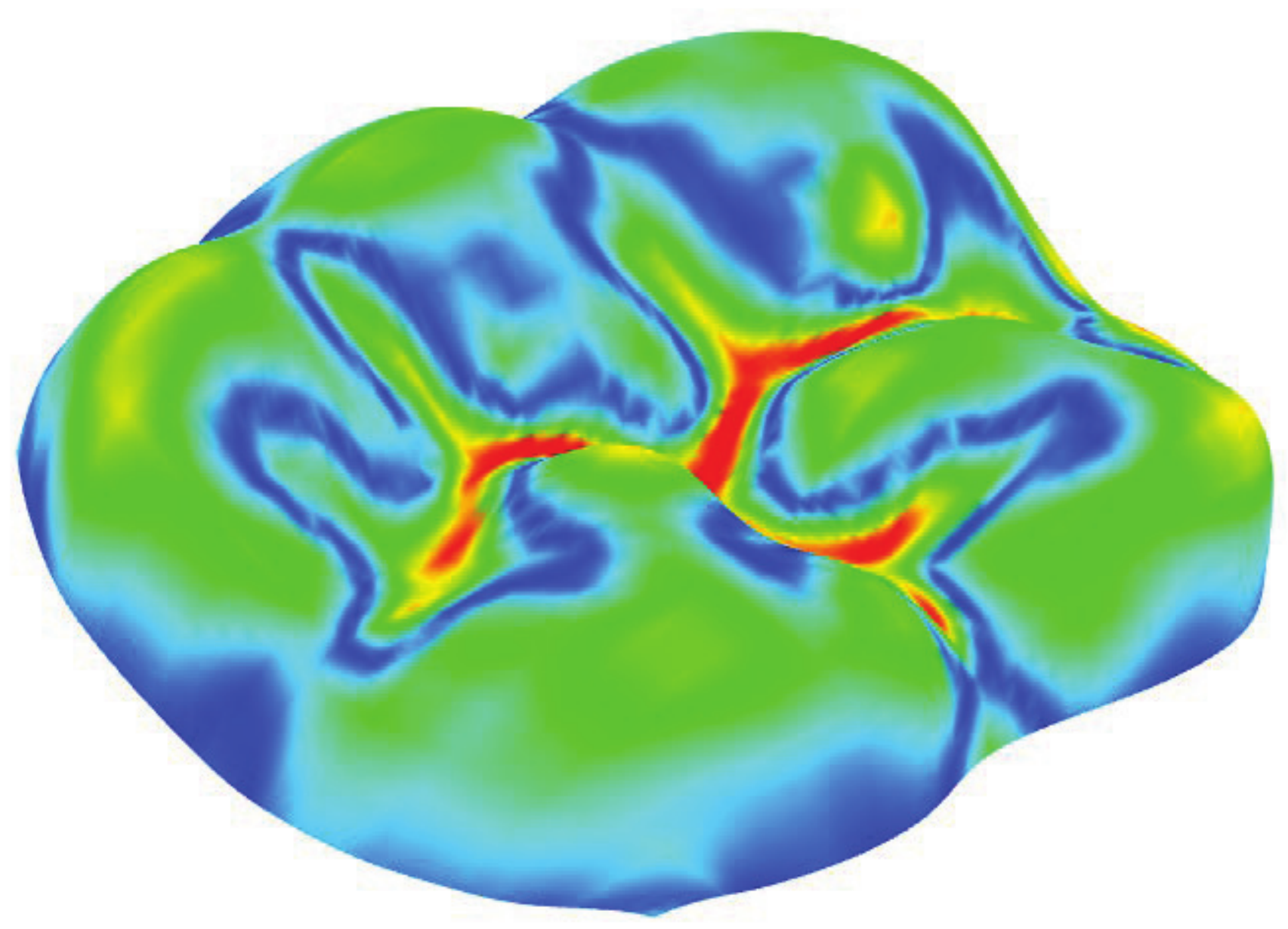}
    \label{fig:tooth_energy_curvature}
  }
\hspace{-0.35cm}
  \subfigure[]{
    \includegraphics[width=2.15in]{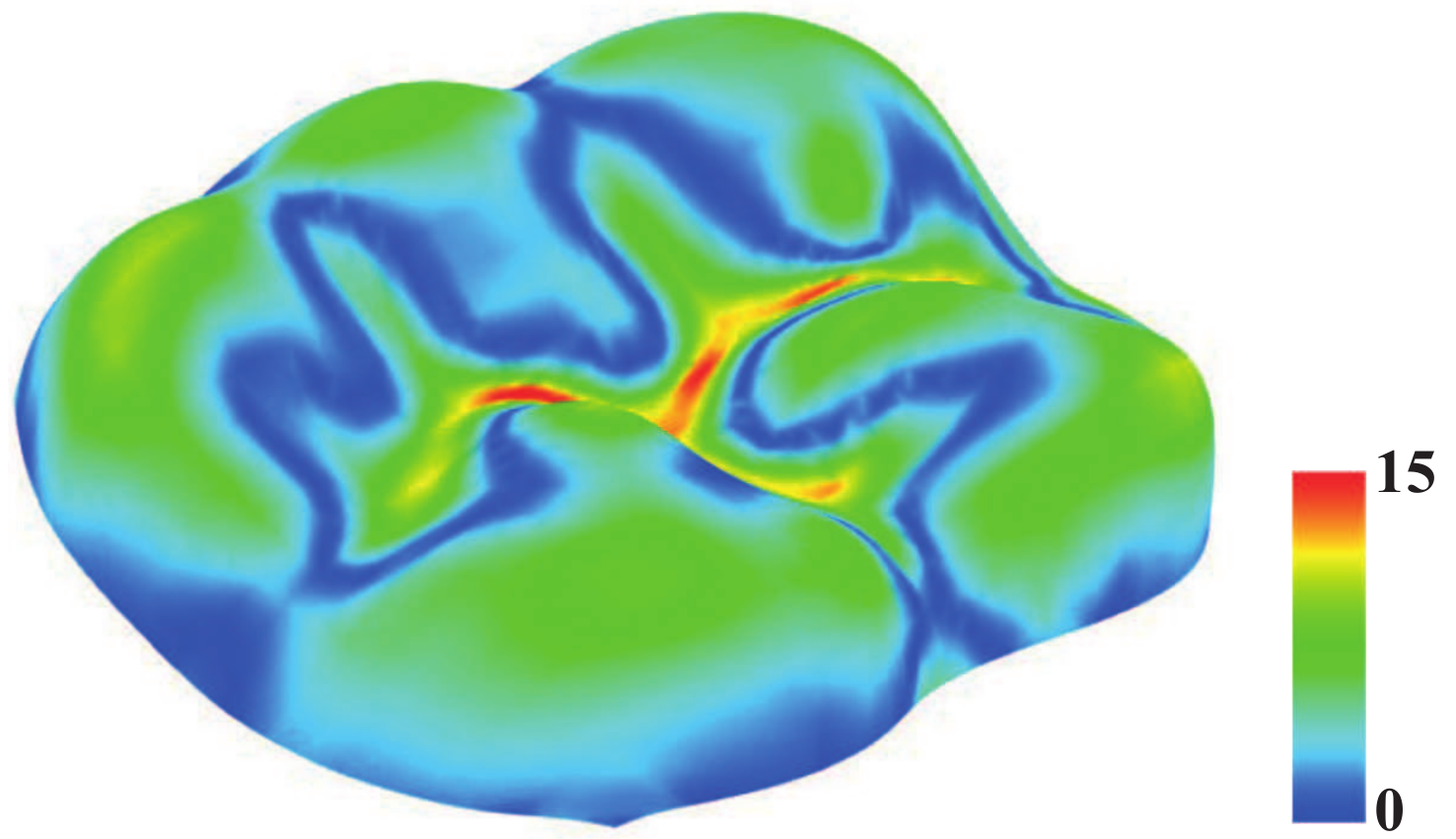}
    \label{fig:tooth_FPIA_curvature}
  }
  \caption
  {
    The \textit{tooth} model.
  (a, b, c) Zebra on the fitting surface,
    fairing surface by the traditional energy minimization method,
    and fairing surface by fairing-PIA, respectively.
  (d, e, f) Absolute value of mean curvature on the fitting surface,
    fairing surface by the energy minimization method,
    and fairing surface by fairing-PIA, respectively.
  }
  \label{fig:tooth_zebra_mean}
\end{figure*}

\begin{figure*}[!htb]
  \centering
\hspace{-1cm}
  \subfigure[]{
    \includegraphics[width=1.8in]{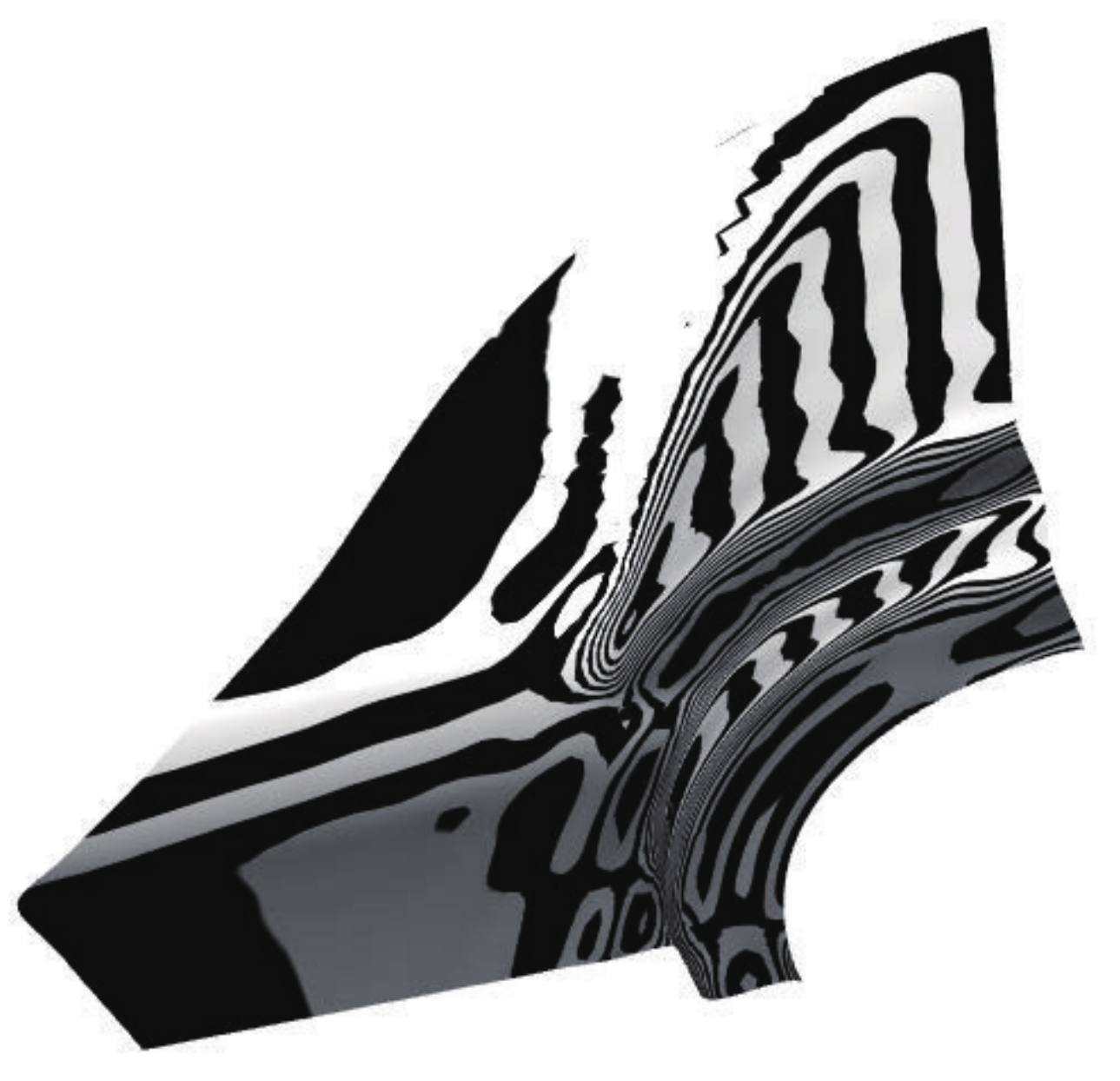}
    \label{fig:Fdisk_fit_zebra}
  }
  \subfigure[]{
    \includegraphics[width=1.8in]{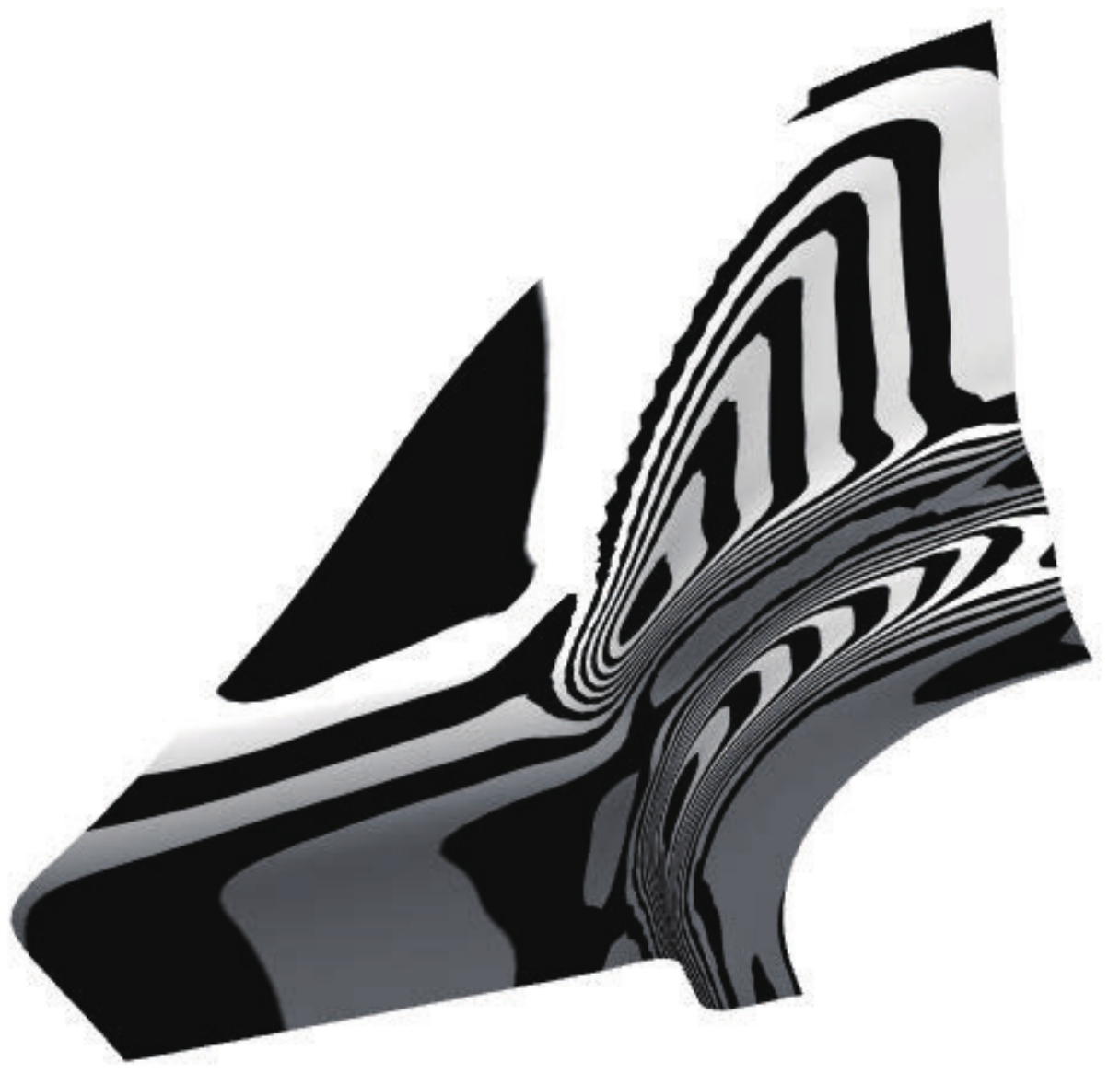}
    \label{fig:Fdisk_energy_zebra}
  }
  \subfigure[]{
    \includegraphics[width=1.8in]{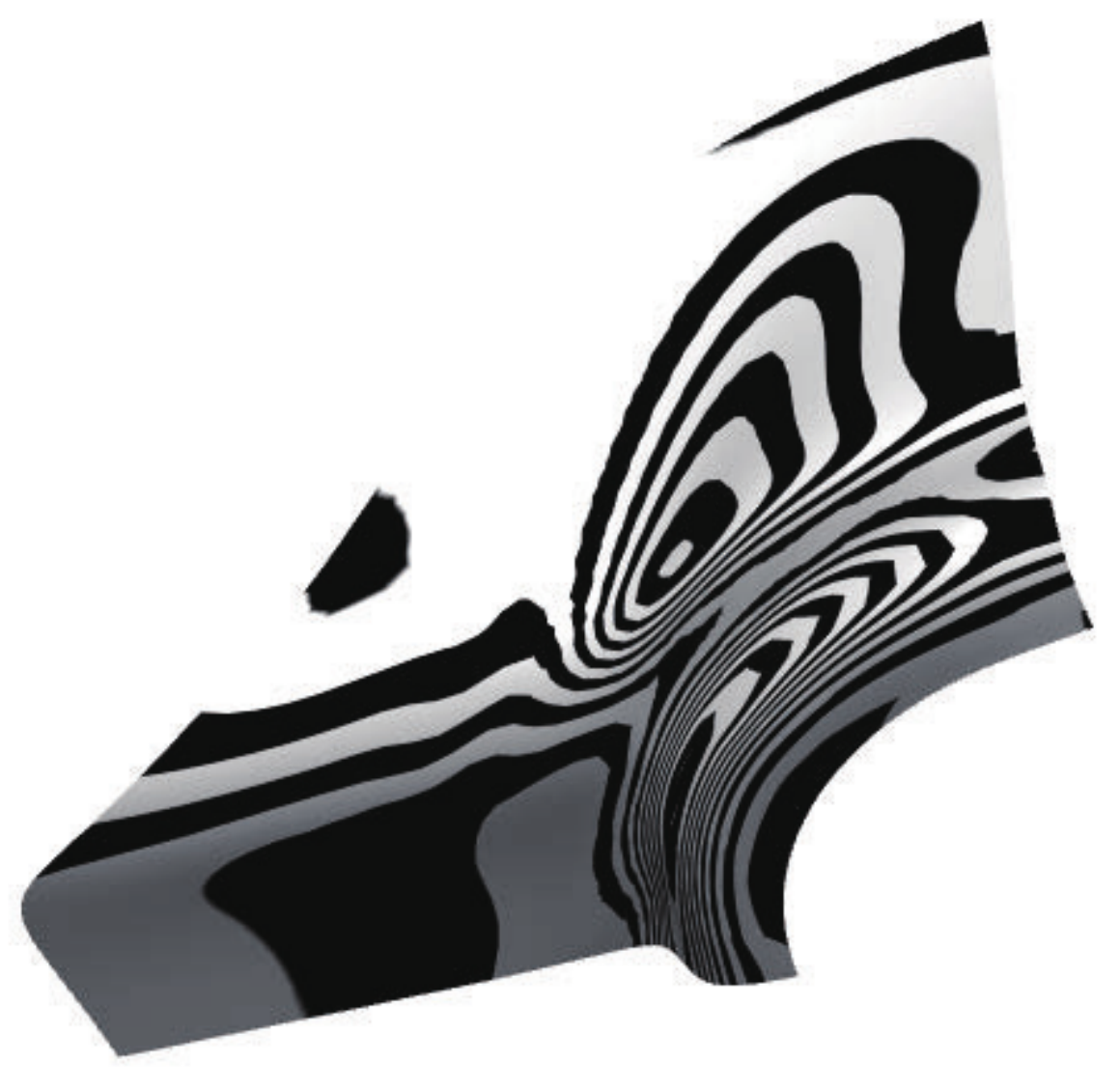}
    \label{fig:Fdisk_FPIA_zebra}
  }

 \hspace{0.5cm}
  \subfigure[]{
    \includegraphics[width=1.8in]{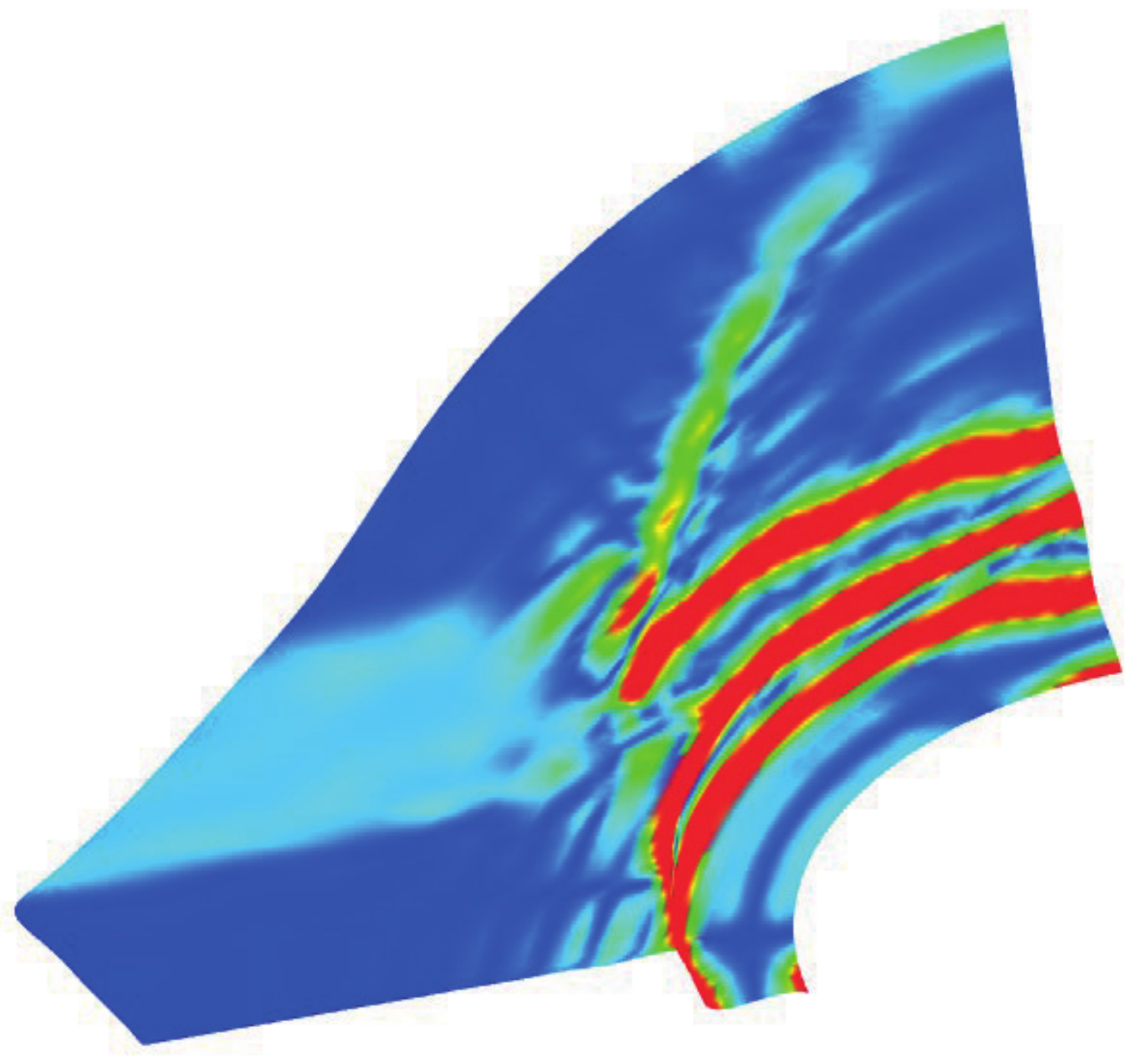}
    \label{fig:Fdisk_fit_curvature}
  }
  \subfigure[]{
    \includegraphics[width=1.8in]{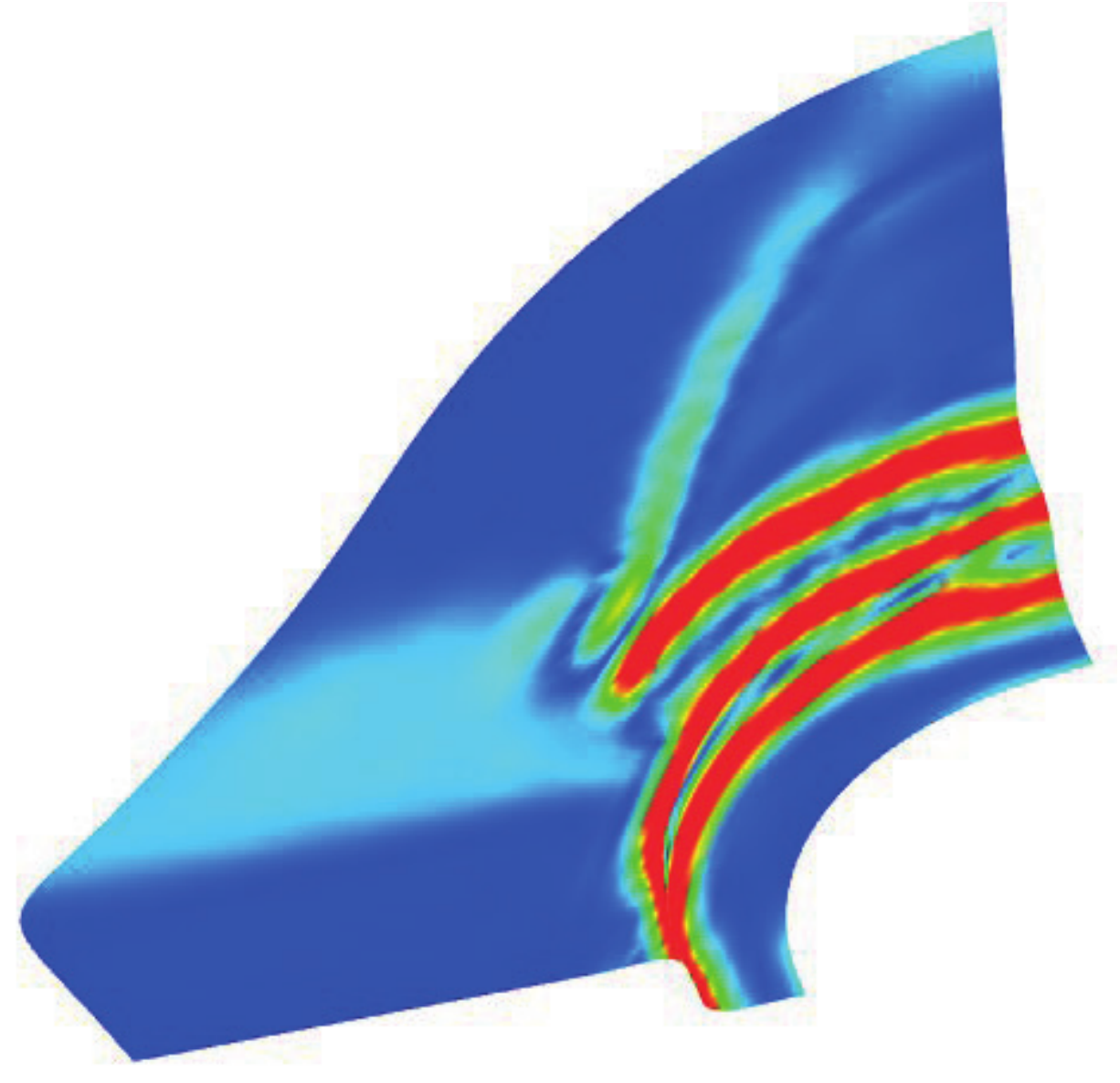}
    \label{fig:Fdisk_energy_curvature}
  }
  \subfigure[]{
    \includegraphics[width=2.2in]{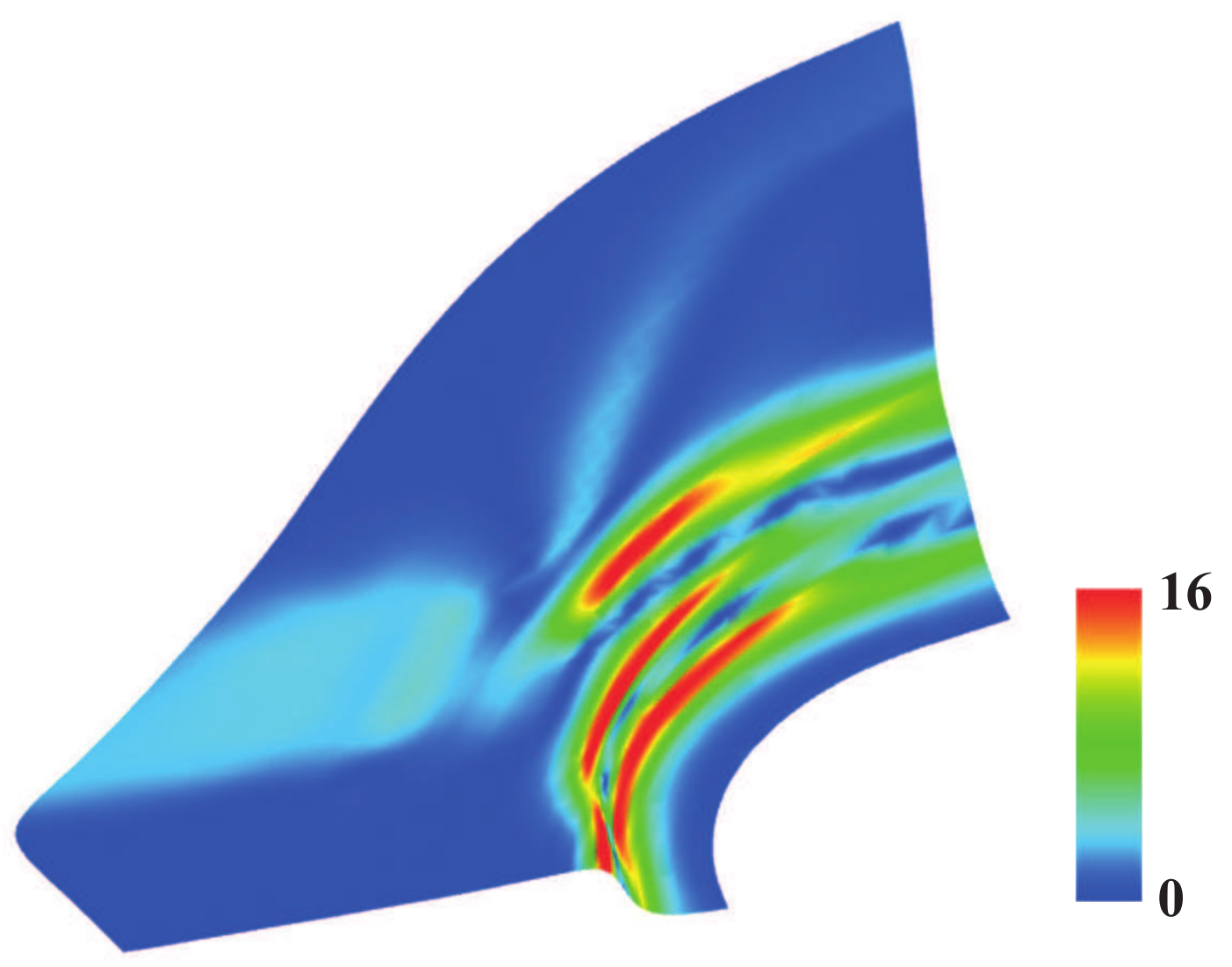}
    \label{fig:Fdisk_FPIA_curvature}
  }
  \caption
  {
    The \textit{fan\_disk} model.
  (a, b, c) Zebra on the fitting surface,
    fairing surface by the traditional energy minimization method,
    and fairing surface by fairing-PIA, respectively.
  (d, e, f) Absolute value of mean curvature on the fitting surface,
    fairing surface by the energy minimization method,
    and fairing surface by fairing-PIA, respectively.
  }
  \label{fig:fan_zebra_mean}
\end{figure*}

\begin{figure*}[!htb]
  \centering
  \subfigure[]{
    \includegraphics[width=1.5in]{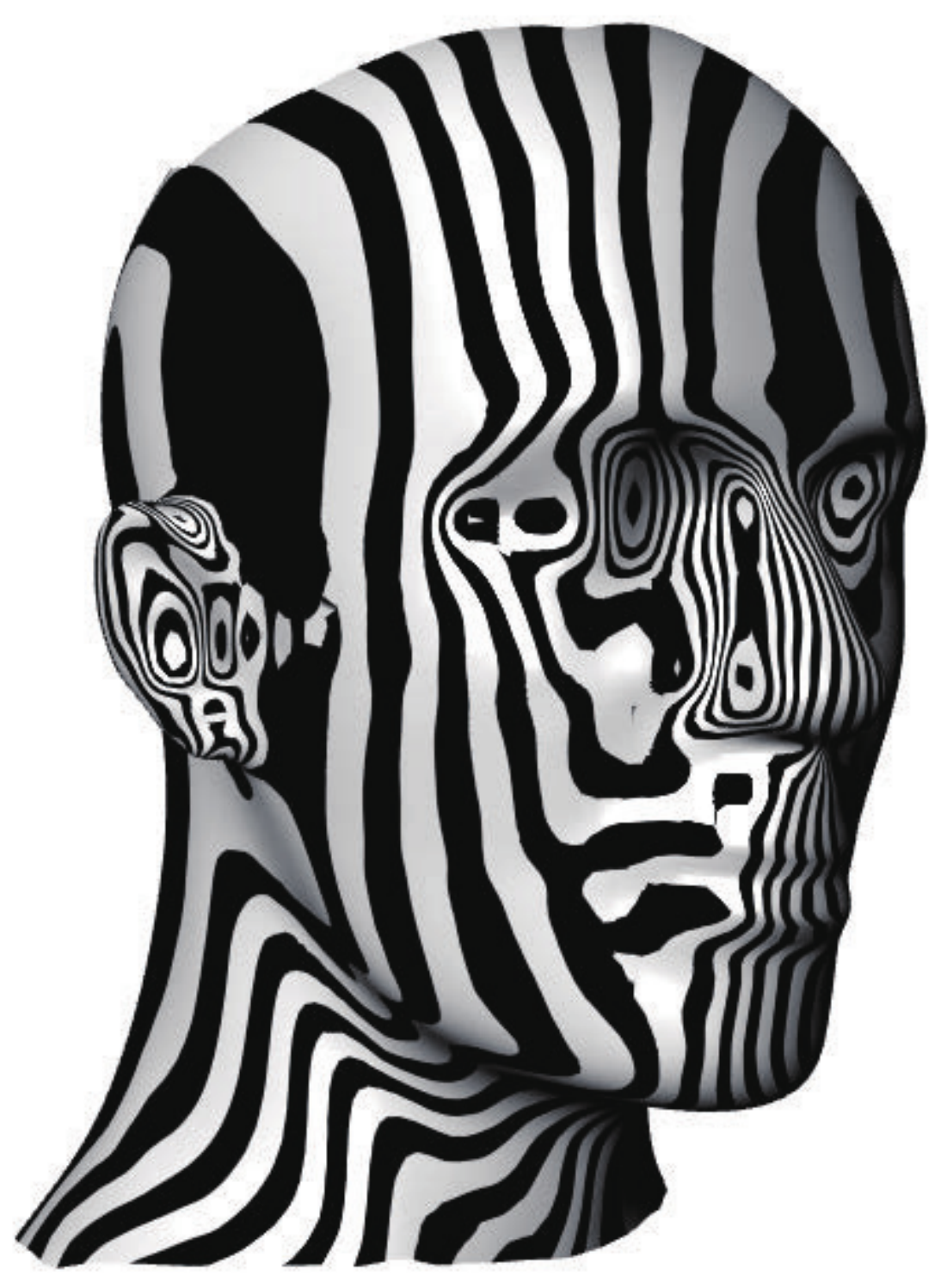}
    \label{fig:Mannequin_fit_zebra}
  }
  \subfigure[]{
    \includegraphics[width=1.5in]{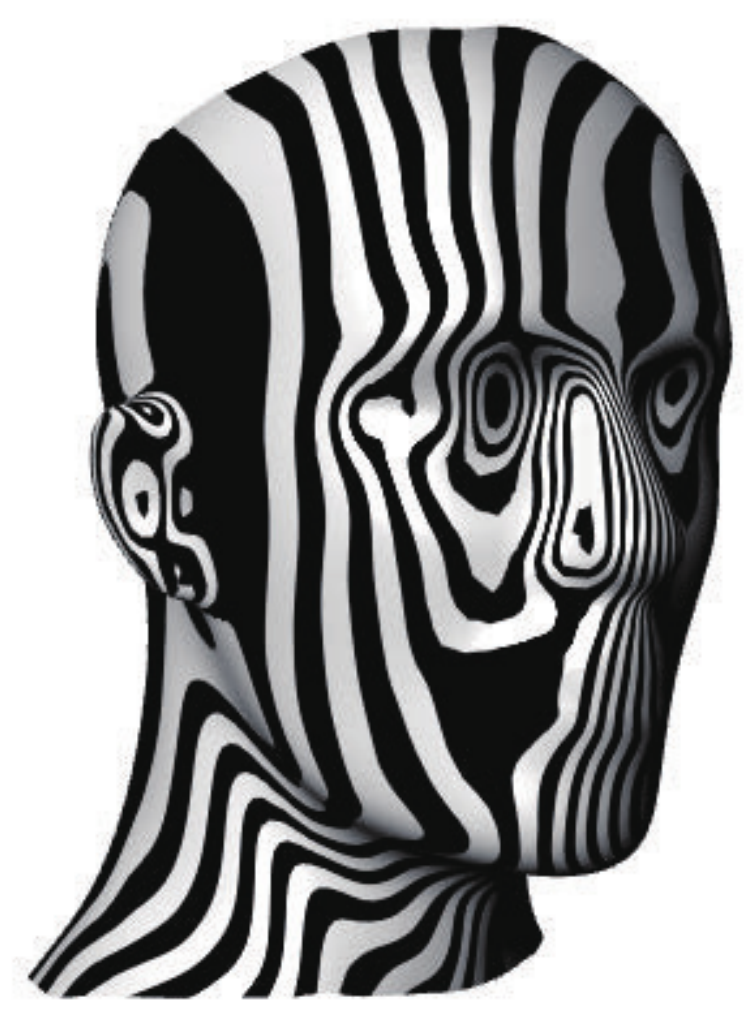}
    \label{fig:Mannequin_energy_zebra}
  }
  \subfigure[]{
    \includegraphics[width=1.5in]{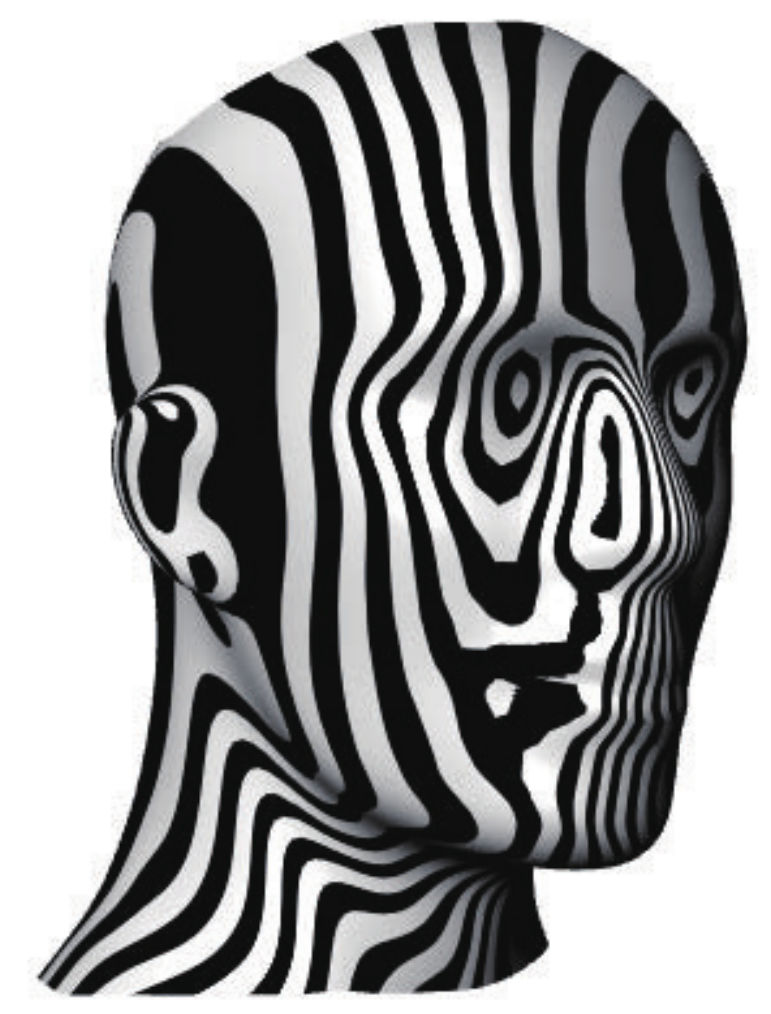}
    \label{fig:Mannequin_FPIA_zebra}
  }

 \hspace{1cm}
  \subfigure[]{
    \includegraphics[width=1.5in]{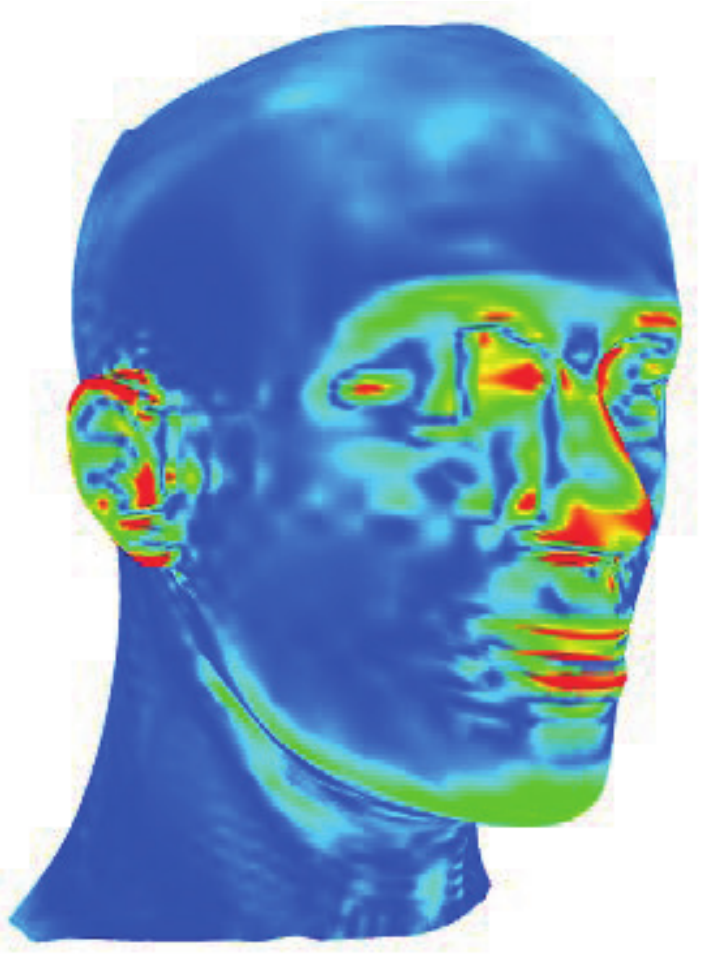}
    \label{fig:Mannequin_fit_curvature}
  }
  \subfigure[]{
    \includegraphics[width=1.5in]{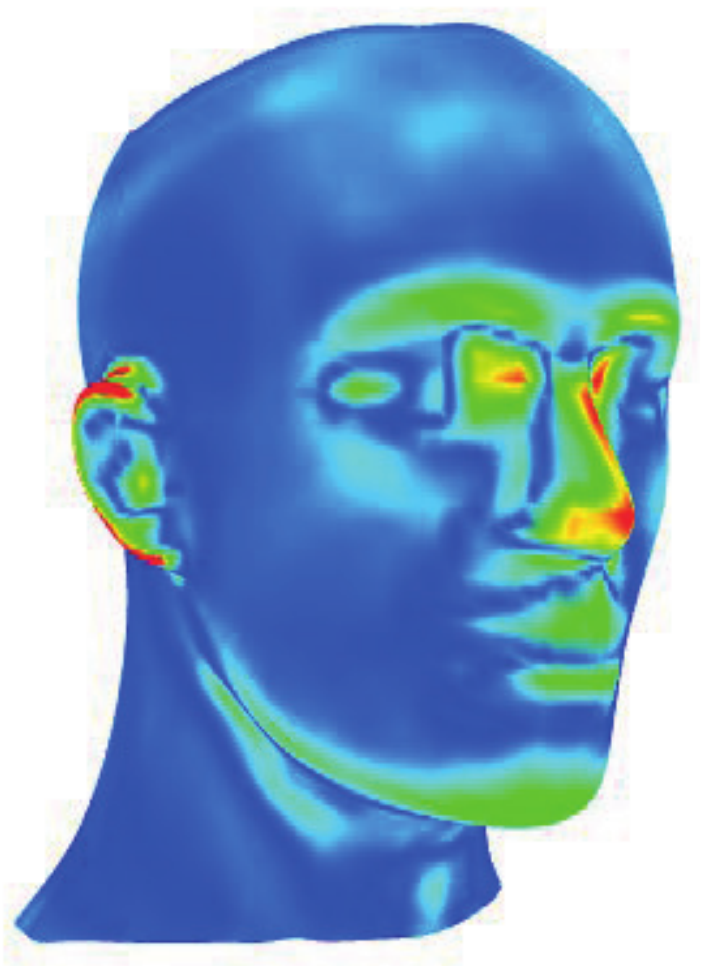}
    \label{fig:Mannequin_energy_curvature}
  }
  \subfigure[]{
    \includegraphics[width=1.75in]{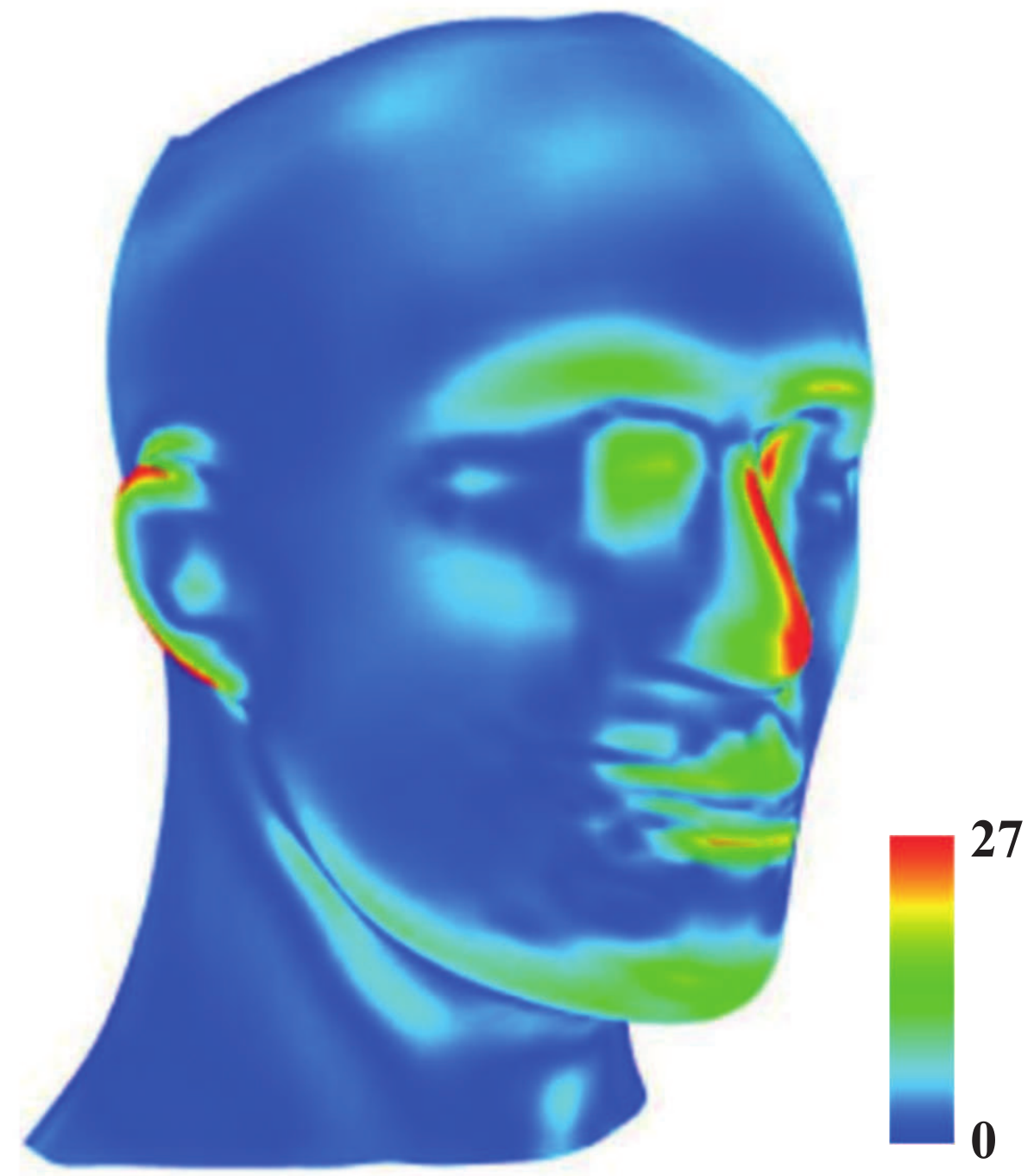}
    \label{fig:Mannequin_FPIA_curvature}
  }
  \caption
  {
    The \textit{mannequin} model.
  (a, b, c) Zebra on the fitting surface,
    fairing surface by the traditional energy minimization method,
    and fairing surface by fairing-PIA, respectively.
  (d, e, f) Absolute value of mean curvature on the fitting surface,
    fairing surface by the energy minimization method,
    and fairing surface by fairing-PIA, respectively.
  }
  \label{fig:man_zebra_mean}
\end{figure*}

\subsection{Surface examples}

 In this section,
    we compare the fairing-PIA method with the traditional energy minimization method~\pref{eq:total_energy} on three models, 
    i.e., \textit{tooth}, \textit{fan\_disk}, \textit{mannequin} (Fig.~\ref{fig:surface_model}).
 The three data models consist of $81\times61$, $41\times61$ and $121\times161$ data points.
 The statistics of the fairing-PIA on the three surface examples are listed in Table~\ref{TABLE:iteration_data}.

 For each model, we generate the fitting surface
    ($\omega =0$ in Eq.~\pref{eq:total_energy}),
    the fairing surface by the traditional energy minimization method, 
    and the fairing surface by the fairing-PIA developed in this study.
 The zebra and \emph{absolute value of mean curvature} (AVMC) are rendered on the surfaces
    to evaluate the fairness of the surfaces.
 In each example, the AVMCs larger than the maximum AVMC of the fairing surface by the fairing-PIA method 
    are mapped to the maximum AVMC of the fairing surface by fairing-PIA 
    to generate a desirable mean curvature distribution image. 
 The minimum and maximum mean curvatures of the surfaces are listed in Table~\ref{TABLE:SurfaceCurvature}.
 The absolute fitting error and absolute energy are given in Table~\ref{TABLE:FitErrorandEnergy}.
 
 For the \textit{tooth} model,
    the smoothing weight $\omega$~\pref{eq:obj_eng} is taken as $0.002$ (Figs.~\ref{fig:tooth_zebra_mean}(b) and~\ref{fig:tooth_zebra_mean}(e))
    in the traditional energy minimization method.
 Some smoothing weights are increased to $0.006$ in the fairing-PIA to improve the fairness of some regions,
    whereas the other smoothing weights remain $0.002$ (Figs.~\ref{fig:tooth_zebra_mean}(c) and ~\ref{fig:tooth_zebra_mean}(f)).
 For the \textit{fan\_disk} model,
    the smoothing weight in the traditional energy minimization method is taken 
    as $\omega = 0.001$~\pref{eq:obj_eng}(Figs.~\ref{fig:fan_zebra_mean}(b) and ~\ref{fig:fan_zebra_mean}(e)).
  Similarly, 
    the corresponding smoothing weights are increased to $0.015$
    to improve the fairness of some regions (Figs.~\ref{fig:fan_zebra_mean}(c) and ~\ref{fig:fan_zebra_mean}(f)).
    
 For the \textit{mannequin} model,
    the smoothing weight in the traditional energy minimization method~\pref{eq:obj_eng} is set as $\omega = 1 \times 10^{-4}$.
 The generated fairing surface is illustrated in Figs.~\ref{fig:man_zebra_mean}(b) and ~\ref{fig:man_zebra_mean}(e).
 To improve the fairness of the ear region, 
    we increase the smoothing weights of the control points at the ear region to $5 \times 10^{-4}$; 
    to keep the feature of the mouth region by reducing the fitting error, we decrease those at the mouth region to $2 \times 10^{-5}$.
 With the modified smoothing weights,
    the fairing surface is produced by fairing-PIA (Figs.~\ref{fig:man_zebra_mean}(c) and ~\ref{fig:man_zebra_mean}(f)).
 The fairness of the ear region is improved, and the feature of the mouth region becomes clear.
 
\begin{table*}[!htb] 
	\centering
	\caption{Minimum and maximum mean curvature value of surfaces.}
  \begin{threeparttable}
  \begin{tabular}{lllllll}
    \toprule
    \multirow{2}{*}{Type} &
    \multicolumn{2}{c}{\textit{Tooth}} &
    \multicolumn{2}{c}{\textit{Fan\_disk}} &
    \multicolumn{2}{c}{\textit{Mannequin}}
    \\
    \cmidrule(rl){2-3} \cmidrule(rl){4-5} \cmidrule(rl){6-7}
    &  MinCur\tnote{a} & MaxCur\tnote{b} & MinCur & MaxCur & MinCur & MaxCur \\
    \midrule
    Fitting surface & -15.8553 & 1791.7996 & -33.7973 & 42.6118  & -23.3849 & 27.0701 \\
    Fairing surface by energy minimization & -11.9428 & 19.2055 & -27.8114 & 31.155 & -20.3517 & 27.0489 \\
    Fairing surface by fairing-PIA & $\mathbf{-9.7458}$ & $\mathbf{14.9174}$ & $\mathbf{-15.3867}$ & $\mathbf{15.7318}$ &  $\mathbf{-20.3517}$ & $\mathbf{23.3825}$ \\
    \bottomrule
  \end{tabular}
  \begin{tablenotes}
    \item[a] Minimum mean curvature value.
    \item[b] Maximum mean curvature value.
  \end{tablenotes}
  \end{threeparttable}
  \label{TABLE:SurfaceCurvature}
\end{table*}

\section{Conclusion}

In this study, we present the fairing-PIA, 
    a novel progressive iterative approximation for fairing curves and surface generation. 
The convergence analysis of the fairing-PIA is presented. 
Fairing-PIA is both a global and local method, 
    where many parameters that can finely improve fairness exist. 
Moreover, fairing-PIA is flexible. 
Thus, the smoothing weights and knot vector can be changed in the iterations. 
Finally, many examples are presented to demonstrate the efficiency and effectiveness of the developed fairing-PIA method.

\section*{Acknowledgements}
This work is supported by the National Natural Science Foundation of China under Grant nos. 61872316, 61932018, and the National Key R\&D Plan of China under Grant no. 2020YFB1708900.

\section*{References}

\bibliography{mybib}

\begin{thebibliography}{10}
\expandafter\ifx\csname url\endcsname\relax
  \def\url#1{\texttt{#1}}\fi
\expandafter\ifx\csname urlprefix\endcsname\relax\def\urlprefix{URL }\fi
\expandafter\ifx\csname href\endcsname\relax
  \def\href#1#2{#2} \def\path#1{#1}\fi

\bibitem{Li2005}
W.~Li, S.~Krist, Spline-based airfoil curvature smoothing and its applications,
  Journal of aircraft 42~(4) (2005) 1065--1074.

\bibitem{SARIOZ20062105}
E.~Sariöz, An optimization approach for fairing of ship hull forms, Ocean
  Engineering 33~(16) (2006) 2105--2118.

\bibitem{WESTGAARD2001619}
G.~Westgaard, H.~Nowacki, A process for surface fairing in irregular meshes,
  Computer Aided Geometric Design 18~(7) (2001) 619--638.

\bibitem{Hashemian2018}
A.~Hashemian, S.~F. Hosseini, An integrated fitting and fairing approach for
  object reconstruction using smooth nurbs curves and surfaces, Computers \&
  mathematics with applications (1987) 76~(7) (2018) 1555--1575.

\bibitem{Gili2011}
R.~F. Gilimyanov, Recursive method of smoothing curvature of path in path
  planning problems for wheeled robots, Automation and remote control 72~(7)
  (2011) 1548--1556.

\bibitem{Song2021}
B.~Song, Z.~Wang, L.~Zou, An improved pso algorithm for smooth path planning of
  mobile robots using continuous high-degree b\'{e}zier curve, Applied soft
  computing 100 (2021).

\bibitem{Hagen1991}
H.~Hagen, G.-P. Bonneau, Variational design of smooth rational bézier curves,
  Computer aided geometric design 8~(5) (1991) 393--399.

\bibitem{WANG2010}
W.~Wang, Y.~Zhang, Wavelets-based nurbs simplification and fairing:
  Computational geometry and analysis, Computer methods in applied mechanics
  and engineering 199~(5-8) (2010) 290--300.

\bibitem{Kjellander1983}
J.~A.~P. Kjellander, Smoothing of cubic parametric splines, Computer aided
  design 15~(3) (1983) 175--179.

\bibitem{FARIN198791}
G.~Farin, G.~Rein, N.~Sapidis, A.~Worsey, Fairing cubic b-spline curves,
  Computer Aided Geometric Design 4~(1) (1987) 91--103.

\bibitem{Wang2015}
A.~Wang, G.~Zhao, Y.-D. Li, Wavelet based local fairing algorithm for surfaces
  with data compression, International journal of wavelets, multiresolution and
  information processing 13~(5) (2015) 1550041.

\bibitem{Pottmann1990}
H.~Pottmann, Smooth curves under tension, Computer aided design 22~(4) (1990)
  241--245.

\bibitem{Campbell1957}
E.~S. Campbell, Smoothest curve approximation, Mathematics of computation
  11~(60) (1957) 233--243.

\bibitem{Meier1987}
H.~Meier, H.~Nowacki, Interpolating curves with gradual changes in curvature,
  Computer aided geometric design 4~(4) (1987) 297--305.

\bibitem{10.5555/917424}
T.~C. Rando, Automatic fairness in computer-aided geometric design, Ph.D.
  thesis, aAI9102032 (1990).

\bibitem{Vassilev1996}
T.~I. Vassilev, Fair interpolation and approximation of b-splines by energy
  minimization and points insertion, Computer aided design 28~(9) (1996)
  753--760.

\bibitem{Zhang2001}
C.~Zhang, P.~Zhang, F.~Cheng, Fairing spline curves and surfaces by minimizing
  energy, Computer aided design 33~(13) (2001) 913--923.

\bibitem{Cali2010}
F.~Caliò, E.~Miglio, M.~Rasella, Curve fairing using integral spline
  operators, International journal for numerical methods in biomedical
  engineering 26~(12) (2010) 1674--1686.

\bibitem{LI2004499}
W.~Li, S.~Xu, J.~Zheng, G.~Zhao, Target curvature driven fairing algorithm for
  planar cubic b-spline curves, Computer Aided Geometric Design 21~(5) (2004)
  499--513.

\bibitem{Aimin2011}
A.~Li, G.~Wei, H.~Tian, F.~Kou, Wavelet-based multiresolution nurbs curve
  fairing, Advanced Materials Research 314-316 (2011) 1562--1565.

\bibitem{Aimin2012}
A.~Li, H.~Tian, A multiresolution fairing approach for nurbs curves, Applied
  Mechanics and Materials 215-216 (2012) 1205--1208.

\bibitem{Josef1993}
J.~Hoschek, D.~Lasser, L.~L. Schumaker, Fundamentals of Computer Aided
  Geometric Design, A. K. Peters, Ltd., USA, 1993.

\bibitem{Qi75}
D.~Qi, Z.~Tian, Y.~Zhang, The method of numeric polish in curve fitting, Acta
  Math. Sin. 18 (1975) 173--184.

\bibitem{Deboor1979}
C.~de~Boor, How does agee's smoothing method work?, Proceedings of the 1979
  Army Numerical Analysis and Computers Conference (01 1979).

\bibitem{Lin04}
H.~Lin, G.~Wang, C.~Dong, Constructing iterative non-uniform b-spline curve and
  surface to fit data points, Science in China Series F: Information Sciences
  47 (2004) 315--331.

\bibitem{LIN2005575}
H.~Lin, H.~Bao, G.~Wang, Totally positive bases and progressive iteration
  approximation, Computers \& Mathematics with Applications 50~(3) (2005)
  575--586.

\bibitem{Shi06}
L.~Shi, R.~Wang, An iterative algorithm of nurbs interpolation and
  approximation, Journal of Mathematical Research \& Exposition 26 (01 2006).

\bibitem{LU2010}
Weighted progressive iteration approximation and convergence analysis, Computer
  Aided Geometric Design 27~(2) (2010) 129--137.

\bibitem{Wang2018}
Z.~Wang, Y.~Li, C.~Deng, Convergence proof of gs-pia algorithm, Journal of
  Computer-Aided Design \& Computer Graphics 30 (2018) 2035.

\bibitem{DENG201432}
C.~Deng, H.~Lin, Progressive and iterative approximation for least squares
  b-spline curve and surface fitting, Computer-Aided Design 47 (2014) 32--44.

\bibitem{Chen2008}
Z.~Chen, X.~Luo, L.~Tan, B.~Ye, J.~Chen, Progressive interpolation based on
  catmull clark subdivision surfaces, Comput. Graph. Forum 27 (2008)
  1823--1827.

\bibitem{Deng2012}
C.~Deng, W.~Ma, Weighted progressive interpolation of loop subdivision
  surfaces, Computer-Aided Design 44 (2012) 424--431.

\bibitem{Hamza2021}
Y.~Hamza, H.~Lin, Conjugate-gradient progressive-iterative approximation for
  loop and catmull-clark subdivision surface interpolation, Journal of Computer
  Science and Technology 37~(2) (2022) 487--502.

\bibitem{Hamza2020}
Y.~Hamza, H.~Lin, Z.~Li, Implicit progressive-iterative approximation for curve
  and surface reconstruction, Computer Aided Geometric Design 77 (2020) 101817.

\bibitem{Zhang2019}
Y.~Zhang, P.~Wang, F.~Bao, X.~Yao, C.~Zhang, H.~Lin, A single image
  super-resolution method based on progressive-iterative approximation, IEEE
  Transactions on Multimedia 22 (2020) 1407--1422.

\bibitem{linreview2018}
H.~Lin, T.~Maekawa, C.~Deng, Survey on geometric iterative methods and their
  applications, Computer aided design 95 (2018) 40--51.

\bibitem{veltkamp1995modeling}
R.~C. Veltkamp, W.~Wesselink, Modeling 3d curves of minimal energy, in:
  Computer Graphics Forum, Vol.~14, Wiley Online Library, 1995, pp. 97--110.

\bibitem{zhang2001fairing}
C.~Zhang, P.~Zhang, F.~F. Cheng, Fairing spline curves and surfaces by
  minimizing energy, Computer-Aided Design 33~(13) (2001) 913--923.

\bibitem{meier1987interpolating}
H.~Meier, H.~Nowacki, Interpolating curves with gradual changes in curvature,
  Computer Aided Geometric Design 4~(4) (1987) 297--305.

\bibitem{wang1997energy}
X.~Wang, F.~F. Cheng, B.~A. Barsky, Energy and b-spline interproximation,
  Computer-Aided Design 29~(7) (1997) 485--496.

\bibitem{Lin2018}
H.~Lin, Q.~Cao, X.~Zhang, The convergence of least-squares progressive
  iterative approximation for singular least-squares fitting system, Journal of
  systems science and complexity 31~(6) (2018) 1618--1632.

\end{thebibliography}

\newpage
\appendix
\section*{Appendix}
\label{Appendix}
In this section,
   we show the convergence of the fairing-PIA method when the diagonal elements of matrix $\bm{\Omega}$ are equal to a constant value $\omega$,
   and coefficient matrix $\bm{B}$ is a positive semidefinite matrix.

\begin{lemma}
Let the diagonal elements of matrix $\bm{\Lambda}$
   satisfy $0<\mu_i<\frac{2}{\lambda_{\text{max}}\left(\bm{B}\right)},\;i=1,2,\cdots,n$,
   where $\lambda_{\text{max}}$ is the largest eigenvalue of $\bm{B}$.
The eigenvalues $\beta$ of the matrix $\bm{\Lambda B}$\eqref{eq:equal_weight_ctrlpts} are all real and satisfy $0\le\beta<2$.
  \label{lemma:eigenvalue_range}
\end{lemma}

\begin{proof}
Matrix $\bm{B}$ is a positive semidefinite matrix.
Thus, real orthogonal matrix $\bm{U}$ and diagonal matrix $\bm{S}=\text{diag}\left(s_1,s_2,
    \cdots,s_n\right)$, $s_i\ge 0,\; i=1,2,\cdots,n$ exist,
   such that $\bm{B} = \bm{USU}^T$.
We obtain the following by denoting $\bm{C}=\bm{S}^{\frac{1}{2}}\bm{U}^T$:
  \begin{equation}
    \bm{B} = \bm{C}^T\bm{C}.
    \label{eq:semidefinite_coef}
  \end{equation}

Suppose that $\beta$ is an arbitrary eigenvalue of matrix $\bm{\Lambda C}^T\bm{C}$ w.r.t. eigenvector $\bm{v}$, i.e.,
  \begin{equation}
    \bm{\Lambda C}^T\bm{Cv} = \beta\bm{v}.
    \label{eq:equal_weight_symm_semidefinite}
  \end{equation}
In this case, we obtain
  \begin{equation*}
    \bm{C\Lambda C}^T\left(\bm{Cv}\right) = \beta\left(\bm{Cv}\right),
  \end{equation*}
  by multiplying both sides of Eq.~\eqref{eq:equal_weight_symm_semidefinite} by $\bm{C}$,
  indicating that $\beta$ is also an eigenvalue of the matrix $\bm{C\Lambda C}^T$ w.r.t. eigenvector $\bm{Cv}$.
Moreover, for all nonzero vectors $\bm{x}\in R^{n}$, it holds
  \begin{equation*}
    \bm{x}^T\bm{C\Lambda C}^T\bm{x} =
    \left(\bm{x}^T\bm{C\Lambda}^{\frac{1}{2}}\right)\left(\bm{x}^T\bm{C\Lambda}^{\frac{1}{2}}\right)^T\ge 0.
  \end{equation*}
Therefore, we conclude that the matrix $\bm{C\Lambda C}^T$ is a positive semidefinite matrix.
Thus, all its eigenvalues are nonnegative real numbers, i.e., $\beta\ge 0$.
According to Eq.~\eqref{eq:equal_weight_symm_semidefinite},
   the eigenvalues of matrix $\bm{\Lambda B}=\bm{\Lambda C}^T\bm{C}$ are also all nonnegative real numbers.

On the other hand,  given that $0<\mu_i<\frac{2}{\lambda_{\text{max}}\left(\bm{B}\right)},\;i=1,2,\cdots,n$,
   the eigenvalues of $\bm{\Lambda B}$ meet
  \begin{equation*}
    0 \le\lambda\left(\bm{\Lambda B}\right)
    < \left\| \bm{\Lambda B}\right\|_2
    \leq \left\| \bm{\Lambda}\right\|_2 \left\| \bm{B}\right\|_2
    =
     \lambda_{\text{max}}\left(\bm{\Lambda}\right)\lambda_{\text{max}}\left(\bm{B}\right)
    < 2.
  \end{equation*}
\end{proof}

\begin{remark}
Let $n_0$ denote the dimension of the zero eigenspace of $\bm{B}$~\eqref{eq:mtx_B}.
Given that $\bm{\Lambda}$ is nonsingular,
    the ranks of the matrix $\bm{B}$ and $\bm{\Lambda B}$ are the same, i.e.,
  \begin{equation*}
    \text{rank}\left(\bm{\Lambda B}\right)
    =\text{rank}\left(\bm{B}\right)
    =n-n_0.
  \end{equation*}
  \label{rmk4:zero_eigenspace}
\end{remark}

Moreover, in the proof of Theorem~\ref{thrm:equal_weight__semidefinite_convergent},
   a lemma similar to Lemma 2.5 in~\cite{Lin2018} is used as follows:
\begin{lemma}
  The algebraic multiplicity of the zero eigenvalue of matrix $\bm{\Lambda B}$
  is equal to its geometric multiplicity.
  \label{lemma:GM_AM}
\end{lemma}

\begin{proof}
Based on Eq.~\eqref{eq:semidefinite_coef},
    matrix $\bm{\Lambda B}$ can be written as $\bm{\Lambda C}^T\bm{C}$.
Then, the proof of this lemma is similar to that of Lemma 2.5 provided in~\cite{Lin2018}.
\end{proof}

\begin{lemma}
The Jordan canonical form of the matrix $\bm{\Lambda B}$\eqref{eq:equal_weight_ctrlpts} can be represented as
  \begin{equation}
    \bm{J}=
    \left[\begin{matrix}
        \begin{matrix}
          \bm{J}_{n_1}(\beta_1,1) &                         &        &                         \\
                                  & \bm{J}_{n_2}(\beta_2,1) &        &                         \\
                                  &                         & \ddots &                         \\
                                  &                         &        & \bm{J}_{n_k}(\beta_k,1)
        \end{matrix}
         & \text{\Huge 0} \\
        \text{\Huge 0}
         &
        \begin{matrix}
          0 &        &   \\
            & \ddots &   \\
            &        & 0
        \end{matrix}
      \end{matrix}\right]_{n\times n},
    \label{eq:jordan_canonical}
  \end{equation}
  where
  \begin{equation*}
    \bm{J}_{n_i}(\beta_i,1)=
    \begin{bmatrix}
      \beta_i & 1       &        &         \\
              & \beta_i & \ddots &         \\
              &         & \ddots & 1       \\
              &         &        & \beta_i
    \end{bmatrix}\in R^{n_i\times n_i}
  \end{equation*}
  is called a Jordan block of size $n_i$ with eigenvalue $\beta_i$,
  and $0\le\beta_i<2$ and $i=1,2,\cdots,k$ are the nonzero eigenvalues of $\bm{\Lambda B}$.
  \label{lemma:Jordan}
\end{lemma}

\begin{proof}
Based on Lemma.~\ref{lemma:eigenvalue_range},
   the eigenvalues of $\bm{\Lambda B}$ are all nonnegative real numbers.
Thus, the Jordan canonical form of $\bm{\Lambda B}$ can be written as
  \begin{equation*}
    \bm{J}=
    \left[\begin{matrix}
        \begin{matrix}
          \bm{J}_{n_1}(\beta_1,1) &                         &        &                         \\
                                  & \bm{J}_{n_2}(\beta_2,1) &        &                         \\
                                  &                         & \ddots &                         \\
                                  &                         &        & \bm{J}_{n_k}(\beta_k,1)
        \end{matrix}
         & \text{\Huge 0} \\
        \text{\Huge 0}
         &
        \begin{matrix}
          \bm{J}_{m_1}(0,1) &        &                   \\
                            & \ddots &                   \\
                            &        & \bm{J}_{m_l}(0,1)
        \end{matrix}
      \end{matrix}\right]_{n\times n},
  \end{equation*}
  where
  $\bm{J}_{n_i}(\beta_i,1)$ is a Jordan block of size $n_i$ with nonzero eigenvalue $\beta_i$ of $\bm{\Lambda B}$, $i=1,2,\cdots,k$, and
  $\bm{J}_{m_j}(0,1)$ is a Jordan block of size $m_i$ with zero eigenvalue of $\bm{\Lambda B}$, $j=1,2,\cdots,l$.

Moreover, Lemma.~\ref{lemma:GM_AM} states that
    the algebraic multiplicity of the zero eigenvalue of the matrix $\bm{\Lambda B}$
    is equal to its geometric multiplicity.
This finding means that the Jordan blocks
    w.r.t the zero eigenvalues of the matrix $\bm{\Lambda B}$
    are all equal to the matrix $(0)_{1\times1}$.
\end{proof}

\begin{theorem}
If $\bm{B}$ is a positive semidefinite matrix,
   then the iterative method \eqref{eq:equal_weight_ctrlpts} is convergent.
  \label{thrm:equal_weight__semidefinite_convergent}
\end{theorem}

\begin{proof}
A positive semidefinite matrix $\bm{B}$ can be decomposed as
  \begin{equation*}
    \bm{B}=\bm{V}\text{diag}(\lambda_1,\lambda_2,\cdots,\lambda_{n-n_0},\underset{n_0}{\underbrace{0,\cdots,0}})\bm{V}^T,
  \end{equation*}
   where $\bm{V}$ is an orthogonal matrix,
   and $\lambda_i$ and $i=1,2,\cdots,n-n_0$ are both the nonzero eigenvalues and nonzero singular values of $\bm{B}$.
Then, the Moore-Penrose (M-P) inverse of $\bm{B}$ has the corresponding decomposition as follows:
  \begin{equation*}
    \bm{B}^{+}=\bm{V}\text{diag}(\frac{1}{\lambda_1},\frac{1}{\lambda_2},\cdots,\frac{1}{\lambda_{n-n_0}},\underset{n_0}{\underbrace{0,\cdots,0}})\bm{V}^T.
  \end{equation*}
Thus,
  \begin{equation}
    \bm{B}^{+}\bm{B} = \bm{V}\text{diag}(\underset{n-n_0}{\underbrace{1,1,\cdots,1}},\underset{n_0}{\underbrace{0,\cdots,0}})\bm{V}^T.
    \label{eq:equal_weight_decomposition}
  \end{equation}

Based on Lemma.~\ref{lemma:Jordan},
   the Jordan canonical form of the matrix $\bm{\Lambda B}$ is $\bm{J}$~\eqref{eq:jordan_canonical}.
We obtain the following by using the Jordan decomposition: $\bm{\Lambda B}=\bm{T}^{-1}\bm{JT}$,
   where $\bm{T}$ is an invertible matrix.
Hence,
  \begin{equation*}
    \bm{I}-\bm{\Lambda B} =
    \bm{T}^{-1}
    \left[\begin{matrix}
        \begin{matrix}
          \bm{J}_{n_1}(1-\beta_1,-1) &                            &        &                            \\
                                     & \bm{J}_{n_2}(1-\beta_2,-1) &        &                            \\
                                     &                            & \ddots &                            \\
                                     &                            &        & \bm{J}_{n_k}(1-\beta_k,-1)
        \end{matrix}
         & \text{\Huge 0} \\
        \text{\Huge 0}
         &
        \begin{matrix}
          1 &        &   \\
            & \ddots &   \\
            &        & 1
        \end{matrix}
      \end{matrix}\right]
    \bm{T}.
  \end{equation*}
Based on Lemma.~\ref{lemma:eigenvalue_range},
    $-1<1-\beta_i<1$, $i=1,2,\cdots,k$.
Then, we obtain the following by combining Eq.~\eqref{eq:equal_weight_decomposition}:
  \begin{align}
    \lim_{h\to \infty}\left ( \bm{I}-\bm{\Lambda B} \right )^h
    \nonumber
     & = \bm{T}^{-1}\text{diag}( \underset{n-n_0}{\underbrace{0,0,\cdots,0}},\underset{n_0}{\underbrace{1,\cdots,1}})\bm{T}         \\
    \nonumber
     & = \bm{T}^{-1}(\bm{I}-\text{diag}(\underset{n-n_0}{\underbrace{1,1,\cdots,1}},\underset{n_0}{\underbrace{0,\cdots,0}}))\bm{T} \\
     & = \bm{I}-\left(\bm{VT}\right)^{-1}\bm{B}^{+}\bm{B}\left(\bm{VT}\right).
    \label{eq:iter_mtx_limit}
  \end{align}
On the other hand, the coefficient matrix $\bm{B}$ is singular;
    thus, the linear system in Eq.~\eqref{eq:equal_weight_classical} has its solution
    if and only if
  \begin{equation*}
    \left(1-\omega\right)\bm{BB}^{+}\bm{N}^T\bm{Q}=\left(1-\omega\right)\bm{N}^T\bm{Q}.
  \end{equation*}
Therefore, we obtain the following by subtracting $\left(1-\omega\right)\bm{B}^{+}\bm{N}^T\bm{Q}$ from both sides of the iterative scheme \eqref{eq:equal_weight_ctrlpts}:
  \begin{equation*}
    \begin{aligned}
      \bm{P}^{[k+1]}-\left(1-\omega\right)\bm{B}^{+}\bm{N}^T\bm{Q}
       & = \left(\bm{I}-\bm{\Lambda B}\right)\bm{P}^{[k]}
      +\left(1-\omega\right)\bm{\Lambda R}^T\bm{Q}
      - \left(1-\omega\right)\bm{B}^{+}\bm{N}^T\bm{Q}                                                                          \\
       & = \left(\bm{I}-\bm{\Lambda B}\right)\bm{P}^{[k]}
      + \left(1-\omega\right)\bm{\Lambda}\bm{BB}^{+}\bm{N}^T\bm{Q}
      - \left(1-\omega\right)\bm{B}^{+}\bm{N}^T\bm{Q}                                                                          \\
       & = \left(\bm{I}-\bm{\Lambda B}\right)\bm{P}^{[k]}
      - \left(1-\omega\right)\left(\bm{I}-\bm{\Lambda B}\right)\bm{B}^{+}\bm{N}^T\bm{Q}                                        \\
       & = \left(\bm{I}-\bm{\Lambda B}\right)\left( \bm{P}^{[k]}-\left(1-\omega\right)\bm{B}^{+}\bm{N}^T\bm{Q}\right)          \\
       & = \left(\bm{I}-\bm{\Lambda B}\right)^{k+1}\left( \bm{P}^{[0]}-\left(1-\omega\right)\bm{B}^{+}\bm{N}^T\bm{Q}\right).
    \end{aligned}
  \end{equation*}
From Eq.~\eqref{eq:iter_mtx_limit}, when $k\rightarrow\infty$,
   the above equation tends to
  \begin{equation*}
    \bm{P}^{[\infty]} - \left(1-\omega\right)\bm{B}^{+}\bm{N}^T\bm{Q}
    = \left(\bm{I}-\left(\bm{VT}\right)^{-1}\bm{B}^{+}\bm{B}\left(\bm{VT}\right)\right)
    \left( \bm{P}^{[0]}-\left(1-\omega\right)\bm{B}^{+}\bm{N}^T\bm{Q}\right).
  \end{equation*}
Consequently,
  \begin{equation*}
    \begin{aligned}
      \bm{P}^{[\infty]}
       & = \left(\bm{I}-\left(\bm{VT}\right)^{-1}\bm{B}^{+}\bm{B}\left(\bm{VT}\right)\right)
      \left( \bm{P}^{[0]}-\left(1-\omega\right)\bm{B}^{+}\bm{N}^T\bm{Q}\right)
      + \left(1-\omega\right)\bm{B}^{+}\bm{N}^T\bm{Q}                                                                 \\
       & = \left(1-\omega\right)\left(\bm{VT}\right)^{-1}\bm{B}^{+}\bm{B}\left(\bm{VT}\right)\bm{B}^{+}\bm{N}^T\bm{Q}
      + \left(\bm{I}-\left(\bm{VT}\right)^{-1}\bm{B}^{+}\bm{B}\left(\bm{VT}\right)\right)\bm{P}^{[0]},
    \end{aligned}
  \end{equation*}
  where $\bm{P}^{[0]}$ is the initial control point,
  which can be randomly selected.
When $\bm{P}^{[0]}=0$,
   the iterative method \eqref{eq:equal_weight_ctrlpts} converges to
  \begin{equation}
    \bm{P}^{[\infty]} = \left(1-\omega\right)\left(\bm{VT}\right)^{-1}\bm{B}^{+}\bm{B}\left(\bm{VT}\right)\bm{B}^{+}\bm{N}^T\bm{Q}.
    \label{eq:equal_weight_limit_semidefinite}
  \end{equation}
\end{proof}
\begin{remark}
If matrix $\bm{V}$ is the inverse matrix of $\bm{T}$,
   that is, $\bm{VT}=\bm{TV}=\bm{I}$, then Eq.~\eqref{eq:equal_weight_limit_semidefinite} becomes
  \begin{equation*}
    \bm{P}^{[\infty]} = \left(1-\omega\right)\bm{B}^{+}\bm{N}^T\bm{Q},
  \end{equation*}
   where $\bm{B}^{+}$ is the M-P pseudo-inverse of the matrix $\bm{B}$.
Then the iterative method \eqref{eq:equal_weight_ctrlpts} converges to the solution of the classical energy minimization equation \eqref{eq:equal_weight_classical}.
\end{remark}

\end{document}